%% file: main.tex
\documentclass[a4paper,UKenglish]{llncs}

\usepackage[utf8]{inputenc}
\usepackage[english]{babel}

\usepackage{amssymb,amsmath,amsfonts}

\usepackage{enumerate}
\usepackage{wrapfig}
\usepackage{xspace}
\usepackage{stmaryrd}
\usepackage{textcomp} %tilde

\usepackage{tikz,pgffor}
\usetikzlibrary{arrows}
\usetikzlibrary{shapes}
\usetikzlibrary{automata}
\usetikzlibrary{decorations.pathmorphing} % for snake lines
\usetikzlibrary{decorations.pathreplacing} % for snake lines
\usetikzlibrary{calc} % for manimulation of coordinates

\newcommand{\playercon}{{\mathbf{con}}}
\newcommand{\playerenv}{{\mathbf{env}}}
\newcommand{\playerconc}{\text{the player $\playercon$} }
\newcommand{\playerenvc}{\text{the player $\playerenv$} }

\newcommand{\Nset}{\mathbb{N}}
\newcommand{\Nseto}{\Nset_0}

\newcommand{\Rset}{\mathbb{R}}
\newcommand{\Rsetp}{\mathbb{R}_{>0}}
\newcommand{\Rsetpo}{\mathbb{R}_{\ge 0}}

\newcommand{\dist}{\mathcal{D}}

\renewcommand{\vec}[1]{\mathbf{#1}}
\newcommand{\de}[1]{\mathit{d#1}}  %

\newcommand{\sigmafield}{\mathcal{F}}
\newcommand{\imc}{\mathcal{C}} %M
\newcommand{\env}{\mathcal{E}} %M
\newcommand{\closed}{\imc|\env}                                                                                    %{\imc(\env)}

\newcommand{\states}{S}
\newcommand{\act}{\mathbb{A}\mathrm{ct}^\tau}
\newcommand{\actnotau}{\mathbb{A}\mathrm{ct}}

\newcommand{\itran}{\mathord{\hookrightarrow}}
\newcommand{\mtran}{\mathord{\rightsquigarrow}} %twoheadrightarrow

\newcommand{\acttran}[1]{{}\mathchoice%
    {\stackrel{#1}{\itran}}
     {\mathop {\smash\itran}\limits^{\vrule width 0pt height 0pt depth 4pt\smash{#1}}}
    {\stackrel{#1}{\itran}}
    {\stackrel{#1}{\itran}}
{}}
\newcommand{\acttranlong}[1]{{}\mathchoice%
    {\stackrel{#1}{\itran}}
     {\mathop {\smash\itran}\limits^{\vrule width 0pt height 0pt depth 0pt\smash{#1}}}
    {\stackrel{#1}{\itran}}
    {\stackrel{#1}{\itran}}
{}}
\newcommand{\probtran}[1]{{}\mathchoice%
    {\stackrel{#1}{\mtran}}
    {\mathop {\smash\mtran}\limits^{\vrule width 0pt height 0pt depth 4pt\smash{#1}}}
    {\stackrel{#1}{\mtran}}
    {\stackrel{#1}{\mtran}}
{}}
\newcommand{\probtranlong}[1]{{}\mathchoice%
    {\stackrel{#1}{\mtran}}
    {\mathop {\smash\mtran}\limits^{\vrule width 0pt height 0pt depth 0pt\smash{#1}}}
    {\stackrel{#1}{\mtran}}
    {\stackrel{#1}{\mtran}}
{}}

\newcommand{\tauc}{\acttranlong{\playercon}}
\newcommand{\taue}{\acttranlong{\playerenv}}

\newcommand{\init}{{s_0}}

\newcommand{\rates}{\mathrm{R}}
\newcommand{\prob}{\mathbf{P}}

\newcommand{\reach}{\Diamond}
\newcommand{\goal}{G}
\newcommand{\setofruns}{\mathbb{R}\mathrm{uns}}
\newcommand{\probm}{\mathcal{P}}

\newcommand{\scheduler}{\mathfrak{S}}
\newcommand{\histories}{\mathbb{H}\mathrm{istories}} %in CE games
\newcommand{\history}{\mathfrak{h}} %in CE games
\newcommand{\paths}{\mathbb{P}\mathrm{aths}} %in IMC
\newcommand{\mypath}{\mathfrak{p}} %in IMC

\newcommand{\IMC}{\mathrm{IMC}}
\newcommand{\ENV}{\mathrm{ENV}}

\newcommand{\comp}[1]{\parallel_{#1}}
\newcommand{\hide}[1]{\diagdown{#1}}

\newcommand{\eps}{{\varepsilon}}

\newcommand{\last}[1]{#1\mathord{\downarrow}}
\newcommand{\wc}{\mathrm{WC}}

\newcommand{\game}{\mathcal{G}}

\newcommand{\gamed}{\Delta}
\newcommand{\igame}{{\game'}}
\newcommand{\dgame}{{\game''}}
\newcommand{\dPi}{\bar{\Pi}}
\newcommand{\almostzero}{0^\leftarrow}
\newcommand{\almost}[1]{^{\rightarrow}{#1}}
\newcommand{\pattern}{\sim}

\newcommand{\pathsfield}{\mathfrak{P}}

\newcommand{\QED}{\qed}

\newcommand{\commit}{\mathit{Commit}}

\newcommand{\destutter}{destutter}
\newcommand{\id}{proj_1}

\newcommand{\exponential}{\mathrm{Exp}}

\newcommand{\kernel}{P}

\newcommand{\step}{\kappa}
\newcommand{\supp}{\mathrm{supp}}

\newcommand{\maxconst}{b}

\newcommand{\appref}[1]{#1}
\newcommand{\apprefwhole}{the appendix}

\newcommand{\out}[1]{} %{#1}

\newcommand{\theoremlike}[2]{\par\medskip\penalty-250\refstepcounter{theorem}{{\bfseries\noindent#2
\ref{#1}.}}}

\newcommand{\thmhelperpre}[2]{\theoremlike{#1}{#2}}
\newcommand{\thmhelperpost}{\par\medskip}

\newenvironment{reftheorem}[1]{\thmhelperpre{#1}{Theorem}}{\thmhelperpost}

\newcommand{\semantics}[1]{\llbracket #1 \rrbracket}
\newcommand{\spec}{\mathcal S}
\newcommand{\may}[1]{\stackrel{#1}{\dashrightarrow}}
\newcommand{\must}[1]{\stackrel{#1}{\longrightarrow}}
\newcommand{\timetran}[1]{\stackrel{#1}{\rightsquigarrow}}
\newcommand{\timetranlong}[1]{{}\mathchoice%
    {\stackrel{#1}{\rightsquigarrow}}
    {\mathop {\smash\rightsquigarrow}\limits^{\vrule width 0pt height 0pt depth 0pt\smash{#1}}}
    {\stackrel{#1}{\rightsquigarrow}}
    {\stackrel{#1}{\rightsquigarrow}}
{}}

\newcommand{\omay}{\mathord{\may{}}}
\newcommand{\omust}{\mathord{\must{}}}
\newcommand{\otimetran}{\mathord{\rightsquigarrow}}
\newcommand{\rel}{\mathcal{R}}
\newcommand{\product}{\imc\times\spec}
\newcommand{\closedspec}{(\product)|\env}                                                                            
\newcommand{\closedspeci}[1]{(\product_{#1})|\env}                                                           %{\imc\times\spec_{#1}\,(\env)}
\newcommand{\nowtran}{\mathsf{\mathbf{Now}}}
\newcommand{\committran}{\mathsf{\mathbf{Change}}}
\newcommand{\commitstate}{\mathit{commit}}
\newcommand{\nowstate}{\mathit{now?}}
\newcommand{\arena}{G_i}%{S_\game}

\newcommand{\current}{\textrm{last}}%{S_\game}

\newcommand{\leftend}{\vdash}
\newcommand{\inside}{+}
\newcommand{\rightend}{\dashv}

\newcommand{\myspace}{\vspace*{-1em}}
\newcommand{\myspaceb}{\vspace*{-0.5em}}

\title{Compositional Verification and Optimization \\ of Interactive Markov Chains}

\author{ 
Holger Hermanns\inst{1} \and Jan Kr\v{c}\'{a}l\inst{2} \and Jan K\v{r}et\'{i}nsk\'{y}\inst{2,3}
}

\institute{
Saarland University -- Computer Science, Saarbr\"ucken, Germany
    \texttt{hermanns@cs.uni-saarland.de}
\and 
Faculty of Informatics, Masaryk University, Czech Republic
    \texttt{\{krcal,jan.kretinsky\}@fi.muni.cz} 
\and
    Institut f\"ur Informatik, Technical University Munich, Germany
}

\begin{document}

\pagestyle{plain}

\maketitle

\setcounter{footnote}{0}
\renewcommand{\thefootnote}{(\arabic{footnote})}

\myspace

%%%%%% theorems for appendix

\newcommand{\theoremerrorbound}{
For every IMC $\imc$ and MCA $\spec$, $v_\game$ is approximated by $v_\gamed$: \myspaceb

% $v_\gamed \, \leq \, v_\game\, \leq \, v_\gamed + 2\step\lambda^2 T$
$$ \textstyle |v_\game - v_\gamed| \;\;\; \leq \;\;\;  10 \step (\maxconst T)^2 \ln\frac{1}{\step}.$$
A strategy $\sigma^\ast$ optimal in $\gamed$ defines a strategy $(10 \step (\maxconst T)^2 \ln\frac{1}{\step})$-optimal in $\game$. %$v_\spec(\imc)$. 
Further, $v_\gamed$ and $\sigma^\ast$ can be computed in time polynomial in $|\gamed|$, hence in time $2^{\mathcal{O}(|\game|)}$.
}

%%%%%%

\begin{abstract}
Interactive Markov chains (IMC) are compositional behavioural models
extending labelled transition systems and continuous-time Markov chains. 
We provide a framework and algorithms for compositional verification and optimization of IMC with respect to time-bounded properties.
Firstly, we give a specification formalism for IMC. 
Secondly, given a time-bounded property, an IMC component  and the assumption that its unknown environment satisfies a given specification, we synthesize a scheduler for the component optimizing the probability that the property is satisfied in any such environment.
\end{abstract}

% \begin{bibunit}[alpha]

\input{intro-new}

\input{defs}
\input{compverif}
\input{specs}

\input{product}

\input{cegames}
\input{discretization}

\input{conclusion}
%\input{other}

% \putbib[str-short,li,concur]
% \end{bibunit}
 
\bibliographystyle{alpha}
\bibliography{str-short,li,concur}

%%%%%%%%%%%%%%%%%%%%%%%%%%%%% PLEASE bibtex bu1.aux AND bu2.aux SEPARATELY

\newpage
% \onecolumn
\appendix

% \begin{bibunit}[alpha]

\input{app-results}
% \input{app-disc}
\input{app-disc2}

%\input{old_and_garbage}

% \putbib[str-short,li,concur]
% \end{bibunit}

\end{document}

%% file: intro-new.tex
% $Id: intro.tex 341 2012-10-09 13:11:50Z xkretins $

\myspace

\section{Introduction}

\myspaceb

The ever increasing complexity and size of systems together with software reuse strategies  naturally enforce the need for component based system development. For the same reasons, checking reliability and optimizing performance of such systems needs to be done in a \emph{compositional} way. The task is to get useful guarantees on the behaviour of a component of a larger system. The key idea is to incorporate assumptions on the rest of the system into the verification process. This \emph{assume-guarantee reasoning} is arguably a successful divide-and-conquer technique in many contexts~\cite{DBLP:journals/tse/MisraC81,DBLP:conf/lics/AlurH96,\out{DBLP:conf/concur/AlurH97,}DBLP:conf/hybrid/HenzingerMP01}. 

In this work, we consider a continuous-time stochastic model called \emph{interactive Markov chains} (IMC). First, we give a language for expressing assumptions about IMC. Second, given an IMC, an assumption on its environment and a property of interest, we synthesize a controller of the IMC that optimizes the guarantee, and we compute this optimal guarantee, too. 

\smallskip

\textbf{Interactive Markov chains} are behavioural models of
probabilistic systems running in continuous real time appropriate for
the component-based approach~\cite{fmco09\out{,Her02a}\out{,Bra06}}.
IMC have a well-understood compositional theory rooted in process
algebra\out{~\cite{BePS01}}, and are in use as semantic backbones for
dynamic fault trees,\out{~\cite{Bou10},} architectural description
languages\out{~\cite{arcade,esa10}}, generalized stochastic Petri
nets\out{~\cite{HKNZ10}} and Statemate\out{~\cite{Boe09}} extensions,
see~\cite{fmco09} for a survey. IMC are applied in a large spectrum of
practical applications, ranging from \out{networked hardware on
chips~\cite{Cos09} to }water treatment facilities~\cite{water} to %and
ultra-modern satellite designs~\cite{jpk-icse12}.

\begin{wrapfigure}{r}{0.25\textwidth}
\vspace*{-1.3em}
\begin{tikzpicture}[x=1.7cm,->,>=stealth',
% state style
state/.style={shape=circle,draw,font=\scriptsize,inner sep=1mm,outer sep=0.8mm, minimum size=0.75cm},
branch/.style={shape=circle,draw,inner sep=0.5mm,thick,outer sep=0.8mm},
multiline/.style={text width=12mm,text centered},
% target state
target/.style={double distance=0.2mm},
% probability style
transition/.style={},
rate/.style={above,color=black,font=\scriptsize},
tau/.style={rate,color=blue,font=\large},
prob/.style={rate,color=black},
]

%%%% original layout
% \node[state] (X) {$\mathit{init}$};
% \node[state] (Y) at (1,1) {$s$};
% \node[state] (Y') at (1,0) {$t$};
% \node[state] (Y'') at (1,-1) {$u$};
% \node[state] (Z) at (2,0) {$\mathit{goal}$};
% 
% \path ($(X)+(-0.5,0)$) 		edge [transition]		(X);
% \path[->] (X) edge node[above]{$a$} (Y); 
% \path[->] (X) edge node[below]{$\tau$} (Y');
% \path[->] (X) edge node[below]{$\tau$} (Y'');
% \path[->] (Y) edge node[above]{$2$} (Z);
% \path[->] (Y') edge node[below]{$1$} (Z);
% \path[->] (Y'') edge node[below]{$3$} (Z);

\node[state] (X) {$\mathit{init}$};
\node[state] (Y) at (1.2,0) {$s$};
\node[state] (Y') at (0.6,-0.75) {$u$};
\node[state] (Y'') at (0,-1.5) {$v$};
\node[state] (Z) at (1.2,-1.5) {$\mathit{goal}$};

\path ($(X)+(-0.5,0)$) 		edge [transition]		(X);
\path[->] (X) edge node[above]{$a$} (Y); 
\path[->] (X) edge node[below left]{$\tau$} (Y');
\path[->] (X) edge node[left]{$\tau$} (Y'');
\path[->] (Y) edge node[left]{$2$} (Z);
\path[->] (Y') edge node[below left]{$1$} (Z);
\path[->] (Y'') edge node[below]{$3$} (Z);
\end{tikzpicture}
\vspace*{-4em}
\end{wrapfigure}
IMC arise from classical labelled transition systems by incorporating the possibility to change state according to a random \emph{delay} governed by a negative exponential distribution with a given rate, see transitions labelled 1, 2 and 3 in the figure. Apart from delay expirations, state transitions may be triggered by the execution of \emph{internal} ($\tau$) actions or \emph{external} (synchronization) actions. Internal actions are assumed to happen instantaneously and therefore take precedence over delay transitions. External actions are the process algebraic means for interaction with other components, see $a$ in the figure.
By dropping the delay transitions, labelled transition systems are regained in their entirety. Dropping
action-labelled transitions instead yields continuous-time Markov chains -- one of the most used performance and reliability models.

The fundamental problem in the analysis of IMC
%akin to problems of paramount importance in many general real-time systems, 
is that of \emph{time-bounded reachability}. It is the problem to approximate the probability that a given set of states is reached within a given deadline. We illustrate the compositional setting of this problem in the following examples. 

\smallskip

\textbf{Examples.} 
% We illustrate the compositional verification task on an example. 
In the first example, consider the IMC $\imc$ from above and an unknown environment $\env$ with no assumptions. Either $\env$ is initially not ready  to synchronize on the external action $a$ and thus one of the internal actions is taken, or $\env$ is willing to synchronize on $a$ at the beginning. In the latter case, whether $\tau$ or $a$ happens is resolved non-deterministically. Since this is out of control  of $\imc$, we must assume the worst case and let the environment decide which of the two options will happen. For more details on this design choice, see~\cite{DBLP:conf/fsttcs/BrazdilHKKR12}. If there is synchronization on $a$, the probability to reach $\mathit{goal}$ within time $t=1.5$ is $1-e^{-2t}\approx0.95$. Otherwise, $\imc$ is given the choice to move to $u$ or $v$. Naturally, $v$ is the choice maximizing the chance to get to $\mathit{goal}$ on time as it has a higher rate associated. In this case the probability amounts to  $1-e^{-3t}\approx0.99$, while if $u$ were chosen, it would be only $0.78$. Altogether, the guaranteed probability is $95\%$ and the strategy of $\imc$ is to choose $v$ in $\mathit{init}$.

\begin{wrapfigure}{r}{0.5\textwidth}
\vspace*{-2.3em}
\begin{tikzpicture}[x=1.7cm,->,>=stealth',
% state style
state/.style={shape=circle,draw,font=\scriptsize,inner sep=1mm,outer sep=0.8mm, minimum size=0.75cm},
branch/.style={shape=circle,draw,inner sep=0.5mm,thick,outer sep=0.8mm},
multiline/.style={text width=12mm,text centered},
% target state
target/.style={double distance=0.2mm},
% probability style
transition/.style={},
rate/.style={above,color=black,font=\scriptsize},
tau/.style={rate,color=blue,font=\large},
prob/.style={rate,color=black},
]

\node[state] (X) {$\mathit{init}$};
\node[state] (Y) at (1,0) {$\mathit{proc}$};
\node[state] (Y') at (2,0) {$\mathit{ret}$};
\node[state] (Z) at (3,0) {$\mathit{goal}$};

\path ($(X)+(-0.5,0)$) 		edge [transition]		(X);
\path[->] (X) edge node[above]{$\mathsf{req}$} (Y); 
\path[->] (Y) edge node[above]{$\tau$} (Y');
\path[->] (Y') edge node[above]{$\mathsf{resp}$} (Z);
\draw [->,rounded corners] (X) |- ($(X)+(1,-0.8)$) node[above,pos=1.3]{$\tau$} -| (Z);
\end{tikzpicture}
\vspace*{-2.5em}
\end{wrapfigure}
The example depicted on the right illustrates the necessity of assumptions on the environment:
%Without any assumptions 
As it is, the environment can drive the component %depicted below 
to state $\mathit{ret}$  and let it get stuck there by not  synchronising on $\mathsf{resp}$ ever. Hence no better guarantee than $0$ can be derived. 
However, this changes if we know some specifics about the
behaviour of the environment: Let us assume that we know that once synchronization on $\mathsf{req}$ occurs, the environment must be ready to synchronise on $\mathsf{resp}$ within some random time according to, say, an exponential distribution with rate $2$. Under this assumption, we are able to  derive a guarantee of $95\%$, just as in the previous example.

Observe the form of the time constraint we imposed in the last
example: ``within a random time distributed according to Exp(2)'' or
symbolically $\Diamond_{\leq Exp(2)}\varphi$. We call this a
\emph{continuous time constraint}. If a part of the environment is
e.g.~a model of a communication network, it is clear we cannot impose
hard bounds (discrete time constraints) such as ``within 1.5'' as in
e.g.~a formula of MTL $\Diamond_{\leq 1.5}\varphi$. Folklore tells us
that messages might get delayed for longer than that. Yet we want to
express high assurance that they arrive on time. In this case one
might use e.g.~a formula of CSL $Pr_{\geq0.95}(\Diamond_{\leq
  1.5}\varphi)$. 
However, consider now a system with two transitions labelled with $\mathsf{resp}$ in a row. Then this CSL formula yields only a zero guarantee. By splitting the time $1.5$ in halves, the respective $Pr_{\geq0.77}(\Diamond_{\leq 0.75}\varphi)$ yields only the guarantee $0.77^2=0.60$. The actual guarantee $0.80$ is given by the convolution of the two exponential distributions and as such can be exactly obtained from our continuous time constraint $\Diamond_{\leq Exp(2)}\varphi$.
% However, when considering sequential composition of
% delay transitions as in $\mathit{init}\tran{2}s\tran{3}\mathit{goal}$
% it is not clear which combination of time bounds and probabilities we
% should ask for. 
% Instead we will arrive at more precise results by
% supporting all conceivable combinations, which corresponds to the information about
% the whole distribution.
% Further, to cover operators like $Pr_{\geq0.95}$, we can consider subdistributions.

\smallskip

\textbf{Our contribution} is the following:
\begin{enumerate}
 \item We introduce a specification formalism to express assumptions on continuous-time stochastic systems. The novel feature of the formalism are the continuous time constraints, which are vital for getting guarantees with respect to time-bounded reachability in IMC.
%  \item We show how to synthesize $\varepsilon$-optimal schedulers for IMC in an unknown environment and approximate the resulting guarantee.% in polynomial time.
  \item We incorporate the assume-guarantee reasoning to the IMC framework. We show how to synthesize $\epsilon$-optimal schedulers for IMC in an \emph{unknown environment satisfying a given specification} and approximate the respective guarantee.% in polynomial time.
\end{enumerate}

In our recent work~\cite{DBLP:conf/fsttcs/BrazdilHKKR12} we considered a very restricted setting of the second point. Firstly, we considered no assumptions on the environment as the environment of a component might be entirely unknown in many scenarios. Secondly, we were restricted to IMC that never enable internal and external transitions at the same state. This was also a severe limitation as this property is not preserved during the IMC composition process and restricts the expressivity significantly. Both examples above violate this assumption. In this paper, we lift the assumption. 

Each of the two extensions shifts the solution methods from complete information stochastic games to (one-sided) \emph{partial observation} stochastic games, where we need to solve the quantitative reachability problem. While this is undecidable in general, we reduce our problem to a game played on an \emph{acyclic graph} and show how to solve our problem in exponential time. (Note that even the qualitative reachability in the acyclic case is PSPACE-hard~\cite{DBLP:conf/lpar/ChatterjeeD10}.)

% Although partial observation these are lot more complicated, we show that this class can be solved in polynomial time as opposed to the general PSPACE-complete case. The class of games we introduce is interesting in its own right since it allows for modelling continuous-time components in an environment which can delay communication both in a probabilistic and non-deterministic way.
% 
% \todo{co time-bounded ``properties''? chceme neco obecnejsiho jako csl...?}

\smallskip

\textbf{Related work.}
The \emph{synthesis} problem is often stated as a game where the first player controls a component and the second player simulates an environment~\cite{ramadge1989}. Model checking of \emph{open} systems, i.e.~operating in an unknown environment, has been proposed in~\cite{KV_CAV_96}. There is a body of work on \emph{assume-guarantee} reasoning for parallel composition of \emph{real-time}  systems~\cite{DBLP:conf/concur/TasiranAKB96,\out{DBLP:conf/concur/AlurH97,}DBLP:conf/hybrid/HenzingerMP01}. Lately, games with \emph{stochastic continuous-time} have gained attention, for a very general class see~\cite{DBLP:conf/icalp/BouyerF09\out{,DBLP:journals/acta/RabeS11}\out{,Bra09}\out{,DBLP:conf/concur/BrazdilKKKR10}}. While the second player models possible schedulers of the environment, the structure of the environment is fixed there and the verification is thus not compositional. The same holds for~\cite{DBLP:conf/qest/Sproston11,DBLP:conf/qest/HahnNPWZ11}, where time is under the control of the components.

A compositional framework requires means for specification of systems. A specification can be also viewed as an \emph{abstraction} of a set of systems. Three valued abstractions stemming from~\cite{DBLP:conf/lics/LarsenT88} have also been applied to the timed setting, namely in~\cite{DBLP:conf/cav/KatoenKLW07} to continuous-time Markov chains (IMC with no non-determinism), or in~\cite{DBLP:conf/formats/KatoenKN09} to IMC. Nevertheless, these abstractions do not allow for constraints on time distributions. Instead they would employ abstractions on transition probabilities. Further, a compositional framework with timed specifications is presented in~\cite{DBLP:journals/sttt/DavidLLMNRSW12}. This framework explicitly allows for time constraints. However, since the systems under consideration have non-deterministic flow of time (not stochastic), the natural choice was to only allow for discrete (not continuous) time constraints.

Although IMC support compositional design very well, analysis techniques for IMC proposed so far (e.g.~\cite{\out{cadp,}mrmc,\out{conf/mmb/HermannsJ08,}DBLP:conf/formats/KatoenKN09,Neu10,imca} are not compositional. They are all bound to the assumption that the analysed IMC is a \emph{closed} system, i.e.\ it does not depend on interaction with the environment (all actions are internal). Some preliminary steps to develop a framework for synthesis of controllers based on models of hardware and control requirements have been taken in~\cite{DBLP:conf/acsd/Markovski11}. The first attempt at compositionality is our very recent work~\cite{DBLP:conf/fsttcs/BrazdilHKKR12} discussed above. 

Algorithms for the \emph{time-bounded reachability} problem for closed IMC %has been studied and  to compute it 
have been given in~\cite{\out{conf/mmb/HermannsJ08,}Neu10,buchholz-schulz11,HH13} and compositional
abstraction techniques to compute it are developed in~\cite{DBLP:conf/formats/KatoenKN09}. In the closed interpretation,
IMC have some similarities with continuous-time Markov decision
processes. For this formalism, %CTMDP.  
algorithms for time-bounded reachability %and corresponding games 
are developed in~\cite{Bai05b,\out{Bra09,}\out{DBLP:journals/acta/RabeS11,}buchholz-schulz11}.

%%% Local Variables:
%%% mode: latex
%%% TeX-master: "main"
%%% End:

%% file: defs.tex
% $Id: defs.tex 341 2012-10-09 13:11:50Z xkretins $

\myspace

\section{Interactive Markov Chains}
\myspaceb

In this section, we introduce the formalism of interactive Markov chains
together with the standard way to compose them. 
We denote by $\Nset$, $\Rsetp$, and $\Rsetpo$ the sets of
positive integers, % natural numbers, natural numbers with zero, 
positive real numbers and non-negative real
numbers, respectively. Further, let $\dist(S)$ denote the set of probability distributions over the set $S$.

\begin{definition}[IMC]
  An \emph{interactive Markov chain (IMC)} is a %quin
quintuple $\imc=
  (\states,\act,\mathord{\itran},\mathord{\mtran},\init)$ where
$\states$ is a finite set of \emph{states}, $\act$ is a finite set of
  \emph{actions} containing a designated \emph{internal action}
  $\tau$, $\init\in\states$ is an \emph{initial} state, 
  \begin{itemize}
  \item $\itran\subseteq\states\times\act \times \states$ is an
    \emph{interactive transition} relation, and
  \item $\mtran\subseteq\states\times\Rsetp\times\states$ is a
    \emph{Markovian transition} relation.  %\item $\rate\in\Rsetpo$ is a \emph{rate},
  \end{itemize}
\end{definition}

Elements of $\actnotau:=\act\smallsetminus\{\tau\}$ are called \emph{external
  actions}.  We write $s\acttran{a}t$ whenever $(s,a,t)\in\itran$, and
$s\probtran{\lambda}t$ whenever $(s,\lambda,t) \in \mtran$ where $\lambda$
is called a \emph{rate} of the transition. We say that an external action $a$, or internal $\tau$, or Markovian transition
is \emph{available} in $s$, if $s\acttran{a}t$, $s\acttran{\tau}t$ or $s\probtran{\lambda}t$ for some $t$ (and $\lambda$), respectively. %A state with available $\tau$ is called \emph{immediate}.

IMC are well suited for compositional modelling, where systems are
built out of smaller ones using standard composition operators. 
\emph{Parallel composition} $\parallel_A$ over  a \emph{synchronization alphabet} $A$ produces a product of two IMC with transitions given by the rules 
\vspace*{-0.5em}
\begin{description}
  \item[\textbf{(PC1)}] $(s_1,s_2)\acttran{a}(s'_1,s'_2)$ for each $s_1\acttran{a}s'_1$ and $s_2\acttran{a}s'_2$ and $a\in A$,
  \item[\textbf{(PC2, PC3)}] $(s_1,s_2)\acttran{a}(s'_1,s_2)$ for each $s_1\acttran{a}s'_1$ and $a\not\in A$, and symetrically,
  \item[\textbf{(PC4, PC5)}] $(s_1,s_2)\probtran{\lambda}(s'_1,s_2)$ for each $s_1\probtran{\lambda} s'_1$, and symmetrically.
\end{description}
\vspace*{-0.5em}
Further, \emph{hiding} $\hide{A}$ an alphabet $A$, yields a system, where each $s\acttran{a}s'$ with $a\notin A$ is left as it is, and each $s\acttran{a}s'$ with $a\in A$ is replaced by internal $s\acttran{\tau}s'$. 

Hiding $\hide{\actnotau}$ thus yields a 
\emph{closed} IMC, where external actions do not
appear as transition labels (i.e.\ $\itran\subseteq\states\times\{\tau\}
\times \states$). 
A closed IMC 
(under a scheduler $\sigma$, see below) moves from state to state and thus produces a {\em run} which is an
infinite sequence of the form $s_0\,t_1\,s_1\,t_2\,s_2\cdots$ where $s_n$ is the
$n$-th visited state and $t_n$ is the time of arrival to $s_n$.
After $n$ steps, the scheduler resolves the non-determinism among internal $\tau$ transitions based on the {\em path}
$\mypath=s_0\, t_1 \cdots t_n \, s_n$.

\begin{definition}[Scheduler]
A \emph{scheduler} of an IMC $\imc= (\states,\act,\mathord{\itran},\mathord{\mtran},\init)$ is a measurable function $\sigma:(\states\times\Rsetpo)^*\times\states\to\dist(\states)$ such that
for each path $\mypath = s_0\,t_1\,s_1 \cdots t_n\, s_n$ with $s_n$ having $\tau$ available,
$\sigma(\mypath)(s)>0$ implies $s_n\acttran{\tau}s$.  The set of all
schedulers for $\imc$ is denoted by $\scheduler(\imc)$.
\end{definition}
The decision of the scheduler $\sigma(\mypath)$ determines $t_{n+1}$ and
$s_{n+1}$ as follows. If $s_n$ has available $\tau$, then the run proceeds
immediately, i.e.\ at time $t_{n+1}:=t_n$, to a state
$s_{n+1}$ randomly chosen according to the distribution $\sigma(\mypath)$.
Otherwise, only Markovian transitions are available in $s_n$. In such a case, after waiting for a random time $t$ chosen according to the exponential distribution with the rate $\rates(s_n)= \sum_{s_n\probtran{\lambda}s'} \lambda$, the run moves at time $t_{n+1}:=t_n+t$ to a randomly chosen next state $s_{n+1}$ with probability $\lambda/r$ where $s_n\probtran{\lambda}s_{n+1}$.
This defines a probability space 
% $\mathcal P_{\imc,\sigma}=
$(\setofruns,\sigmafield,\probm_\imc^\sigma)$ over the runs in the standard way~\cite{Neu10}.

%%% Local Variables:
%%% mode: latex
%%% TeX-master: "main"
%%% End:

%% file: compverif.tex
% $Id: compverif.tex 338 2012-07-16 09:08:22Z xrehak $
\myspace

\section{Time-Bounded Reachability}
\myspaceb

In this section, we introduce the studied problems.
One of the fundamental problems in verification and performance analysis of
continu\-ous-time stochastic systems is time-bounded reachability. Given
a \emph{closed} IMC $\imc$, a set of goal
states $\goal\subseteq\states$ and a~time bound $T\in\Rsetpo$, the
\emph{value of time-bounded reachability} is defined as
${\sup_{\sigma\in\scheduler(\imc)}\probm^{\sigma}_\imc\big[\reach^{\leq T}\goal\big]}$
where $\probm^{\sigma}_\imc\big[\reach^{\leq T}\goal\big]$ denotes the probability
that a run of $\imc$ under the scheduler $\sigma$ visits a state of $G$
before time $T$. We have seen an example in the introduction. A standard assumption over all analysis techniques published for IMC~\cite{\out{cadp,}mrmc,\out{conf/mmb/HermannsJ08,}DBLP:conf/formats/KatoenKN09,Neu10,imca} is that each cycle contains a Markovian transition.
It implies that the probability of taking infinitely many transitions in
finite time, i.e. of Zeno behaviour, is zero.
One can $\varepsilon$-approximate the value and compute the respective scheduler in time $\mathcal O(\lambda^2T^2/\varepsilon)$~\cite{Neu10} recently improved to $\mathcal O(\sqrt{\lambda^3T^3/\varepsilon})$~\cite{HH13}.

For an \emph{open} IMC to be put in parallel with an unknown environment, the optimal scheduler is computed so that it optimizes the guarantee against all possible environments. Formally, for an IMC
$\imc=(C,\act,\mathord{\itran},\mathord{\mtran},c_0)$ and an environment IMC $\env$ with the same action alphabet $\act$, we introduce a composition
$\closed = (\imc\comp{\actnotau}\env)\hide{\actnotau }$ %\quad 
where all open actions are hidden, yielding a closed system.
In order to compute guarantees on $\closed$ provided we use a scheduler $\sigma$ in $\imc$, 
we consider schedulers $\pi$ of $\closed$ that \emph{respect $\sigma$} on the internal actions of $\imc$, written $\pi\in\scheduler_\sigma(\closed)$; the formal definition is below. The \emph{value of compositional time-bounded reachability} is then defined in~\cite{DBLP:conf/fsttcs/BrazdilHKKR12} as
% \vspace*{-1em}

$$\sup_{\sigma\in\scheduler(\imc)} \inf_{\substack{\env\in \ENV\\
      \pi\in\scheduler_\sigma(\closed)}}
  \probm^{\pi}_{\closed}\big[\reach^{\leq T}\goal\big]$$
\vspace*{-1em}

\noindent where $\ENV$ denotes the set of all IMC with the action alphabet $\act$ and
$\reach^{\leq T}\goal$ is the set of runs that reach $G$ in the first component before $T$. 
Now $\pi$ \emph{respects} $\sigma$ on internal actions of $\imc$ if 
for every path $\mypath=(c_0,e_0)\,t_1\cdots t_{n}(c_n,e_n)$ of $\closed$ there is $p\in[0,1]$ such that for each internal transition $c_n\acttran{\tau}c$ of $\imc$, we have $\pi(\mypath)(c,e_n)=p\cdot\sigma(\mypath_\imc)(c)$. Here $\mypath_\imc$ is the projection of $\mypath$ where %to states of $\imc$ only, since 
$\sigma$ can only see the path of moves in $\imc$ and not in which states $\env$ is. Formally, we define \emph{observation} of a path $\mypath=(c_0,e_0)\,t_1\cdots t_{n}(c_n,e_n)$ as $\mypath_\imc=c_0 t_1\cdots t_{n} c_n$ where each maximal consecutive sequence $t_i\,c_i\cdots t_j\,c_j$ with $c_k=c_i$ for all $i\leq k\leq j$ is rewritten to $t_i\,c_i$. This way, $\sigma$ ignores precisely the internal steps of $\env$.

%%% Local Variables:
%%% mode: latex
%%% TeX-master: "main"
%%% End:

%% file: specs.tex
% \myspace
\myspaceb

\subsection{Specifications of environments}
% \myspaceb

In the second example in the introduction, %we have seen that 
without any assumptions on the environment
only zero guarantees could be derived. The component was thus indistinguishable from an entirely useless one. In order to get a better guarantee, we introduce a formalism to
specify assumptions on the behaviour of environments. \myspace

\begin{example}
In the mentioned example, if we knew that after an occurrence of $\mathsf{req}$ the environment is ready to synchronize on $\mathsf{resp}$ in time distributed according to $\mathrm{Exp}(3)$ or faster, we would be able to  derive a guarantee of $0.26$. We will depict this assumption as shown below.
\end{example}
\vspace*{-0.5em}

\begin{wrapfigure}{r}{0.38\textwidth}
\vspace*{-3.5em}
\begin{tikzpicture}[->,>=stealth',
% state style
state/.style={shape=circle,draw,inner sep=1mm,outer sep=0.8mm, minimum size=.5cm,},
branch/.style={shape=circle,draw,inner sep=0.5mm,thick,outer sep=0.8mm},
multiline/.style={text width=13mm,text centered},
% target state
target/.style={double distance=0.2mm},
% probability style
transition/.style={rounded corners},
rate/.style={auto,color=black},
external/.style={rate,color=red},
tau/.style={rate,color=blue},
prob/.style={rate,color=black},
]

%%%%%%%%%%%%%%%%%%%%%%%%%%%%%%%%%%%%%%%%%%%%%%%%%%%%%
% STATES

\node (s) [state] {};
\node (t) at (1.8,0) [state] {};
\node (u) at (4,0) [state] {};
%%%%%%%%%%%%%%%%%%%%%%%%%%%%%%%%%%%%%%%%%%%%%%%%%%%%%
% TRANSITIONS

\path ($(s)+(-0.6,0)$) 		edge 		(s);
\path (s) edge [loop above,densely dashed] node[above] {$\mathsf{resp}$} (s);
\path (u) edge [loop above,densely dashed] node[above] {$\mathsf{req}$} (u);
\path (t) edge [loop above,densely dashed] node[above] {$\mathsf{req}$} (t);
\path (s) edge [densely dashed] node[above] {$\mathsf{req}$} (t);
% \path (u) edge [bend left] node[below] {$\mathsf{resp}$} (s);

\draw [->,rounded corners] (u) |- ($(u)+(-0.8,-0.8)$) -- node[below=-1,pos=0.8]{$\mathsf{resp}$} ($(s)+(0.8,-0.8)$)  -| (s);
% \draw [->,rounded corners] (u) -- ($(u)+(-0.6,-0.6)$) -- node[above=-1,pos=0.2]{$\top$} ($(s)+(0.6,-0.6)$)  -- (s);
\path (u) edge [bend left=30] node[above=-1,pos=0.35] {$\top$} (s);

\path (t) edge [] node[above] {$\leq\mathrm{Exp}(3)$} (u);
\end{tikzpicture}
\vspace*{-3.8em}
\end{wrapfigure}
The dashed arrows denote \emph{may} transitions, which may or may not be available, whereas the full arrows denote \emph{must} transitions, which the environment is ready to synchronize on. 
Full arrows are further used for time transitions. % could be any distribution. \out{\todo{phase-type commented out}}
%but we restrict to phase-type distributions.
% % (and to trivially true, denoted by $T$). 
%The reason for this is twofold. Firstly, they are dense on the space of continuous functions, hence  we can approximate any distributions with a phase-type distribution. Secondly, they are composed of exponential distributions, hence easy to work with.

Although such a system resembles a timed automaton, there are several fundamental differences. Firstly, the time constraints are given by probability distributions instead of constants. 
Secondly, there is only one clock that, moreover, gets reset
whenever the state is \emph{changed}. 
Thirdly, we allow modalities of may and must transitions. 
Further, as usual with timed or stochastic specifications, we require determinism.
\vspace*{-0.7em}

\begin{definition}[MCA syntax]
A \emph{continuous time constraint} is either $\top$ or of the form $\bowtie d$ with  $\mathord{\bowtie}\in\{\leq,\geq\}$ and $d$ a continuous
%\todo{removed ``phase-type'' to justify approximation sequences more, do we need it now?} %phase-type 
distribution.%
%\footnote{Our framework can easily be extended to handle conjunctions of constraints, for details, see~\appref{Appendix~\ref{app:conjunctions}}.} 
%, or empty, denoted by $T$. \todo{(stands for \emph{t}ime)} 
We denote the set of all continuous time constraints by $\mathcal {CTC}$. 
A \emph{modal continuous-time automaton (MCA)} over $\Sigma$ is a tuple 
$\spec = (Q, q_0, \omay,\omust,\otimetran)$, where
\vspace*{-0.7em}

\begin{itemize}
 \item $Q$ is a non-empty finite set of \emph{locations} and $q_0 \in Q$ is an \emph{initial location},
 \item $\omust,\omay:Q\times\Sigma\to Q$ are \emph{must} and \emph{may} transition functions, respectively, satisfying $\omust\subseteq\omay$,
 \item $\otimetran: Q \to \mathcal {CTC} \times Q$ is a \emph{time flow} function.
\end{itemize}
\end{definition}

We have seen an example of an MCA in the previous example. Note that upon taking $\mathsf{req}$ from the first state, the waiting time 
%distributed according to $\mathrm{Exp(3)}$ - or faster!!
is chosen and the waiting starts. On the other hand, when $\mathsf{req}$ \emph{self-loop} is taken in the middle state, the waiting process is not restarted, but continues on the background independently.\footnote{This makes no difference for memoryless exponential distributions, but for all other distributions it does.} We introduce this independence as a useful feature to model properties as ``response follows within some time after request'' in the setting with concurrently running processes. Further, we have transitions under $\top$ corresponding to ``$>0$'', meaning there is no restriction on the time distribution except that the transition takes non-zero time. We formalize this in the following definition. With other respects, the semantics of may and must transitions follows the standards of modal transition systems~\cite{DBLP:conf/lics/LarsenT88}.

\begin{definition}[MCA semantics]\label{def:semantics}
An IMC $\env=(E,\act,\mathord{\itran},\mathord{\mtran},e_0)$ \emph{conforms} to 
an MCA specification $\spec=(Q, q_0, \omay,\omust,\otimetran)$, written $\env\models\spec$, 
if there is a \emph{satisfaction relation} $\rel\subseteq E\times Q$ containing $(e_0,q_0)$ and satisfying for each $(e,q)\in\rel$ that whenever
\begin{enumerate}
 \item $q\must{a}q'$ then there is some $e\acttran{a}e'$ and if, moreover, $q\neq q'$ then $e'\rel q'$,
 \item $e\acttran{a}e'$ then there is (unique) $q\may{a}q'$ and if, moreover, $q\neq q'$ then $e'\rel q'$,
 %\item $s\timetran{T}s'$ then $m\rel s'$,
 \item $e\acttran{\tau}e'$ then $e'\rel q$,
 \item $q\timetran{ctc}q'$ then for every IMC $\imc$ and every scheduler $\pi\in\scheduler(\imc|e)$,\footnote{Here $e$ stands for the IMC $\env$ with the initial state $e$.} 
%choosing whenever possible only transitions under $a$ (before hiding) with $s\may{a}s$, 
there is a random variable $\mathit{Stop}:\setofruns\to\Rsetp$ on 
% $\mathcal P_{\imc(m)}^\pi$
the probability space $(\setofruns,\sigmafield,\probm_{\imc|e}^\pi)$
such that 
\begin{itemize}
 \item if $ctc$ is of the form $\bowtie d$ then the cumulative distribution function of $\mathit{Stop}$ is point-wise $\bowtie$ cumulative distribution function of $d$ (there are no constraints when $ctc=\top$), and
 \item for every run $\rho$ of $\imc|e$ under $\pi$, either a transition corresponding to synchronization on action $a$ with 
$q\may{a}q''\neq q$ is taken before time $\mathit{Stop(\rho)}$, or 
\begin{itemize}
 \item the state $(c,e')$ visited at time $\mathit{Stop}(\rho)$ satisfies $e'\rel q'$, and
 \item for all states $(\bar c,\bar e)$ visited prior to that, whenever
%  $(\bar e,q)$ satisfies that whenever
\begin{enumerate}
 \item $q\must{a}q''$ then there is $\bar e \acttran{a}e'$,
 \item $\bar e\acttran{a}e'$ then there is $q\may{a}q''$.
\end{enumerate}
\end{itemize}
\end{itemize}
\end{enumerate}
The semantics of $\spec$ is the set $\semantics{\spec}=\{\env\in\IMC\mid \env\models\spec\}$ of all conforming IMC.
\end{definition}

\begin{wrapfigure}{r}{0.305\textwidth}
\vspace*{-2.5em}
\begin{tikzpicture}[->,>=stealth',
% state style
state/.style={shape=circle,draw,inner sep=1mm,outer sep=0.8mm, minimum size=.5cm,},
branch/.style={shape=circle,draw,inner sep=0.5mm,thick,outer sep=0.8mm},
multiline/.style={text width=13mm,text centered},
% target state
target/.style={double distance=0.2mm},
% probability style
transition/.style={rounded corners},
rate/.style={auto,color=black},
external/.style={rate,color=red},
tau/.style={rate,color=blue},
prob/.style={rate,color=black},
]

%%%%%%%%%%%%%%%%%%%%%%%%%%%%%%%%%%%%%%%%%%%%%%%%%%%%%
% STATES

\node (s1) [state] {};
\node (t1) at (2.5,0) [state] {};
%%%%%%%%%%%%%%%%%%%%%%%%%%%%%%%%%%%%%%%%%%%%%%%%%%%%%
% TRANSITIONS
\path ($(s1.west)+(-0.3,0)$) 		edge [transition]		(s1);
\path (t1) edge [loop right] node[right] {$b$} (t1);
\path (s1) edge [] node[above] {$\leq\mathrm{Er}(3,1)$} (t1);
\path (s1) edge [bend right] node[below] {$a$} (t1);
% \end{tikzpicture}
% ~
% \begin{tikzpicture}[->,>=stealth',
% % state style
% state/.style={shape=circle,draw,inner sep=1mm,outer sep=0.8mm, minimum size=.8cm,},
% branch/.style={shape=circle,draw,inner sep=0.5mm,thick,outer sep=0.8mm},
% multiline/.style={text width=13mm,text centered},
% % target state
% target/.style={double distance=0.2mm},
% % probability style
% transition/.style={rounded corners},
% rate/.style={auto,color=black},
% external/.style={rate,color=red},
% tau/.style={rate,color=blue},
% prob/.style={rate,color=black},
% ]

%%%%%%%%%%%%%%%%%%%%%%%%%%%%%%%%%%%%%%%%%%%%%%%%%%%%%
% STATES

\node (s) at (0,-1.3) [state] {};
\node (t) at (1.3,-1.3) [state] {};
\node (u) at (2.5,-1.3) [state] {};
%%%%%%%%%%%%%%%%%%%%%%%%%%%%%%%%%%%%%%%%%%%%%%%%%%%%%
% TRANSITIONS
\path ($(s.west)+(-0.3,0)$) 		edge [transition]		(s);
\path (s) edge node[above] {$1$} (t);
\path (t) edge node[above] {$1$} (u);
\path (s) edge[bend right] node[below] {$a$} (u.south);
\path (t) edge[bend right=20] node[below] {$a$} (u);
\path (u) edge [loop right] node[right] {$b$} (u);
\end{tikzpicture}
\vspace*{-3.5em}
\end{wrapfigure}
\addtocounter{example}{1}\noindent
\textit{Example \theexample.\ }
We illustrate this definition.
Consider the MCA on the right above specifying that $a$ is ready and $b$ will be ready either immediately after taking $a$ or within the time distributed according to the Erlang distribution $\mathrm{Er}(3,1)$, which is a convolution of three $\mathrm{Exp(1)}$ distributions. 
The IMC below conforms to this specification (here, $\mathit{Stop} \sim Er(2,1)$ can be chosen). However, observe that 
% that if were not for the right $a$ transition, it would not satisfy it. 
% This explains the last bullet of the previous definition.
%% Krc
it would not conform, if there was no transition under $a$ from the middle to the right state. Satisfying the modalities throughout the waiting is namely required by the last bullet of the previous definition.

%\myspace

\subsection{Assume-Guarantee Optimization}
\label{sec:assume}

\myspaceb

We can now formally state what guarantees on time-bounded reachability we can derive provided the unknown environment conforms to a specification $\spec$.
Given an \emph{open} IMC $\imc$,
a set of goal states $\goal\subseteq C$ and a time bound $T\in\Rsetpo$,
the \emph{value of compositional time-bounded reachability conditioned by an MCA $\spec$} is defined as
\begin{align*}
v_{\spec}(\imc)&\quad:=\quad
\sup_{\sigma\in\scheduler(\imc)} \inf_{\substack{\env\in \ENV: \env\models\spec\\
      \pi\in\scheduler_\sigma(\closed)}}
  \probm^{\pi}_{\closed}\big[\reach^{\leq T}\goal\big] %\tag{$v_{\spec}(\imc)$}
\end{align*}

In this paper, we pose a technical assumption on the set of schedulers of $\imc$. For some clock resolution $\delta > 0$, we consider only such schedulers $\sigma$ that take the same decision for any pair of paths $c_0 t_1 \ldots t_{n} c_n$ and $c_0 t'_1 \ldots t'_{n} c_n$ with $t_i$ and $t'_i$ equal when rounded down to a multiple of $\delta$ for all $1 \leq i \leq n$. This is no practical restriction as it is not possible to achieve arbitrary resolution of clocks when implementing the scheduler. Observe this is a safe assumption as it is \emph{not} imposed on the unknown environment.

We consider specifications $\spec$ where distributions have differentiable density functions. 
In the rest of the paper we show how to 
approximate $v_\spec(\imc)$ for such $\spec$. Firstly, we make a product of the given IMC and MCA. Secondly, we transform the product to a game. This game is further discretized into a partially observable stochastic game played on a dag 
where the quantitative reachability is solved.
% We conclude by showing how to solve quantitative reachability there.
%Full proofs can be found in~\cite{myproofs}.
For full proofs, see~\apprefwhole.%Appendix http://www.model.in.tum.de/\texttildelow kretinsk/concurIMC.pdf

%% file: product.tex
\myspace

\section{Product of IMC and Specification}\label{sec:product}
\myspaceb

In this section, we first translate MCA $\spec$ into a sequence of IMC $(\spec_i)_{i\in\Nset}$. Second, we combine the given IMC $\imc$ with the sequence $(\spec_i)_{i\in\Nset}$ into a sequence of product IMC $(\imc\times\spec_i)_{i\in\Nset}$ that will be further analysed. The goal is to reduce the case where the unknown environment is bound by the specification to a setting where we solve the problem for the product IMC while quantifying over all possible environments (satisfying only a simple technical assumption  discussed at the end of the section), denoted $\ENV'$. The reason why we need a sequence of products instead of one product is that we need to approximate arbitrary distributions with more and more precise and detailed hyper-Erlang distributions expressible in IMC. Formally, we want to define the sequence of the products $\product_i$ so that 
$$
v_{\mathit{product}}(\imc\times\spec_i) \quad:=\quad
\sup_{\sigma\in\scheduler(\imc)} \inf_{\substack{\env\in \ENV'\\ \pi\in\scheduler_\sigma(\closedspeci{i})}}
  \probm^{\pi}_{\closedspeci{i}}\big[\reach^{\leq T}\goal\big]
$$
approximates the compositional value:

\begin{theorem}\label{thm:product}
For every IMC $\imc$ and MCA $\spec$, $\displaystyle v_{\spec}(\imc)=\lim_{i\to\infty}v_{\mathit{product}}(\imc\times\spec_i)$.
\end{theorem}
\noindent Note that in $v_{\mathit{product}}$, $\sigma$ is a scheduler over $\imc$, not the whole product $\imc\times\spec_i$.\footnote{Here we overload the notation $\scheduler_\sigma(\closedspeci{i})$ introduced for pairs in a straightforward way to triples, where $\sigma$ ignores both the second and the third components.} Constructing a product with the specification intuitively corresponds to adding a known, but uncontrollable and unobservable part of the environment to $\imc$. We proceed as follows: We translate the MCA $\spec$ into a sequence of IMC $\spec_i$ and then the product will be defined as basically a parallel composition of $\imc$ and $\spec_i$.

There are two steps in the translation of $\spec$ to $\spec_i$. Firstly, we deal with the modal transitions. A \emph{may} transition under $a$ is translated to a standard external transition under $a$ 
% Krc 
that has to synchronize with $a$ in both $\imc$ and $\env$ simultaneously, 
so that the environment may or may not 
% synchronize on it. 
let the synchronization occur.
Further, each \emph{must} transition under $a$ is replaced by an external transition, that synchronizes with $a$ in $\imc$, but is 
% immediately hidden. 
hidden before making product with the environment.
This way, we guarantee that $\imc$ can take $a$ and make progress no matter if the general environment $\env$ would like to synchronize on $a$ or not.

%  Krc: prejmenovani na q
% 
% Formally, the must transition are transformed into special ``barred'' transition that will be immediately hidden in the product $\imc\times \spec$ as opposed to transitions arising from may transitions.
% Let $\overline \actnotau=\{\bar a\mid a\in\actnotau\}$ denote a fresh copy of the original alphabet. 
% We replace all modal transitions as follows
% \begin{itemize}
%  \item whenever $s\may{a}t$ set $s\acttran{a}t$,
%  \item whenever $s\must{a}t$ set $s\acttran{\bar a}t$. 
% \end{itemize}
% 

Formally, the must transitions are transformed into special ``barred'' transitions that will be immediately hidden in the product $\product_i$ as opposed to transitions arising from may transitions.
Let $\overline \actnotau=\{\bar a\mid a\in\actnotau\}$ denote a fresh copy of the original alphabet. 
We replace all modal transitions as follows
\begin{itemize}
 \item whenever $q\may{a}r$ set $q\acttran{a}r$,
 \item whenever $q\must{a}r$ set $q\acttran{\bar a}r$. 
\end{itemize}

% 
% The second step is to deal with the \emph{timed} transitions, especially with the phase-type-distribution constraints. Every phase type distribution corresponds to a continuous-time Markov chain (an IMC with only timed transitions) with a given starting state and a sink state.
% 
%% Krc
The second step is to deal with the \emph{timed} transitions, especially with the constraints  of the form $\bowtie d$. Such a transition is, roughly speaking, replaced by a phase-type approximation of $d$. This is a continuous-time Markov chain (an IMC with only timed transitions) with a sink state such that the time to reach the sink state is distributed with $d'$. For any continuous distribution $d$, we can find such $d'$ arbitrarily close to $d$.

% \begin{example}
\vspace{0.5em}
\addtocounter{example}{1}\noindent
\textit{Example \theexample.\ }
Consider the following MCA on the left. It specifies that whenever $\mathsf{ask}$ is taken, it cannot be taken again for at least the time distributed by $\mathrm{Er}(2,\lambda)$ and during all that time, it is ready to synchronize on $\mathsf{answer}$.
% all the time. 
This specifies systems that are allowed to \text{ask}, but not too often, and whenever they ask, they must be ready to receive (possibly more) $\mathsf{answer}$s for at least the specified time.

\begin{center}
\vspace*{-2em}
\begin{tikzpicture}[->,>=stealth',
% state style
state/.style={shape=circle,draw,inner sep=0mm,outer sep=0.8mm, minimum size=.7cm,},
branch/.style={shape=circle,draw,inner sep=0.0mm,thick,outer sep=0.8mm},
multiline/.style={text width=13mm,text centered},
% target state
target/.style={double distance=0.2mm},
% probability style
transition/.style={rounded corners},
rate/.style={auto,color=black},
external/.style={rate,color=red},
tau/.style={rate,color=blue},
prob/.style={rate,color=black},
]
%%%%%%%%%%%%%%%%%%%%%%%%%%%%%%%%%%%%%%%%%%%%%%%%%%%%%
% STATES
\node (s1) [state] {r};
\node (t1) at (1.8,0) [state] {q};
%%%%%%%%%%%%%%%%%%%%%%%%%%%%%%%%%%%%%%%%%%%%%%%%%%%%%
% TRANSITIONS
\path ($(s1.west)+(-0.4,0)$) 		edge [transition]		(s1);
\path (t1) edge [loop above] node[above] {$\mathsf{answer}$} (t1);
\path (s1) edge [densely dashed] node[above] {$\mathsf{ask}$} (t1);
\path (t1) edge [bend left] node[below] {$\geq\mathrm{Er}(2,\lambda)$} (s1);
%%%%%%%%%%%%%%%%%%%%%%%%%%%%%%%%%%%%%%%%%%%%%%%%%%%%%

\begin{scope}[xshift=3.3cm]
% STATES
\node (s) at (0,0) [state] {1};
\node (t) at (1.1,0) [state] {2};
\node (u) at (2.2,0) [state] {0};
%%%%%%%%%%%%%%%%%%%%%%%%%%%%%%%%%%%%%%%%%%%%%%%%%%%%%
% TRANSITIONS
\path ($(s.west)+(-0.4,0)$) 		edge [transition]		(s);
\path (s) edge node[above] {$\lambda$} (t);
\path (t) edge node[above] {$\lambda$} (u);
\end{scope}

\begin{scope}[xshift=7cm]
 %%%%%%%%%%%%%%%%%%%%%%%%%%%%%%%%%%%%%%%%%%%%%%%%%%%%%
% STATES
\node (s) [state] {r};
\node (t) at (1.4,0) [state] {q=1};
\node (aux) at (2.6,0) [state] {2};
\node (sink) at (3.8,0) [state] {0};
%%%%%%%%%%%%%%%%%%%%%%%%%%%%%%%%%%%%%%%%%%%%%%%%%%%%%
% TRANSITIONS
\path ($(s.west)+(-0.3,0)$) 		edge [transition]		(s);
\path (t) edge [loop above,min distance=5mm,in=75,out=105] node[above] {$\overline{\mathsf{answer}}$} (t);
\path (aux) edge [loop above,min distance=5mm,in=75,out=105] node[above] {$\overline{\mathsf{answer}}$} (aux);
\path (sink) edge [loop above,min distance=5mm,in=75,out=105] node[above] {$\overline{\mathsf{answer}}$} (sink);
\path (s) edge [] node[above] {$\mathsf{ask}$} (t);
\path (t) edge [] node[above] {$\lambda$} (aux);
\path (aux) edge [] node[above] {$\lambda$} (sink);
\draw [->,rounded corners] (sink) |- ($(sink)+(-1,-0.8)$) node[above,pos=1.4]{$\nowtran$} -| (s);
% \path (sink) edge [bend left] node[below] {$\nowtran$} (s);
\end{scope}
\end{tikzpicture}
\vspace*{-2em}
\end{center}
After performing the first step of replacing the modal transitions as described above, we proceed with the second step as follows. We replace the timed transition with a phase-type, e.g.~the one represented by the IMC in the middle. Observe that while the Markovian transitions are taken, $\mathsf{answer}$ must still be available. Hence, we duplicate the corresponding self-loops on all the new states. Further, since the time constraint is of the form $\geq$, getting to the state $(q,0)$ does not guarantee that we already get to the state $r$. It can possibly take longer. To this end, we connect the states $(q,0)$ and $r$ by a special external action $\nowtran$. Since this action is synchronized with $\env \in \ENV'$, the environment can block the progress for arbitrarily long time. Altogether, we obtain the IMC on the right.

% \begin{center}
% \begin{tikzpicture}[->,>=stealth',
% % state style
% state/.style={shape=circle,draw,inner sep=1mm,outer sep=0.8mm, minimum size=1.0cm,},
% branch/.style={shape=circle,draw,inner sep=0.5mm,thick,outer sep=0.8mm},
% multiline/.style={text width=5mm,text centered},
% % target state
% target/.style={double distance=0.2mm},
% % probability style
% transition/.style={rounded corners},
% rate/.style={auto,color=black},
% external/.style={rate,color=red},
% tau/.style={rate,color=blue},
% prob/.style={rate,color=black},
% ]
% 
% \end{tikzpicture}
% \end{center}
% %
In the case of ``$\leq$'' condition, we would instead add the $\nowtran$ transition from each auxiliary state to the sink, which could instead shorten the waiting time.
% \end{example}
\vspace{0.5em}

When constructing $\spec_i$, we replace each distribution $d$ with its hyper-Erlang phase-type approximation $d_i$ with $i$ branches of lengths 1 to $i$ and rates $\sqrt i$ in each branch. For formal description, see~\appref{Appendix~\ref{app:hypererlang}}.
Formally, let $\nowtran\notin\actnotau\cup\overline\actnotau$ be a fresh action. We replace all timed transitions as follows:
\begin{itemize}
 \item whenever $q\timetran{\top}r$ such that $q \neq r$ set $q\acttranlong{\nowtran}r$,
 \item whenever $q\timetran{\bowtie d}r$ where the phase-type $d_i$ corresponds to a continuous-time Markov chain (IMC with only timed transitions) with the set of states $D$, the initial state $\mathit{1}$ and the sink state $\mathit{0}$, then 
\begin{enumerate}
 \item identify the states $q$ and $\mathit{1}$,
 \item for every $u\in D$ and $q\acttran{\alpha}q$, %with $\mathord{\Rightarrowtran{}}\in\{\omust,\omay\}$, 
set $u\acttran{\alpha}u$,
 \item for every $u\in D$ and $q\acttran{\alpha}p$ with %$\mathord{\Rightarrowtran{}}\in\{\omust,\omay\}$ and 
$p\neq q$, set $u\acttran{\alpha}p$,
% \item if $\mathord{\bowtie} = \mathord{=}$, then identify $r$ and $\mathit{0}$,
 \item if $\mathord{\bowtie} = \mathord{\leq}$, then identify $r$ and $\mathit{0}$, and set $u\acttranlong{\nowtran}r$ for each $u\in D$,
 \item if $\mathord{\bowtie} = \mathord{\geq}$, then set $\mathit{0}\acttranlong{\nowtran}r$.
\end{enumerate}
\end{itemize}
Intuitively, the new timed transitions model the delays, while in the ``$\leq$'' case, the action $\nowtran$ can be taken to speed up the process of waiting, and  in the ``$\geq$'' case, $\nowtran$ can be used to block further progress even after the delay has elapsed.

The product is now the parallel composition of $\imc$ and ${\spec_i}$, where each action $\bar a$ synchronizes with $a$ and the result is immediately hidden. Formally, the product $\imc\times\spec$ is defined as $\,\imc\;\parallel_{\actnotau\cup\overline{\actnotau}}^{\textbf{PC6}}\;{\spec_i}\,$, where $\parallel_{\actnotau\cup\overline{\actnotau}}^{\textbf{PC6}}$ is the parallel composition with one additional axiom:
\begin{description}
 \item[\textbf{(PC6)}] $s_1\acttran{a}s'_1$ and $s_2\acttran{\bar{a}}s'_2$ implies
    $(s_1,s_2)\acttran{\tau}(s'_1,s'_2)$,
\end{description}
saying that $a$ synchronizes also with $\bar a$ and, in that case, is immediately hidden (and any unused $\bar a$ transitions are thrown away).

The idea of $\nowtran$ is that it can be taken in arbitrarily short, but non-zero time. To this end, we define $\ENV'$ in the definition of %Theorem~\ref{thm:product} 
$v_{\mathit{product}}(\imc\times\spec_i)$
to denote all environments where $\nowtran$ is only available in states that can be entered by only a Markovian transition. Due to this requirement, each $\nowtran$ can only be taken after waiting for some time.
%We now give a brief idea why our construction of the product $\imc\times\spec$ is correct. For the full proof, see~\cite{myproofs}.

% Kratsi:
% The idea of $\nowtran$ is that it can be taken in arbitrarily short, but non-zero time. To this end, we define $\ENV'$ of %Theorem~\ref{thm:product} 
% $v_{\mathit{product}}(\imc\times\spec_i)$
% to denote all environments where $\nowtran$ is only available in states that can be entered by only a Markovian transition. Thus each $\nowtran$ is only taken after waiting for some time.
% %We now give a brief idea why our construction of the product $\imc\times\spec$ is correct. For the full proof, see~\cite{myproofs}.

%% file: cegames.tex
% $Id: cegames.tex 341 2012-10-09 13:11:50Z xkretins $

\myspace

\section{Controller-Environment Games}\label{sec:ceg}

\myspaceb

So far, we have reduced our problem to computing $\lim_{i\to\infty}v_{\mathit{product}}(\imc\times\spec_i)$.
%$$\displaystyle v_{\mathit{product}}(\imc\times\spec)=\sup_{\sigma\in\scheduler(\imc)} \inf_{\substack{\env\in \ENV'\\
%       \pi\in\scheduler(\closedspec, \sigma)}}
%   \probm^{\pi}_{\closedspec}\big[\reach^{\leq T}\goal_\env\big]$$
%
% Note that we are still quantifying over unknown environments. Therefore, the behaviour of each $\env$ is limited by the uncontrollable \emph{stochastic} flow of time caused by its timed transitions. For example, in some environment, $a$ is not available until some particular random time passes. This setting is still too difficult to be solved directly. Therefore, in this section, we reduce this setting to one, where the stochastic flow of time of the environement (limited in an unknow way) is replaced by a free \emph{non-deterministic} choice of the \emph{second player}. % Krc: shortened  heavily
%
Note that we are still quantifying over unknown environments. Further, the behaviour of each environment is limited by the uncontrollable \emph{stochastic} flow of time caused by its Markovian transitions. 
%For example, when waiting in a state without an $a$ transition, it cannot synchronize over $a$ before a Markovian transition is taken. 
%Since $\ENV$ is an infinite set, t
This setting is still too difficult to be solved directly. Therefore, in this section, we reduce this setting to one, where the stochastic flow of time of the environment (limited in an unknown way) is replaced by a free \emph{non-deterministic} choice of the \emph{second player}.

We want to turn the product IMC $\imc\times\spec_i$ into a two-player \emph{controller--environment game} (CE game) $\game_i$, where
player $\playercon$ controls the decisions over internal transitions in $\imc$; and player $\playerenv$ simulates the environment including speeding-up/slowing-down $\spec$ using $\nowtran$ transitions.
In essence, $\playercon$ chooses in each state with internal transitions one of
them, and $\playerenv$ chooses in each state with external (and hence
synchronizing) transitions either which of them should be taken, or a \emph{delay} $d \in \Rsetp$ 
during which no synchronization occurs.
The internal and external transitions take zero time to be executed if chosen. 
% If no zero time transition is chosen, the delay $t_\ext$ determined by $\playerenv$ competes with the Markovian transitions.
% The delay, if chosen, may be terminated sooner by a Markovian transitions.
Otherwise, the game waits until either the delay $d$ elapses or a Markovian transition occurs.

% for $t_M < t_\ext$ if a Markovian transition at a sooner time instant $t_M$.

This is the approach taken in~\cite{DBLP:conf/fsttcs/BrazdilHKKR12} where no specification is considered. However, there is a catch. This construction is only correct under the assumption of~\cite{DBLP:conf/fsttcs/BrazdilHKKR12} that there are no states of $\imc$ with both external and internal transitions available.

\begin{wrapfigure}{r}{0.35\textwidth}
\vspace*{-4em}
\begin{tikzpicture}[outer sep=2,/tikz/initial text=,->,
state/.style={shape=circle,draw,inner sep=1mm,outer sep=0.8mm, minimum size=.7cm,},
transition/.style={rounded corners}]
\node[state] (5) {$\mathsf{i}$};
\node[state] (0) at (1.1,0) {$\mathsf{?}$};
%\node[state] (0) [below right of=5] {$\mathsf{?}$};
\node[state] (1) at (2,0.6) {$\mathsf{yes}$};
\node[state] (2) at (2,-0.6) {$\mathsf{no}$};
\node[state,accepting] (4) at (3.2,0.6) {$\mathsf{win}$};
\node[state] (3) at (3.2,-0.6) {$\mathsf{fail}$};
\path  (5) edge node[auto] {$\lambda$} (0) 
         (1) edge node[auto] {$\mathsf{a}$} (4)
         (0) edge node[auto,swap] {$\tau$} (2)
         (0) edge node[auto] {$\tau$} (1)
         (2) edge node[auto,swap,below right=-3.5] {$\tau$} (4)
         (2) edge node[auto,swap] {$\mathsf{a}$} (3);
\path ($(5.west)+(-0.3,0)$) 		edge [transition]		(5);
%\node  [font=\large] at (0,0.9) {$\imc_\mathit{imperfect}$};
\end{tikzpicture}
% ~
% \begin{tikzpicture}[node distance=1.2cm, outer sep=2,/tikz/initial text=,->,font=\scriptsize]
% \node[state,initial] (s) {q};
% \path  (s) edge[loop above, densely dashed] node[above] {$\mathsf{a}$} (s);
% \end{tikzpicture}
\vspace*{-3.5em}
\end{wrapfigure}
% \begin{example}
\vspace{0.5em}
\addtocounter{example}{1}\noindent
\textit{Example \theexample.\ }
Consider the IMC $\imc$ on the right 
(for instance with a trivial specification not restricting the environment).
%(having a single state and a may self-loop for every action).
Note that there are both internal and external actions available in $\mathsf{no}$.

As $\tau$ transitions take zero time, the environment $\env$ must spend almost all the time in states without $\tau$.
Hence, when $\mathsf{?}$ is entered, $\env$ is almost surely in such a state $e$.
Now $\tau$ form $\mathsf{?}$ is taken and $\env$ cannot move to another state when $\mathsf{yes}$/$\mathsf{no}$ is entered. 
Since action $\mathsf{a}$ either \emph{is} or \emph{is not} available in $e$, the environment cannot choose to synchronize in $\mathsf{no}$ and not to synchronize in $\mathsf{yes}$.
%
% Therefore, if $\mathsf{yes}$ is entered and $\mathsf{a}$ was available in $e$, the $\mathsf{a}$ transition to $\mathsf{win}$ has to be taken. Similarly, if the scheduler of $\imc$ decides to move to $\mathsf{no}$ and there was no $\mathsf{a}$ available in $e$, it cannot be taken now, and $\tau$ to $\mathsf{win}$ must be taken.
%
% Yet, the availability of $\mathsf{a}$ is crucial: the scheduler of the IMC wins in $\mathsf{yes}$ iff $\mathsf{a}$ is available in $e$ (and must be taken) and wins in $\mathsf{no}$ iff $\mathsf{a}$ is not available in $e$ (and cannot be taken).
%
%All in all, 
As a result, the environment \emph{``commits'' in advance} to synchronize over $\mathsf{a}$ either in both $\mathsf{yes}$ and $\mathsf{no}$ or in none of them. Therefore, in the game we define, $\playerenv$ cannot completely freely choose which external transition is/is not taken.
Further, note that the scheduler of $\imc$ cannot observe whether $\mathsf{a}$ is currently available in $\env$, which intrinsically induces imperfect information.
% \end{example}
\vspace{0.5em}

In order to transfer these ``commitments'' to the game, we again make use of the compositionality of IMC and put the product $\imc\times\spec_i$ in parallel with an 
IMC $\mathit{Commit}$ and then define the game on the result. 

\begin{wrapfigure}{r}{0.38\linewidth}
\vspace*{-2.2em}
\begin{tikzpicture}[/tikz/initial text=,font=\footnotesize,->,sloped,rounded corners,
state/.style={shape=circle,draw,inner sep=0mm,outer sep=0.8mm, minimum size=.7cm,}
]

\node[state,inner sep=4.25] (C) at (-0.3,-0.5) {$com.$};
\node[state,inner sep=4.25] (N) at (3.7,-0.5) {$now?$};
\node[state] (a) at (1.5,0) {$\{\mathsf{a}\}$};
\node[state] (e) at (1.5,-1) {$\emptyset$};

\path ($(C.west)+(-0.3,0)$) 		edge		(C);

\path (a) edge [bend right=30] node[above] {$\mathsf{a}$} (C);
\path (C) edge node[above] {$\tau$} (a);
\path (C) edge node[below] {$\tau$} (e);
\path (a) edge node[above] {$\committran$} (N);
\path (e) edge node[below] {$\committran$} (N);

% \draw (N) |- (1.5,.5) node[above,pos=0.3] {$\tau$} -| (C);
% \draw (N) |- (1.5,-1.5) node[right,pos=0.1] {$\nowtran$} -| (C);
\draw (N) |- (1.5,.5)  -| (C);
\draw (N) |- (1.5,-1.5)   -| (C);
\node at (3.9,0.3) {$\tau$};
\node at (4.2,-1.3) {$\nowtran$};
\end{tikzpicture}
\vspace*{-3.2em}
\end{wrapfigure}

The action alphabet of $\mathit{Commit}$ is $\actnotau\cup\{\nowtran,\committran\}$ and the state space is $2^{\actnotau} \cup \{\commitstate,\nowstate\}$ (in the figure, $\actnotau = \{\mathsf{a}\}$; for formal description, see~\appref{Appendix~\ref{app:commit}}).
State $A \subseteq \actnotau$ corresponds to $\env$ being committed to the set of currently available actions $A$. Thus $A\acttran{a}\commitstate$ for each $a\in A$.
% 
% We have just seen that in each moment, 
% In each moment, the environment $\env$ is ``committed'' to a subset of currently available actions that it must take if possible, and must not take the others. 
% 
% Then it has to respect its commitment until it takes a transition or a timed transition occurs.
% 
% As this commitment of $\env$ corresponds to its current state, 
This commitment must be respected until the state of $\env$ is changed: either (1) by an external transition from the commitment set (which in $\mathit{Commit}$ leads to the state $\commitstate$ where a new commitment is immediately chosen); or (2) by a $\committran$ transition (indicating the environment changed its state due to its Markovian transition).
% (on which $\mathit{Commit}$ and $\env$ synchronize only after $\env$ makes a Markovian transition into a state where $\committran$ is available).
% 

% \begin{wrapfigure}{r}{0.38\linewidth}
% \vspace*{-1em}
% \begin{tikzpicture}[/tikz/initial text=,font=\footnotesize,->,sloped,rounded corners,
% state/.style={shape=circle,draw,inner sep=1mm,outer sep=0.8mm, minimum size=.6cm,}
% ]
% 
% \node[state,inner sep=4.25] (C) at (-0.3,-0.5) {$c.$};
% \node[state,inner sep=4.25] (N) at (3.7,-0.5) {$n?$};
% \node[state] (ab) at (1.5,1) {$ab$};
% \node[state] (a) at (1.5,0) {$a$};
% \node[state] (b) at (1.5,-1) {$b$};
% \node[state] (e) at (1.5,-2) {$\emptyset$};
% \node (middle) at (0.7,-0.5) {};
% \node (middle2) at (2.2,-0.5) {};
% 
% \path ($(C.west)+(-0.3,0)$) 		edge		(C);
% 
% \path (ab) edge [bend right=30] node[above] {$a,b$} (C);
% \path (a) edge [bend right=30] node[above] {$a$} (C);
% \path (b) edge [bend left=30] node[below] {$b$} (C);
% 
% \draw (C) -- (middle.center)  node[right=2] {$\tau$}  -- (a);
% \draw (C) -- (middle.center) -- (b);
% \draw (C) -- (middle.center) -- (ab.230);
% \draw (C) -- (middle.center) -- (e.140);
% 
% \draw (ab) -- (middle2.center)  -- node[above] {$\committran$} (N);
% \draw (a) -- (middle2.center) -- (N);
% \draw (b) -- (middle2.center) -- (N);
% \draw (e) -- (middle2.center) -- (N);
% 
% \draw (N) |- (1.5,1.5) node[below,pos=0.7] {$\tau$} -| (C);
% \draw (N) |- (1.5,-2.5) node[above,pos=0.7] {$\nowtran$} -| (C);
% \end{tikzpicture}
% \vspace*{-3.2em}
% \end{wrapfigure}
% 
% 
The game $\game_i$ is played on the arena $\big(\imc\times\spec_i \parallel_{\actnotau\cup\{\nowtran\}}\commit\big)\,\hide\,\big({\actnotau\cup\{\nowtran\}}\big)$ with its set of states denoted by ${\arena}$. Observe that external actions have either been hidden (whenever they were available in the commitment), or discarded (whenever not present in the current commitment). The only external action that remains is $\committran$. 
%\todo{Krc: I would prefer replacing the definition with a figure, if there is time to draw one.}
% 
The game $\game_i$ is played as follows. There are two types of states: \emph{immediate} states with some $\tau$ transitions available and \emph{timed} states with no $\tau$ available. The game starts in $v_0=(c_0,q_0,\mathit{\commitstate})$.
\begin{itemize}
\item %
% In an immediate vertex $v_n=(c,s,e)$, we distinguish two type of internal transitions. \emph{Controller's} transitions correspond to the internal transitions in $\imc$. \emph{Environment's} transitions are all the remaining, i.e. correspond to the synchronization transitions or to the internal transitions in $\spec$ or $\commit$.
% %  
In an immediate state $v_n=(c,q,e)$, $\playercon$ chooses a probability distribution over transitions corresponding to the internal transitions in $\imc$ (if there are any). Then, $\playerenv$ either approves this choice (chooses $\checkmark$) and $v_{n+1}$ is chosen randomly according to this distribution, or rejects this choice and chooses a $\tau$ transition to some $v_{n+1}$ such that the transition does \emph{not} correspond to any internal transitions of $\imc$. Then the game moves at time $t_{n+1}=t_{n}$ to $v_{n+1}$.
%  
%  In an immediate vertex, $v_n=(c,s,e)$, $\playercon$ chooses a $\tau$ transition to some $v_{n+1}:=(c',s,e)$ if there is any (this corresponds to choosing a $\tau$ transition of $\imc$). The $\playerenv$ either approves this choice, or rejects and chooses another $\tau$ transition to another $v_{n+1}$ that does \emph{not} correspond to any $\tau$ transition of $\imc$. Then the game moves in time $t_n=0$ to $v_{n+1}$.
%  
 \item In a timed state $v_n=(c,q,e)$, $\playerenv$ chooses a delay $d>0$. Then Markovian transitions (if available) are resolved by randomly sampling a time $t$ according to the exponential distribution with rate $\rates(v_n)$ and randomly choosing a target state $v_{n+1}$ where each $v_n\probtran{\lambda}v$ is chosen with probability $\lambda/\rates(v_n)$.
\begin{itemize}
 \item If $t<d$, $\game_i$ moves at time $t_{n+1}=t_n+t$ to $v_{n+1}$,\;
 {(\scriptsize Markovian transition wins)}
 \item else $\game_i$ moves at time $t_{n+1}=t_n+d$ to $(c,q,\nowstate)$.\hfill
 {(\scriptsize $\env$ takes $\committran$)}
\end{itemize}
\end{itemize}
This generates a run $v_0 t_1 v_1 t_1\cdots$. The set $(\arena\times \Rsetpo)^\ast\times \arena$ of prefixes of runs is denoted $\histories(\game)$.
% 
% According to the definition of schedulers in IMC, 
We formalize the choice
of $\playercon$ as a \emph{strategy} $\sigma: \histories(\game_i) \rightarrow
\dist(\arena)$. We further allow the $\playerenv$ to randomize and thus his \emph{strategy}  is
$\pi:\histories(\game_i) \rightarrow \dist(\{\checkmark\}\cup\arena)\cup\dist(\Rsetp)$.
%\{\checkmark\}\cup\arena\cup\Rsetp$.
% 
We denote by $\Sigma$ and $\Pi$ the sets of all strategies of the players $\playercon$ and $\playerenv$, respectively. 

Since $\playercon$ is not supposed to observe the state of the specification and the state of $\commit$, 
% Further, to keep the correspondence of a strategy $\sigma$ and a scheduler $\sigma\in\scheduler(\closedspec)$, 
we consider in $\Sigma$ only those strategies that satisfy $\sigma(p)=\sigma(p')$, whenever \emph{observations} of $p$ and $p'$ are the same. Like before, the observation of $(c_0,q_0,e_0)t_1\cdots t_{n}(c_n,q_n,e_n)\in\histories(\game)$ is a sequence obtained from $c_0t_1\cdots t_{n}c_n$ by replacing each maximal consecutive sequence $t_i\,c_i\cdots t_j\,c_j$ with all $c_k$ the same, by $t_i\,c_i$. This replacement takes place so that the player cannot observe transitions that do not affect $\imc$. Notice that now $\scheduler(\imc)$ is in one-to-one correspondence with $\Sigma$.
Further, in order to keep CE games out of Zeno behaviour, we consider in $\Pi$ only those strategies 
for which the induced Zeno runs have zero measure, i.e.\ the sum of the chosen
delays diverges almost surely no matter what $\playercon$ is doing.
The \emph{value of} $\game_i$ is now defined as
$$
v_{\game_i} \quad:=\quad 
\sup_{\sigma \in \Sigma} \inf_{\pi \in \Pi} \probm^{\sigma, \pi}_{\game_i}\big[\reach^{\leq T}\goal\big]
$$
where $\probm^{\sigma, \pi}_{\game_i}\big[\reach^{\leq T}\goal\big]$ is the probability of all runs of $\game_i$ induced by $\sigma$ and $\pi$ and reaching a state with the first component in $\goal$ before time $T$. %\todo{leave out this sentence?}
We now show that it coincides with the value of the $i$th product:

\begin{theorem}\label{thm:ceg-ith}
For every IMC $\imc$, MCA $\spec$, $i\in\Nset$, we have $v_{\game_i}=v_{\mathit{product}}(\imc\times\spec_i)$.
\end{theorem}

This result allows for approximating $v_{\spec}(\imc)$ through computing $v_{\game_i}$'s. However, from the algorithmic point of view, we would prefer approximating $v_{\spec}(\imc)$ by solving a single game $\game$ whose value $v_\game$ we could approximate directly. This is indeed possible. But first, we need to clarify, why the approximation sequence $\spec_i$ was crucial even in the case where all distributions of $\spec$ are already exponential.

\begin{wrapfigure}{r}{0.25\textwidth}
\vspace*{-2.4em}
\begin{tikzpicture}[->,>=stealth',
% state style
state/.style={shape=circle,draw,inner sep=1mm,outer sep=0.8mm, minimum size=.7cm,},
branch/.style={shape=circle,draw,inner sep=0.5mm,thick,outer sep=0.8mm},
multiline/.style={text width=13mm,text centered},
% target state
target/.style={double distance=0.2mm},
% probability style
transition/.style={rounded corners},
rate/.style={auto,color=black},
external/.style={rate,color=red},
tau/.style={rate,color=blue},
prob/.style={rate,color=black},
]
% %%%%%%%%%%%%%%%%%%%%%%%%%%%%%%%%%%%%%%%%%%%%%%%%%%%%%
% % STATES
% \node (s1) [state] {};
% \node (t1) at (1.7,0) [state] {};
% \node (u1) at (2.8,0) [state] {};
% 
% \node (s2) at (0,-1)[state] {};
% \node (t2) at (1.7,-1) [state] {};
% \node (u2) at (2.8,-1) [state] {};
% %%%%%%%%%%%%%%%%%%%%%%%%%%%%%%%%%%%%%%%%%%%%%%%%%%%%%
% % TRANSITIONS
% \path ($(s1.west)+(-0.4,0)$) 		edge [transition]		(s1);
% \path ($(s2.west)+(-0.4,0)$) 		edge [transition]		(s2);
% \path (s1) edge [] node[above] {$1$} (t1);
% \path (s2) edge [] node[above] {$\geq\mathrm{Exp}(1)$} (t2);
% \path (t1) edge [] node[above] {$a$} (u1);
% \path (t2) edge [] node[above] {$a$} (u2);

%%%%%%%%%%%%%%%%%%%%%%%%%%%%%%%%%%%%%%%%%%%%%%%%%%%%%
% STATES
\node (s1) [state] {$q$};
\node (t1) at (2,0) [state] {$r$};

%%%%%%%%%%%%%%%%%%%%%%%%%%%%%%%%%%%%%%%%%%%%%%%%%%%%%
% TRANSITIONS
\path ($(s1.west)+(-0.4,0)$) 		edge [transition]		(s1);
\path (s1) edge [] node[above] {$\geq\mathrm{Exp}(1)$} (t1);
\path (s1) edge [loop above,dashed] node[left,pos=0.2] {$a$} (s1);
\path (t1) edge [loop above] node[right,pos=0.8] {$b$} (t1);
\end{tikzpicture}
\vspace*{-3.2em}
\end{wrapfigure}
Consider the MCA on the right and a conforming environment $\env$, in which $a$ is available iff $b$ becomes available within 0.3 time units. If Player $\playerenv$ wants to simulate this behaviour, he needs to know how long the transition to $r$ is going to take so that he can plan his behaviour freely, only sticking to satisfying the specification. If we translate Exp(1) directly to a single Markovian transition (with no error incurred), $\playerenv$ knows nothing about this time as exponential distributions are memoryless. On the other hand, with finer hyper-Erlang, he knows how long the current branch of hyper-Erlang is roughly going to take. In the limit, he knows the precise waiting time right after coming to $q$.
%
% However, in $\closedspeci{i}$, the Markovian transitions of $\spec$ and  $\env$ proceed \emph{independently}. Therefore, with 50\% chance $\env$ will offer $a$ before $\spec$ moves to the location with the outgoing $a$. At that moment, the composition with $\spec_i$ disallows to perform $a$ yet. Hence the distribution on when $a$ can occur will be greater than Exp(1). As a result, $\env$ is not able to be offer $a$ as fast as it does alone.
%

To summarize, $\playerenv$ is too weak in $\game_i$, because it lacks the information about the precise time progress of the specification. The environment needs to know how much time is left before changing the location of $\spec$. Therefore, the game $\game$ is constructed from $\game_1$ by multiplying the state space with $\Rsetpo$ where we store the exact time to be waited. After the product changes the state so that the specification component switches to a state with $\bowtie d$ constraint, this last component is overwritten with a number generated according to $d$. This way, the environment knows precisely how much time is left in the current specification location. This corresponds to the infinitely precise hyper-Erlang, where we at the beginning randomly enter a particular branch, which is left %from which $sink$ is reached 
in time with Dirac distribution. For more details, see~\appref{Appendix~\ref{app:ce-precise}}.

Denoting the \emph{value of} $\game$ by
$%\begin{align*}
\displaystyle 
v_{\game} := \sup_{\sigma \in \Sigma} \inf_{\pi \in \Pi} \probm^{\sigma, \pi}_{\game}\big[\reach^{\leq T}\goal\big],
$ %\end{align*}
we obtain:
%where $\probm^{\sigma, \pi}_{\game}\big[\reach^{\leq T}\goal'\big]$ is the probability of all runs of $\game$ induced by $\sigma$ and $\pi$ and reaching a state with the first component in $\goal$ before time $T$.\todo{leave out this sentence?}

% \newcommand{\theoremcegamesoktwo}{
%  For every IMC $\imc$ and MCA $\spec$, we have $v_\spec(\imc)=v_{\game}$.
% }

% \myspaceb

\begin{theorem}\label{thm:ceg-direct}
For every IMC $\imc$ and MCA $\spec$, we have $\displaystyle v_{\game}=\lim_{i\to\infty}v_{\game_i}$. %\todo{Krc: quantifying over $\imc$ is weird when $\game$ seems to be fixed.}
\end{theorem}

\myspace

%%% Local Variables:
%%% mode: latex
%%% TeX-master: "main"
%%% End:

%% file: discretization.tex
% $Id: discretization.tex 339 2012-07-16 09:31:05Z xkretins $

\myspace

\section{Approximation using discrete-time PO games}\label{sec:disc}

\myspaceb

In this section, we briefly discuss the approximation of $v_\game$ by a discrete time turn-based partial-observation stochastic game $\gamed$. 
% 
% It is a turn-based stochastic game where each vertex belongs either to the player 0 (corresponding to $\playercon$ with partial observation), or to player 1 (corresponding to $\playerenv$ with full observation), or to nature (playing randomly). 
The construction is rather standard; hence, we do not treat the technical difficulties in great detail (see Appendix~\ref{app:disc}). % in the main body of the paper.
We divide the time bound $T$ into $N$ intervals of length $\step = T/N$
%. We assume 
such that the clock resolution $\delta$ (see Section~\ref{sec:assume}) satisfies $\delta = n \step$ for some $n\in\Nset$.

% We assume the clock resolution $\delta$ (from Section~\ref{sec:assume}) divides the time bound $T$ into $T/\delta = N$ intervals; we further divide each of these intervals into $M \in\Nset$ subintervals of length $\step = \delta/M$.
% 
\begin{enumerate}
 \item We enhance the state space with a counter $i \in \{0,\ldots,N\}$ that tracks that $i\cdot\step$ time has already elapsed. Similarly, the $\Rsetpo$-component of the state space is discretized to $\step$-multiples.
In timed states, time is assumed to pass exactly by $\step$. In immediate states, actions are assumed to take zero time. 

% Players $\playercon$ and $\playerenv$ take their steps in zero time. Delaying by $\playerenv$ or waiting for Markovian transitions corresponds to random steps in $\gamed$. Each random step corresponds to $\delta$ time in $\game$; in this time at most one Markovian transition is allowed to occur.  
\item We let at most one Markovian transition occur in one step in a timed state.
\item We unfold the game into a tree until on each branch a timed state with $i=N$ is reached.
Thereafter, $\gamed$ stops.
We obtain a graph of size bounded by $b^{\leq N \cdot |G|}$ where $b$ is the maximal branching and $G$ is the state space of $\game$.
\end{enumerate}
Let $\Sigma_\gamed$ and $\Pi_\gamed$ denote the set of randomized history-dependent strategies of $\playercon$ and $\playerenv$, respectively, where player $\playercon$ observes in the history only the first components of the states, i.e. the states of $\imc$, and the elapsed time $\lfloor i / n\rfloor$ up to the precision $\delta$. Then
$v_\gamed := \sup_{\sigma\in\Sigma_\gamed} \inf_{\pi\in\Pi_\gamed} \probm^{\sigma,\pi}_\gamed (\reach G)$
denotes the value of the game $\gamed$ 
where $\probm^{\sigma,\pi}_\gamed (\reach G)$ is the probability of the runs of $\gamed$ induced by $\sigma$ and $\pi$ and reaching a state with first component in $G$. Let $b$ be a constant bounding (a) the sum of outgoing rates for any state of $\imc$, and (b) densities and their first derivative for any distribution in $\spec$.

\begin{theorem}\label{thm:disc}
\theoremerrorbound
\end{theorem}

The proof of the error bound extends the technique of the previous bounds of~\cite{Neu10} and \cite{DBLP:conf/fsttcs/BrazdilHKKR12}. Its technical difficulty stems from partial observation and from semi-Markov behaviour caused by the arbitrary distributions in the specification.
The game is unfolded into a tree in order to use the result of~\cite{DBLP:conf/stoc/KollerMS94}. Without the unfolding, the best known (naive) solution would be a reduction to the theory of reals, yielding an EXPSPACE algorithm.

%% file: conclusion.tex
\myspace

\section{Summary}

\myspace

We have introduced an assume-guarantee framework for IMC. We have considered the problem to 
% $\varepsilon$-
approximate the guarantee on time-bounded reachability properties in an unknown environment $\env$ 
% assuming $\env$ 
that
satisfies a given assumption. The assumptions are expressed in a new formalism, which introduces continuous time constraints. The algorithmic solution results from Theorems 1 to 4:
% show that the problem can be solved in exponential time. 
\begin{corollary}
For every IMC $\imc$ and MCA $\spec$ and $\eps > 0$, a value $v$ and a scheduler $\sigma$ can be computed in exponential time such that 
%$v \, \leq \, v_\spec(\imc) \, \leq \, v + \eps$ 
$|v_\spec(\imc)-v| \, \leq \, \eps$ 
and $\sigma$ is $\eps$-optimal in $v_\spec(\imc)$.
\end{corollary}
In future work, we want to focus on identifying structural subclasses of IMC allowing for polynomial analysis.

\myspace

\subsubsection*{Acknowledgement}
The work has received support from the Czech Science Foundation, project No.~P202/12/G061, from the German Science Foundation DFG as part of SFB/TR~14~AVACS, and from the EU FP7 Programme under grant agreement no.\@ 295261 (MEALS) and 318490 (SENSATION).
We also thank Tom\'a\v s Br\'{a}zdil and Vojt\v ech \v{R}eh\'ak for fruitful discussions and for their feedback.

\myspace

%%% Local Variables:
%%% mode: latex
%%% TeX-master: "main"
%%% End:

%% file: app-results.tex
% $Id: app-results.tex 343 2012-10-10 13:49:18Z xkretins $

\allowdisplaybreaks

\section*{Appendix}

\bigskip

%\section{Section~\ref{sec:product}}%: Proof of Theorem~\ref{thm:product}}
%\label{app:product}

\section{Additional Examples and Technical Definitions}

\subsection{Product: An Example}\label{app:ex}

\begin{example}
Let us illustrate the product of the MCA from Example 3 and the IMC below. The MCA specifies a client who might want to $\mathsf{ask}$ a question and then must be able to receive an $\mathsf{answer}$ for some random time. It  may be the case that he is e.g.~asking different services at once and thus he gets answer at a random time, after which he is not willing to get $\mathsf{answer}$ from the server any more.

The following IMC is a server accepting $\mathsf{ask}$ and after some computation lasting a random time, it provides an $\mathsf{answer}$. Observe that it is not ready to proceed if the other side is no more willing to synchronize on $\mathsf{answer}$.

\begin{center}
\begin{tikzpicture}[->,>=stealth',
% state style
state/.style={shape=circle,draw,inner sep=1mm,outer sep=0.8mm, minimum size=.8cm,},
branch/.style={shape=circle,draw,inner sep=0.5mm,thick,outer sep=0.8mm},
multiline/.style={text width=5mm,text centered},
% target state
target/.style={double distance=0.2mm},
% probability style
transition/.style={rounded corners},
rate/.style={auto,color=black},
external/.style={rate,color=red},
tau/.style={rate,color=blue},
prob/.style={rate,color=black},
]

%%%%%%%%%%%%%%%%%%%%%%%%%%%%%%%%%%%%%%%%%%%%%%%%%%%%%
% STATES

\node (s) [state] {x};
\node (t) at (2,0) [state,multiline] {y};
\node (u) at (4,0) [state,multiline] {z};
%%%%%%%%%%%%%%%%%%%%%%%%%%%%%%%%%%%%%%%%%%%%%%%%%%%%%
% TRANSITIONS
\path ($(s.west)+(-0.5,0)$) 		edge [transition]		(s);
\path (s) edge node[above] {$\mathsf{ask}$} (t);
\path (t) edge node[above] {$\kappa$} (u);
\path (u) edge [bend left] node[below] {$\mathsf{answer}$} (s);
\end{tikzpicture}
\end{center}

The product is then the following IMC. Observe that after each $\mathsf{ask}$ we reach the state $(y,0)$ with non-zero probability, from where a $\nowtran$ transition leads to a deadlock state $(y,r)$. Since this is an external action, some environments conforming to the specification can ensure that the system eventually ends up almost surely in a deadlock (revealing this error in the implementation of the server).

\begin{center}
\begin{tikzpicture}[->,>=stealth',
% state style
state/.style={shape=circle,draw,inner sep=1mm,outer sep=0.8mm, minimum size=1.2cm,},
branch/.style={shape=circle,draw,inner sep=0.5mm,thick,outer sep=0.8mm},
multiline/.style={text width=5mm,text centered},
% target state
target/.style={double distance=0.2mm},
% probability style
transition/.style={rounded corners},
rate/.style={auto,color=black},
external/.style={rate,color=red},
tau/.style={rate,color=blue},
prob/.style={rate,color=black},
]
%%%%%%%%%%%%%%%%%%%%%%%%%%%%%%%%%%%%%%%%%%%%%%%%%%%%%
% STATES
\begin{scope}[xscale=1.1]
\node (xt) at (2,-4) [state,initial] {x,r};
\node (yt) at (10,0) [state] {y,r};
\node (zt) at (10,-2) [state] {z,r};
\node (ysink) at (8,0) [state] {y,0};
\node (yaux) at (6,0) [state] {y,2};
\node (ys) at (4,0) [state] {y,q};
\node (xsink) at (8,-4) [state] {x,0};
\node (xaux) at (6,-4) [state] {x,2};
\node (xs) at (4,-4) [state] {x,q};
\node (zsink) at (8,-2) [state] {z,0};
\node (zaux) at (6,-2) [state] {z,2};
\node (zs) at (4,-2) [state] {z,q};
\end{scope}
%%%%%%%%%%%%%%%%%%%%%%%%%%%%%%%%%%%%%%%%%%%%%%%%%%%%%
% TRANSITIONS
\path ($(xt.west)+(-0.4,0)$) 		edge [transition]		(xt);
\path (xt) edge[] node[left]{$\mathsf{ask}$}   (ys);
\path (ys) edge [] node[right] {$\kappa$} (zs);
\path (yaux) edge [] node[right] {$\kappa$} (zaux);
\path (ysink) edge [] node[right] {$\kappa$} (zsink);
\path (zs) edge [] node[right] {$\tau$} (xs);
\path (zaux) edge [] node[right] {$\tau$} (xaux);
\path (zsink) edge [] node[right] {$\tau$} (xsink);
\path (ys) edge [] node[above] {$\lambda$} (yaux);
\path (yaux) edge [] node[above] {$\lambda$} (ysink);
\path (ysink) edge  node[below] {$\nowtran$} (yt);
\path (xs) edge [] node[above] {$\lambda$} (xaux);
\path (zs) edge [] node[above] {$\lambda$} (zaux);
\path (zaux) edge [] node[above] {$\lambda$} (zsink);
\path (zsink) edge  node[below] {$\nowtran$} (zt);
\path (xaux) edge [] node[above] {$\lambda$} (xsink);
\path (xsink) edge [bend left] node[above] {$\nowtran$} (xt);
\end{tikzpicture}
\end{center}
\end{example}

%\newpage

\subsection{Product: Hyper-Erlang Phase-Types}\label{app:hypererlang}

%\noindent {\bf Hyper-Erlang approximations}\todo{mozna i formalni popis h-E}
%
Recall that in $\spec_i$ each distribution $d$ is replaced with its hyper-Erlang phase-type approximation $d_i$ with $i$ branches of lengths 1 to $i$ and rates $\sqrt i$ (and rate $2^{2^i}$ in the initial state). The only degrees of freedom in the approximation for a fixed $i$ are thus the initial probabilities leading to the branches fo lengths $1$ to $i$. For concreteness, we pick a distribution that is lexicographically smallest (w.r.t.\ the order given by the lengths of branches) such that the resulting cumulative distribution function of time when sink is reached is still pointwise greater than or equal to that of $d$ for the case $\geq d$, and lexicographically largest so that the cdf is pointwise smaller or equal.

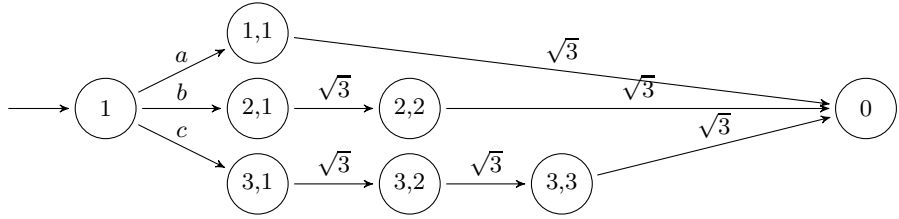
\begin{figure}
\centering
 \begin{tikzpicture}[x=2cm,->,>=stealth',
% state style
state/.style={shape=circle,draw,inner sep=1mm,outer sep=0.8mm, minimum size=.8cm,},
branch/.style={shape=circle,draw,inner sep=0.5mm,thick,outer sep=0.8mm},
multiline/.style={text width=5mm,text centered},
% target state
target/.style={double distance=0.2mm},
% probability style
transition/.style={rounded corners},
rate/.style={auto,color=black},
external/.style={rate,color=red},
tau/.style={rate,color=blue},
prob/.style={rate,color=black},
]
%%%%%%%%%%%%%%%%%%%%%%%%%%%%%%%%%%%%%%%%%%%%%%%%%%%%%
% STATES
\node(1)at(0,1)[state]{1}; 
\node(11)at(1,2)[state]{1,1}; 
\node(21)at(1,1)[state]{2,1}; 
\node(22)at(2,1)[state]{2,2}; 
\node(31)at(1,0)[state]{3,1}; 
\node(32)at(2,0)[state]{3,2}; 
\node(33)at(3,0)[state]{3,3}; 
\node(0)at(5,1)[state]{0}; 

%%%%%%%%%%%%%%%%%%%%%%%%%%%%%%%%%%%%%%%%%%%%%%%%%%%%%
% TRANSITIONS
\path ($(1.west)+(-0.4,0)$) 		edge [transition]		(1);
\path (1) edge[] node[above]{$a$}   (11);
\path (1) edge[] node[above]{$b$}   (21);
\path (1) edge[] node[above]{$c$}   (31);
\path (21) edge[] node[above]{$\sqrt 3$}   (22);
\path (31) edge[] node[above]{$\sqrt 3$}   (32);
\path (32) edge[] node[above]{$\sqrt 3$}   (33);
\path (11) edge[] node[above]{$\sqrt 3$}   (0);
\path (22) edge[] node[above]{$\sqrt 3$}   (0);
\path (33) edge[] node[above]{$\sqrt 3$}   (0);
\end{tikzpicture}

\caption{Hyper-Erlang phase-type for $i=3$ expressing approximately $\frac{a}{256}\cdot Er(1,\sqrt 3)+\frac{b}{256}\cdot Er(2,\sqrt 3)+\frac{c}{256}\cdot Er(3,\sqrt 3)$ with $a+b+c=256$}
\end{figure}

%\newpage

% \subsection{Conjunctions of Continuous-Time Constraints}\label{app:conjunctions}

%\newpage

\subsection{Transitions of Commit}\label{app:commit}

For each commitment $A\subseteq\actnotau$ there is an incoming internal transition from $\commitstate$, an outgoing $\committran$ transition to $\nowstate$, and an outgoing external $a$-transition to $\commitstate$ for each $a\in A$. Furthermore, from $\nowstate$ there is an internal and a $\nowtran$ transition to $\commitstate$:

\begin{itemize}
 \item $\commitstate\acttran{\tau}A$,
 \item $A\acttran{a}\commitstate$, for each $a\in A$,
 \item $A\acttranlong{\committran}\nowstate$,
 \item $\nowstate\acttran{\tau}\commitstate$ and $\nowstate\acttranlong{\nowtran}\commitstate$.
\end{itemize}

\newpage

\section{Proof of Theorem~\ref{thm:product}}\label{app:product}

\begin{reftheorem}{thm:product}
For every IMC $\imc$ and MCA $\spec$, $\displaystyle v_{\spec}(\imc)=\lim_{i\to\infty}v_{\mathit{product}}(\imc\times\spec_i)$, i.e.
\begin{align*}
\sup_{\sigma\in\scheduler(\imc)} \inf_{\substack{\env\in \ENV: \env\models\spec\\
      \pi\in\scheduler_\sigma(\closed)}}
  \probm^{\pi}_{\closed}\big[\reach^{\leq T}\goal_\env\big]
 =
\lim_{i\to\infty}
\sup_{\sigma\in\scheduler(\imc)} \inf_{\substack{\env\in \ENV'\\ \pi\in\scheduler_\sigma(\closedspeci{i})}}
  \probm^{\pi}_{\closedspeci{i}}\big[\reach^{\leq T}\goal_\env\big]
\end{align*}
\end{reftheorem}

\noindent {\bf Proof idea}

% \begin{proof}[Sketch of proof of Theorem~\ref{thm:product}] \todo{to appendix}
``$\leq$'': Given an environment $\env$ over $\act$ and a scheduler $\pi$, we construct $\env'$ satisfying $\spec$ that will ensure the same value with the same $\pi$. Intuitively, $\env'$ is a composition of $\env$ and $\spec$ such that when composed with $\imc$, we obtain $\closedspec$. Formally, this is exactly $\env':=\env\parallel_{\actnotau}\overline{\spec}$ where all ``barred'' transitions are renamed to unbarred afterwards.

%$\parallel_{\actnotau\cup\overline{\actnotau}}^{\textbf{PC7}}$ is the parallel composition with one additional axiom:
% \begin{description}
%  \item[\textbf{PC7}] $s_1\acttran{a}s'_1$ and $s_2\acttran{\bar{a}}s'_2$ implies
%     $(s_1,s_2)\acttran{\tau}(s'_1,s'_2)$,
% \end{description}
 
``$\geq$'': Given an environment $\env$ satisfying $\spec$, we construct a sequence of environments $\env_i$ 
% that when composed with $\imc\times\spec$ 
to be composed with $\product_i$ that
monitor the changes of $\spec$ and behave in such a way that together with $\imc\times\spec$ simulate the original $\env$. The hyper-Erlang form of the phase-type allows for arbitrary precise monitoring. Technically, the environment $\env_i$ takes with very high frequency transitions to a special new state, where it checks the progress of $\spec$ and simulates the corresponding behaviour of $\env$.
%\todo{Krc: I don't like the second part of the sketch, if there is time I will rewrite it}
\qed
% \end{proof}

\bigskip\bigskip

\noindent {\bf Proof}

%\begin{proof}
``$\leq$'':

We show the inequality holds for any $i$.
Given $i\in\Nset, \sigma\in\scheduler(\imc), \env\in\ENV',\pi\in\scheduler_\sigma(\closedspeci{i})$, we construct $\env'$ such that $$\imc|\env'=\closedspeci{i}$$
and hence $\pi$ is a scheduler over both systems and yields the same value on both systems. 

We construct $\env''=\spec_i\parallel_{\actnotau}\env$ and then $\env'$ by renaming $\bar a$ actions to $a$, for each $a\in\act$.
By case distinction, it is easy to see that the very same transitions are created in $\closed'$ and in $\closedspeci{i}$.

Furthermore, we need to prove that $\env'\models\spec$. We show that 
$$\{((s,e),s)\mid \text{$s$ state of $\spec_i$ corresponding to a state of $\spec$ and $e$ state of $\env'$}\}$$
 is a satisfaction relation. Observe that whenever there is a may transition in the specification, there is a corresponding transition in $\spec_i\parallel_{\actnotau}\env$, and if, moreover, there is a must transition, then there is also an $\bar a$, which is then renamed to $a$. Further, we need to define the variable $\mathit{Stop}$ when in state $(s,e)$ and with $s\timetranlong{\bowtie d}s'$. Since the state space of $\env'$ is a product of state spaces of $\spec_i$ and $\env$, we can define $\mathit{Stop}$ to return 
\begin{itemize}
 \item time when the first component becomes $s'$ on runs that do not leave states of the phase-type corresponding to $s$ meanwhile;
 \item an arbitrary time on the remaining runs, which leave the states of the phase-type before reaching a state with $s'$ in the first compoment, so that the cumulative distribution function of $\mathit{Stop}$ is $\bowtie d$.
\end{itemize}
Observe that such a definition is possible because the cumulative distribution function of $\mathit{Stop}$ conditioned by runs of the first item above satisfies $\bowtie d$.
For $\geq d$, the transition $\nowtran$ is only available after the phase-type has elapsed, which takes at least $d$. Hence before time (distributed by) $d$, $s'$ can only be reached using a transition of the form $s \may{a} s''$ for $s \neq s''$ complying the definition of the semantics. For $\geq d$, the phase-type elapses in time (distributed by) at most $d$ and the sink is identified with $s'$ hence is reached on time (if no non-looping may transition is taken meanwhile) which is again according to Definition~\ref{def:semantics}.

\bigskip

``$\geq$'':

For $\sigma\in\scheduler(\imc), \env\models\spec,\pi\in\scheduler_\sigma(\closed)$ and for each $\varepsilon>0$, we construct a sequence $(\env_i,\pi_i)_{i\in\Nset}$ such that
\begin{align*}
 \probm^{\pi}_{\closed}\big[\reach^{\leq T}\goal_\env\big]
&\geq
\lim_{n\to\infty}\probm_{\closedspeci{i}_i}^{\pi_i}\big[\reach^{\leq T}\goal_{\spec_i,\env}\big] - \varepsilon
% \text { and }\\
% \forall \env\in\ENV \quad \probm^{\pi}_{\closedspec}\big[\reach^{\leq T}\goal_\env\big]
% &=
% \lim_{n\to\infty}\probm_{\imc\times\spec_i\, (\env)}^{\pi}\big[\reach^{\leq T}\goal_\env\big]
\end{align*}%\todo{$\forall\pi\in\scheduler(?)$}

Recall that in $S_i$ each distribution $d$ is replaced by a hyper-Erlang distribution with $i$ branches (of all lengths up to $i$ with rates $\sqrt i$. A branch of length $\ell$ corresponds to a time $\ell/\sqrt i$ it takes to walk through it on average, and for great $i$'s almost precisely by the law of large numbers. The initial branching probabilities correspond to the probability of this time according to the distribution $d$. In the limit, the hyper-Erlang thus corresponds to the pdf of $d$. Indeed, since for $i$ we have branches taking from $1/\sqrt i$ to $\sqrt i$, we cover the whole interval $(0,\infty)$ in the limit.

% Let $\delta>0$ be fixed (its dependence on $\varepsilon$ will be described later). 
We define $\env_i$ by actions $\actnotau\cup\{\nowtran\}$, the state space is $2^{\actnotau} \cup \actnotau\cup\{\nowstate,\commitstate\}\cup\{0,1\}\times\{1,\ldots,i\}$ with $\commitstate$ being the initial state. For each $A\subseteq\actnotau$, there are transitions
\begin{itemize}
 \item $\commitstate\acttran{\tau}A$,
 \item $A\acttran{a}a$ and $a\acttran{\tau}\commitstate$, for each $a\in A$,
 \item $A\probtran{2^i}(0,0)$, and $A\probtran{2^i}(1,0)$,
 \item $(b,j)\probtran{2^i}(0,j+1)$, and $(b,j)\probtran{2^i}(1,j+1)$ for all $j<i$ and $b\in\{0,1\}$,
 \item $(b,i)\probtran{2^i}\nowstate$ for $b\in\{0,1\}$,
 \item $\nowstate\acttranlong{\nowtran}\commitstate$ and $\nowstate\acttran{\tau}\commitstate$.
\end{itemize}

In the state $\commitstate$, $\pi$ can perform any sequence of external transitions changing its commitment after each of them (see Section~\ref{sec:ceg}); or a Markovian transition occurs after which a sequence of $i$ random bits is generated; afterwards, $\pi$ returns back to $\commitstate$ (possibly synchronizing over $\nowtran$).

Intuitively, the scheduler $\pi_i\in\scheduler_\sigma(\closedspeci{i}_i)$ simulates behaviour of $\closed$ in such a way that it is never limited by $\spec_i$ (as $\env$ anyway satisfies $\spec$, there is no reason for further limitations). 
% Roughly speaking, for a timed transition $s \timetranlong{\leq d} s'$ in the specification, runs in $\closed$ that are quick to reach  
The transitions under $2^i$ create a sequence of random numbers we remember in the current path. These numbers help to identify, which set of runs is $\pi$ now simulating. 
At all times, the current path of $\closedspeci{i}_i$ induces a set of paths of $\closed$.

Each run of $\closedspeci{i}_i$ is divided into \emph{phases}. A phase starts when the current state $(c,s,e)$ is changed to $(c',s',e')$ with $s\neq s'$, i.e.~the specification enters another state. We show how a path $\mypath$ of length $k$ that starts at the beginning of a new phase induces a set of paths $X_k$ in $\closed$. 

With probability $\to 1$ for $i \to \infty$, the first next step is the Markovian transition of rate $2^{2^i}$ in $\spec_i$ leading into one of the Erlang branches and thus determining very precisely (for great $i$'s) how long the phase is going to take. Assume it is of length $\ell$, thus taking time close to $t=\ell/\sqrt i$. 
Further, the choice of the $i$-th branch corresponds to the interval $[x,y)$ (on the $y$-axes) in the cumulative distribution function of the hyper-Erlang phase-type distribution with $x = \sum_{j=0}^{i-1}a_j$ and $y = \sum_{j=0}^{i}a_j$ where $a_j$ is the probability of taking the $j$-th branch.

Since $\env\models\spec$, there is a random variable $\mathit{Stop}$ on the runs from $(c,e)$ in $\closed$ satisfying the condition of the definition of MCA semantics. The interval $[x,y)$ on the $y$-axes in the CDF of $\mathit{Stop}$ corresponds to a interval $[u,v)$ of times on the $x$-axes.
As the CDF of $\mathit{Stop}$ is pointwise $\bowtie$ the CDF of the hyper-Erlang, for large enough $i$ the expected value $t$ of time to wait in $\spec_i$ lies outside of the interval $[u,v)$. Hence, when simulating the set of runs $X_1 = \{\rho \mid \mathit{Stop}(\rho) \in [u,v) \}$, the scheduler $\pi_i$ is with probability $\to 1$ for $i \to \infty$ \emph{not limited} by $\spec_i$.

Part of the behavior of runs $X_1$ is determined by the randomness in $\imc$ and $\sigma$, rest of the behavior is determined by the randomness in $\env$ and $\pi$. After each step $k$ of the phase, the simulating strategy $\pi_i$ figures out the currently valid subset $X_k \subseteq X$ that conform with the path $\mypath$ in the phase so far. 
\begin{itemize}
 \item If in the $k$-th step a Markovian or internal transition is taken in $\imc$, $X_k$ is defined as those runs of $X_{k-1}$ with the very same move at the very same time (relative to the start of the phase). 
 \item If in the $k$-th step a Markovian transition within the hyper-Erlang of $\spec_i$ is taken, this move is ignored by $X_k = X_{k-1}$.
 \item If in the $k$-th step the Markovian transition to $\nowstate$ in $\env_i$ is taken, we call this moment a emph{control point}. The strategy $\pi_i$ has a fresh sequence of $i$ random bits. The strategy $\pi_i$ divides the runs of $X_{k-1}$ into $2^i$ sets of equal measure (conditioned by $X_{k-1}$) according to the sequence of synchronization performed in these runs since the last control point and according to the current state. This yields with probability $\to 1$ for $i \to \infty$ only constantly many types  of this discrete behavior, hence the number of sets with more than one type of behavior remains constant. Using the random sequence, one such set is assigned to $X_k$.  The strategy simulates the type of discrete behavior with most measure (conditioned by $X_k$): it performs the sequence of synchronization and moves into a commit according to the current state in $X_k$.
\end{itemize}

% Let us consider the set of runs $\rho$ from $(c,e)$ ending at time $\mathit{Stop}(\rho)$. From these we consider those where $\mathit{Stop}(\rho)\in[t-\delta,t+\delta]$. These correspond to the current Erlang distribution and we will 
% simulate 
% these. The string $p_1\cdots p_i$ of random numbers collected on the way from a commitment state to $\nowstate$ and recorded in the path decides whether a Markovian transition of $\env$ in the simulated path of $\closed$ should be taken (and if so, then where). Then we take the $\tau$ transition to $\mathit{commit}$ and to the commitment $A$, which is the set of actions currently available in $e$, where $(c,e)$ is the last vertex of the currently induced path of $\closed$. While waiting for any of the transition under $2^i$ to happen, a synchronization with $\imc$ possibly followed by internal transitions of $\imc$ may happen, which further updates the path. 
After a phase is finished, i.e.~when $\spec_i$ takes the $\ell$th exponential transition, the transition $\nowtran$ is taken and a new phase begins.
The overall induced path is just the concatenation of the paths induced by the previous phases and the current one. 

Now we discuss the possible reasons why the behavior in the simulating IMC might differ from the behavior in $X_k$ w.r.t. time-bounded reachability.
\begin{itemize}
 \item If at any point an Erlang branch of a specification finishes in time outside the assumed interval, $\pi_i$ further behaves arbitrarily. However, the measure of these runs tends $\to 0$ for $i\to\infty$ due to the weak law of large numbers. 
 \item Each synchronization occuring at time $w$ between two control points is simulated later -- at the closest control point at time $w' > w$. The behavior of the simulating system may be different from the simulated system if 
\begin{itemize}
\item a Markovian transition in $\imc$ occurs in the interval $[w,w']$. The number of synchronizations that can occur up to time $T$ is bounded by a constant multiple of the number of Markovian transitions that are taken in $\imc$, due to the acyclicity assumption. The sum of lengths of such intervals $[w,w']$ where a Markovian transition causes trouble thus tends to $0$ with probability $\to 1$ for $i\to\infty$. Hence, this is not a problem.
\item The scheduler takes different decisions because the synchronization has been delayed to $w'$. Notice that this occurs only if $w < n\delta < w'$ for some $n\in\Nset$ (due to our assumption on the set of strategies of $\playercon$). As the points $w$ where synchronization occurs are randomly generated by exponential transitions of $\imc$ or $\env$ and the length of the interval $[w,w']$ tends to $0$, the probability of this behavior also tends to $0$ as $i \to \infty$.
\end{itemize}
\end{itemize}
% 
% 
% Hence the limiting environment simulates $\closed$ up to an error dependent on $\delta$\todo{ten interval musi byt vetsi, a tedy se musi prekryvat. co to dela?}, say at most $\varepsilon_\delta$. Clearly, $\varepsilon_\delta\to0$ for $\delta\to0$.\todo{proboha proc?:-)} Therefore, for every $\varepsilon$ there is a smaller $\varepsilon_\delta$ for some $\delta$ and we choose this $\delta$.
%
% It remains to prove that the limiting $\spec_i$ has the same behaviour as $\spec$. Since the difference is only in the expression of the phase-type distribution as a hyper-Erlang approximation, in the limit they are the same.\todo{$\env$ ma vic informace s jemnejsim hyper-Erlangem nez s divokou malou phase type}
\qed

\newpage

%\section{Section~\ref{sec:ceg}%: Proof of Theorem~\ref{thm:ceg-ith} and Theorem~\ref{thm:ceg-direct}
%}\label{app:ceg}

\section{Proof of Theorem~\ref{thm:ceg-ith}}

\begin{reftheorem}{thm:ceg-ith}
For every IMC $\imc$, MCA $\spec$, $i\in\Nset$, we have $v_{\game_i}=v_{\mathit{product}}(\imc\times\spec_i)$, i.e.
\begin{align*}
\sup_{\sigma \in \Sigma} \inf_{\pi \in \Pi} \probm^{\sigma, \pi}_{\game_i}\big[\reach^{\leq T}\goal\big]
 =
\sup_{\sigma\in\scheduler(\imc)} \inf_{\substack{\env\in \ENV'\\ \pi\in\scheduler_\sigma(\closedspeci{i})}}
  \probm^{\pi}_{\closedspeci{i}}\big[\reach^{\leq T}\goal_\env\big]
\end{align*}
\end{reftheorem}
\bigskip

\noindent {\bf Proof idea}

%\begin{proof}[Proof Idea]\todo{to appendix}
``$\leq$'': We can simulate every $\env$ and scheduler $\pi$ by a strategy of $\playerenv$. The random waiting of $\env$ determined by occurrence of Markovian transitions can be simulated by $\playerenv$ by choosing the delays randomly according to the exponential distribution with the respective rate. Further, in each state of $\env$ only some external actions are available and $\playerenv$ simulates this by changing the commitment to exactly this set of actions.

``$\geq$'': Every strategy of $\playerenv$ can be (approximately) implemented using a suitable environment $\env$ together with a scheduler $\pi$. We need to simulate the discrete delays chosen by $\playerenv$ using random delays available in $\env$. A delay $t$ is simulated by many repetitions of a special Markovian transition with a very fast rate $\lambda$. After the total time adds up to at least $t$, $\pi$ stops repeating it and continues simulating the discrete transitions of $\playerenv$. We get the result by taking $\lambda \to \infty$.
\qed
%\end{proof}

\bigskip

\noindent {\bf Proof}

%\todo{By Lemma~\ref{lem:ce-pure} we can slightly abuse the notation and for the rest of the appendix denote by $\Sigma$ and $\Pi$ the set of \emph{deterministic} strategies.}

``$\leq$'':
%Firstly, we prove the inequality $(*)\geq(**)$. 

This amounts to showing that an arbitrary environment $\env$ can be ``simulated'' by \playerenvc in the CE game. Formally, it is sufficient to prove 
\begin{multline}
\forall\sigma\in\Sigma\ \exists\sigma'\in\scheduler(\imc)\ \forall \env\in\ENV'\ \forall\pi\in\scheduler_{\sigma'}(\closedspeci{i})\ \exists \pi_\env\in\Pi: \\
\probm^{\sigma,\pi_\env}_{\game_i}\big[\reach^{\leq T}\goal\big]
\leq
\probm^{\pi}_{\closedspeci{i}}\big[\reach^{\leq T}\goal\big]
 \tag{$\heartsuit$}
\end{multline}
Note that every strategy $\sigma$ of \playerconc is actually also a scheduler for $\imc$ (and vice versa).
%\todo{strategy can react on sequence of E-waitings, scheduler doesn't see the intermediate steps - define history of CE or lemma that it doe not help}
Thus we set $\sigma':={\sigma}$ and then for every environment $\env$ and its scheduler $\pi$, we give a strategy $\pi_\env$ of \playerenvc that makes ``equivalent'' decisions as $\pi$ in the ``equivalent'' path. We then prove that $\pi_\env$ guarantees the same value as $\pi$ of $\env$ does.

The idea of the simulation is the following. Whenever $\pi$ synchronizes on an external action $a$, $\pi_\env$ chooses $a$. Whenever $\env$ waits with a rate $\lambda$, $\pi_\env$ chooses to wait, too. Here we use randomizing strategies so that we can combine all waiting times $t\in\Rsetp$ with the exponential distribution with rate $\lambda$. In other words, $\pi_\env$ simulates the random waiting of $\env$ using randomizing. 

% One can view this kind of randomizing strategies as strategies where the co-domain contains not only $\Rsetp$ but also exponential distributions and we will use this notation. Similarly, the co-domain contains all distributions over $\{\checkmark\}\cup G_i$ instead of just Dirac distributions as defined before. Altogerther, $$\Pi':\histories(\game_i)\to\dist(\{\checkmark\}\cup G_i)\cup\dist(\Rsetpo)$$ 
% and Lemma~\ref{lem:ce-pure} then shows that such randomizing strategies can be derandomized into strategies in $\Pi$ as introduced in the definition of $\game_i$.\todo{stale plati?}

Let thus $\sigma,\env,\pi$ be arbitrary but fixed. In the following, we define $\pi_\env$ through a function $\wc:\histories(\game_i)\to\paths(\closedspeci{i})$ transforming the paths of the CE game into paths of $\closedspeci{i}$, which $\pi_\env$ uses to ask what $\pi$ would do.
%For a path $\mypath\in\paths$, denote $\myvalue(p):=\sup_{\sigma\in\scheduler(\imc)} \inf_{\substack{\env\in \ENV\\ \pi\in\scheduler(\closed)}}  \probm^{\pi[\sigma]}_{\closed,p}(\reach^{\leq T}\goal_\env) $. 
Since $\env$ can have probabilistic branching and $\pi$ can be randomizing, we need to pick one of possibly more paths of $\closedspeci{i}$ corresponding to the history of the simulating play in $\game_i$. We will pick one where the future chances are the best for the environment, i.e.~worst for the time bounded reachability, hence $\wc$ for the ``worst case''. 

The functions $\pi_\env$ and $\wc$ are defined inductively and only on the reachable histories; one can define them arbitrarily elsewhere. We start with $\wc(c_0,q_0,\commitstate):=(c_0,q_0,e_0)$. For history $\history$ ending with some $t(c,q,\bar e)$ (for the initial one-state path $t=0$) with $\wc(\history)$ ending in $(c,q,e_1)$, we first define what $\pi$ does after a (possibly empty) sequence of internal steps in $\env$. Let 
$(c_1,q,e_1),\ldots,(c_n,q,e_n)$ be such maximal sequence  with
$\pi(\wc(\history) t(c,q,e_2)t\cdots t(c,q,e_i))(c,q,e_{i+1})>0$ and $e_{i}\acttran{\tau}e_{i+1}$ for $i<n$ that minimizes
$$\probm^{\pi}_{\closedspeci{i}}\Big[\reach^{\leq T}\goal_\env\ \big|\ \wc(\history)t(c,q,e_2)t\cdots t(c,q,e_n)\Big]$$ 
If $n=1$ then $(c,q,e_2)\cdots (c,q,e_n)$ is empty.
This way, although $\pi$ is randomizing, we choose a single choice in a unique way which, moreover, is the best one for the environment.

Depending on the type of the last state $(c,q,\bar e)$ of $\history$ and the scheduler's decision $\mathit{dist}:=\pi\big(\wc(\history)t (c,q,e_2)t\cdots t(c,q,e_n)\big)$, we define $\pi_\env(\history)$ as follows:

\begin{itemize}
 \item If $(c,q,\bar e)$ is an immediate state and $\bar e=\commitstate$, then
$$\pi_\env(\history):=A\mapsto1 \text{ for $A$ the set of actions available in $e_n$},$$
$\history':=\history\,t\,(c,q,A)$ and we set $\wc(\history'):=\wc(\history)t(c,e_2)t\cdots t(c,e_n)$.
 \item If $(c,q,\bar e)$ is an immediate state and $\bar e\subseteq\actnotau$, then\\
since $\pi$ respects $\sigma$, there is $p\in[0,1]$ such that $\mathit{dist}=p\cdot \sigma(\wc(\history)_\imc)+(1-p)\cdot \mathit{dist'}$, so we set
$$\pi_\env(\history):=(\checkmark\mapsto p)+(1-p)\cdot\mathit{dist'}$$
Now the next state ($\tau$-successor) is chosen randomly. The corresponding $\tau$ transition is either
\begin{enumerate}
 \item a $\tau$ transition of $\imc$, or
 \item a result of composing $a$ of $\imc$ and $\bar a$ of $\spec_i$, or
 \item a result of composing $a$ of $\imc$ and $a$ of $\spec_i$ and $a$ of $\commit$, 
% \item a result of composing $\nowtran$ of $\spec_i$ and $\nowtran$ of $\commit$.
\end{enumerate}
we thus obtain the next state $(c',q',\bar e')$.
and a new history
$\history':=\history\,t\,(c',q',\bar e')$ and we set $\wc(\history'):=\wc(\history)t(c,q,e_2)t\cdots t(c,q,e_n)(c',q',e')$ where $e'=e_n$ in the first two cases and $e_n\acttran{a}e'$ in the third case. 
\begin{enumerate}
 \item[4.] a result of composing $\nowtran$ of $\spec_i$ and $\nowtran$ of $\commit$: we cannot simulate this in $\bar e\subseteq \actnotau$, but since $\env\in\ENV'$, this case can only happen right after a Markovian transition of $\env$ or $\spec_i$. Therefore, we first discuss the corresponding timed states and we deal with this case below.
\end{enumerate}
 \item If $(c,q,\bar e)$ is a timed state and $\bar e\subseteq\actnotau$, then only Markovian transition(s) are enabled and $\mathit{dist}$ is ignored.
\begin{itemize}
 \item If there are no Markovian transitions available in $e_n$, \hfill /* $\env$ is blocked */\\
 we set $$\pi_\env(\history):=(T+1)\mapsto 1$$
the new history is then either longer than $T$ if no Markovian transition from $c$ or $q$ occurs before $T$, or else a Markovian transition occurs after $m$ still before $T$ and we set $\history'=\history\, (t+m)\, (c',q',e_n)$ given by the respective Markovian successor and further $\wc(\history'):=\wc(\history)t(c,q,e_2)t\cdots t(c,q,e_n)(t+m)(c',q',e_n)$.
 \item Else we set \hfill /* $\env$ waits */
$$\pi_\env(\history):=\exponential(\rates(e_n))$$ 
from which the respective delay $d$ is sampled. Then either a Markovian transition of $\imc$ or $\spec$ happens before $d$, in which case $\history'$ and $\wc(\history')$ are defined as in the previous case;
or else pick arbitrary $e'$ with $e_n\mtran e'$ minimizing
$$\probm^{\pi}_{\closedspeci{i}}\Big[\reach^{\leq T}\goal\ \big|\ \wc(\history')\Big]$$ 
where $\history':=\history\, (t+d)(c,q,\nowstate)$ and $\wc(\history'):=\wc(\history)t(c,q,e_2)t \cdots t(c,q,e_n)(t+d)(c,q,e')$.
\end{itemize}
We distinguish three cases of what happens after a timed transition.
\begin{itemize}
 \item If a Markovian transition of $\imc$ wins, we proceed in the standard way.
 \item If the delay of $\playerenv$ wins then $e'=\nowstate$ and the state is thus immediate. Let $p$ be the probability that $\pi(\wc(\history'))$ chooses a transition stemming from $\nowtran$ ($q\acttranlong{\nowtran}q'',\, e'\acttranlong{\nowtran}e''$). We let $\pi_\env$ choose $\nowtran$ also with $p$, and $\tau$ to $\commitstate$ with $1-p$. 

In the former case, the new history is $\history'':=\history' (t+d)(c,q'',\commitstate)$ and $\wc(\history'')=\wc(\history')(t+d)(c,q'',e'')$.

In the latter case, the new history is $\history'':=\history' (t+d)(c,q,\commitstate)$ and $\wc(\history'')=\wc(\history')$.
 \item If a Markovian transition of $\spec$ wins we get to $(c',q',e_n)$ in both cases with $e_n\subseteq\actnotau$. This state is either without available $\tau$ from $\nowtran$, in which case we proceed in the standard way, or with available $\tau$ from $\nowtran$. The latter happens due to reaching sink in the case of $\geq d$ constraint. Indeed, this is the only case, where $\nowtran$ turns from unavailable to available, and note that the preceding state was timed and $\env$ did not change its state, hence $\nowtran$ indeed was not available. 

Let now $p$ be again the probability that $\pi(\wc(\history'))$ chooses a transition stemming from $\nowtran$ ($q\acttranlong{\nowtran}q'',\, e_n\acttranlong{\nowtran}e''$). We let $\pi_\env$ wait with delay $\delta$ with $p$ (and behave in the standard way with the remaining probability, which is possible as $q$ and $q'$ have the same actions available, see the previous paragraph). If we win, we perform the $\nowtran$ and the new history is $\history'':=\history' (t+d)(c,q',\bar e)(t+d+\delta)(c,q',\nowstate)(t+d+\delta)(c,q'',\commitstate)$ and $\wc(\history'')=\wc(\history)(t+d+\delta)(c,q',e_n)(t+d+\delta)(c,q'',e'')$. Thus, we pretend that the Markovian transition of $\spec_i$ took by $\delta$ longer and the $\tau$ from $\nowtran$ was executed immediately. If $\delta$ does not win, we define the behaviour of the environment arbitrarily. Apparently, for $\delta$ much smaller than inverse of any rate and approaching $0$, the probability that $\delta$ wins is high and the difference in the time distribution of $\
spec$ approaches $0$. Therefore, it is sufficient to pick $\delta:=1/2^{2^{2^i}}$ as the fastest rate is $2^{2^i}$ (the initial rate of the hyper-Erlang).
\end{itemize}
\end{itemize}

\begin{lemma}
For every $\sigma\in\scheduler(\imc),\env\in\ENV',\pi\in\scheduler_\sigma(\closedspeci{i})$, we have 
$$\probm^{\pi}_{\closedspeci{i}}\big[\reach^{\leq T}\goal \big]  \geq
\probm^{\sigma,\pi_\env}_{\game_i}\big[\reach^{\leq T}\goal \big]$$
\end{lemma}
\begin{proof}[Idea]
If there are no probabilistic choices in $\env$ and $\pi$ is deterministic then the values are the same. Indeed, the only difference of the simulating probabilistic space to the original one is that whenever there is a probabilistic choice, the environment is always ``lucky''. Since the minimum of elements is never greater than their affine combination, the result follows.
\end{proof}

\begin{proof}
Formally, we proceed as follows. 

Firstly, we define a measure $\probm^{\env,\pi}_{\game}$ on infinite histories of $\game$ directly induced by $\env$ and $\pi$. As opposed to $\pi_\env$, the probabilistic choices of the environment are reflected here. Let $\mathrm{RealStep}:\paths(\closedspeci{i})\to\histories(\game_i)$ project all internal transitions of the environment out, i.e.\ it maps a run $(c_0,q_0,e_0)t_0(c_1,q_1,e_1)t_1\cdots$ to a run $c_0\,t_0\cdots$ where each $c_i\,t_i$ is omitted whenever $c_i=c_{i-1}$ and $e_{i-1}\acttran{\tau}e_i$. Then we define $\probm^{\env,\pi}_{\game_i}:=\probm^{\pi}_{\closedspeci{i}}\circ\mathrm{RealStep}^{-1}$. Clearly, as $\tau$ transitions take no time we have\footnote{Note that $\env$ and $\pi[\sigma]$ do \emph{not} induce any strategy that would copy the IMC behavior completely. For this, one would need the notion of a strategy with a \emph{stochastic update}, i.e.\ a strategy that can change its ``state'' randomly and thus model where in $\env$ the original path currently is.}
$$\probm^{\pi}_{\closedspeci{i}}[\reach^{\leq T}\goal] = \probm^{\env,\pi}_{\game_i}[\reach^{\leq T}\goal]$$

Secondly, for $j\in\Nseto$, consider the set $\histories_j\subseteq\histories(\game_i)$ of histories of length $j$, i.e.\ after the $j$th step is taken. Let $\vec{p_j}\in\dist(\histories_j)$ denote the transient probability measure according to $\probm^{\env,\pi}_{\game_i}$ after $j$ steps. Further let $\vec{r_j}:\histories_j\to[0,1]$ be given by $\vec{r}_j(\history)=\probm^{\env,\pi}_{\game_i}[\reach^{\leq T}\goal\mid \history]$.
Clearly, as states of $\goal$ are absorbing we have 
$$\probm^{\env,\pi}_{\game_i}[\reach^{\leq T}\goal]=\int\vec{r_j}\de{\vec{p_j}}$$

Thirdly, let $\vec{q_j}\in\dist(\histories_j)$ be the transient probability measure according to $\probm^{\sigma,\pi_\env}_{\game_i}$ after the $i$th step is taken. A simple induction with case distinction from the definition of $\pi_\env$ reveals that 
$$\int\vec{r_j}\de{\vec{p_j}}\geq \int\vec{r_j}\de{\vec{q_j}}$$
Indeed, all but two cases preserve equality. The interesting cases are the Markovian transition in $\env$ and the randomized choice of $\pi$. As the minimum of elements is never greater than their affine combination, we obtain the desired inequality.

Finally, it remains to prove that 
$$\lim_{j\to\infty} \int\vec{r_j}\de{\vec{q_j}}=\probm^{\sigma,\pi_\env}_{\game_i}[\reach^{\leq T}\goal]$$
i.e.\ that the gains of the gradual replacements of the strategy converge to the gain of the limiting strategy. This follows from $\vec{r_j}(\history)$ being zero or one for each path $\history$ longer than $T$ only depending on the state at time $T$, and from the fact that the set of runs that never exceed $T$ is of zero measure due to the acyclicity assumption.
\qed
\end{proof}

The previous lemma proves $(\heartsuit)$ by which the proof of $v_{\game_i}\leq v_{\mathit{product}}(\product_i)$ is concluded.

\bigskip
\bigskip

%Secondly, we prove the inequality $(**)\geq(*)$. 
``$\geq$:

We can divide the proof in two steps: 
% (a) we show that in the CE game \emph{grid} strategies are sufficient for both players; 
\begin{enumerate}
 \item we show that \emph{exponential} strategies are sufficient for \playerenvc; 
 \item any exponential strategy of \playerenvc can be simulated by a specific environment and its scheduler.
\end{enumerate}

 %Let us first define the necessary notions.
%
% We say that a strategy is a \emph{grid} strategy on a grid of size $\delta > 0$ if 
% \begin{itemize}
%  \item it decides only according to the current state
% %, the previous state, 
% and the integer $k$ such that the total time of the history belongs to the interval $[k\delta,(k+1)\delta)$, i.e. for any two histories $\history = s_0 t_0 \ldots t_{n-1} s_n$ and $\history' = s'_0 t'_0 \ldots t'_{m-1} s'_m$ with $s_n = s'_m$
% %, $s_{n-1} = s'_{n-1}$, 
% and $\total{\history}, \total{\history'} \in [k\delta,(k+1)\delta)$ for some $k \in \Nseto$ we have $\sigma(\history) = \sigma(\history')$; and
%  \item it is either a strategy of \playerconc or it chooses waiting times only on the $\delta$-grid, i.e. for any history $\history = s_0 t_0 \ldots t_n s_n$ the strategy $\sigma$ either chooses an action or a time step $t_{n+1}$ such that $\total{\history} + t_{n+1} = k \delta$ for some $k \in \Nset$.
% \end{itemize}
%
% Furthermore, for 
For $\lambda \in \Rset$ we say that a strategy $\pi$ of \playerenvc is \emph{exponential} with rate $\lambda$ if 
%\begin{itemize}
% \item 
it chooses to wait solely with the exponential distribution with rate $\lambda$.
 %\item for any history $\history = s_0 t_0 \ldots t_n s_n$ with $n > 0$, $t_n \neq t_{n-1} 0$, and $s_n \neq s_{n-1}$ it chooses to wait.
%  \item for other histories it behaves as a grid strategy for some $\delta > 0$, i.e. we have $\sigma(\history) = \sigma(\history')$ for any two histories $\history = s_0 t_0 \ldots t_{n-1} s_n$ and $\history' = s'_0 t'_0 \ldots t'_{m-1} s'_m$ with 
% \begin{itemize}
%  \item either $n = m = 0$ or $t_{n-1} = t'_{m-1} = 0$ or $s_n = s_{n-1} = s'_m = s'_{m-1}$, and \\ 
% \phantom{fdfd} \hfill /* the conditions above negated */ 
%  \item $s_n = s'_m$ and $\total{\history}, \total{\history'} \in [k\delta,(k+1)\delta)$ for some $k \in \Nseto$. \hfill /* grid strategy */
% \end{itemize}
%\end{itemize}

%Intuitively, a $\lambda$-exponential strategy cannot take an action right after a Markovian transition (resulting in $s_n \neq s_{n-1}$ and $t_{n-1} > 0$). 
% 
% The set of all grid strategies is denoted by $\Sigma_\grid$ and $\Pi_\grid$, 
The set of all $\lambda$-exponential strategies is denoted by $\Pi_\lambda$.

Exponential strategies for \playerenvc are sufficient:
\begin{lemma}\label{lem:exponential-strategies-suffice}
% against grid strategies, i.e. for any grid strategy $\sigma$ we have
For any strategy $\sigma$ we have 
$$\inf_{\pi \in \Pi} \probm^{\sigma, \pi}_{\game_i}[\reach^{\leq T}\goal]
\; = \; 
\inf_{\substack{\lambda \in \Rsetp\\ \pi_\lambda \in \Pi_\lambda}} \probm^{\sigma, \pi_\lambda}_{\game_i}[\reach^{\leq T}\goal] $$
\end{lemma}
\begin{proof}[Idea]
Intuitively, if $\pi$ chooses to wait for time $t$ and then makes action $a$, the simulating strategy $\pi_\lambda$ repeatedly waits for random time with exponential distribution until the sum of the random waiting times exceeds $t$ and then makes action $a$; the larger the rate $\lambda$, the more precise is this simulation. 
\end{proof}

\begin{proof}
First, we restrict the strategies of $\Pi$ so that on $\Rsetp$ they only pick Dirac distributions, denoted $\Pi'$:
\begin{claim}
$\displaystyle \inf_{\pi \in \Pi'} \probm^{\sigma, \pi}_{\game_i}[\reach^{\leq T}\goal]
\; = \; 
\inf_{\pi \in \Pi} \probm^{\sigma, \pi}_{\game_i}[\reach^{\leq T}\goal] $
\end{claim}
\begin{proof}
We need to simulate $\pi\in\Pi$ by a strategy $\pi'\in\Pi'$. For a history $\history$ and $r\in\Rsetp$, let $v(r)$ be the conditional lower value of the game if $r$ is chosen in $\history$. The value in $\history$ is thus $v(\history):=\int v(r) d\pi(\history)$. By additivity of measure,  there is $r$ for which $v(r)\geq v (\history)$. Therefore, picking $r$ instead of $\pi(\history)$ does not decrease the value.
\qed 
\end{proof}

 We fix arbitrary strategies $\sigma \in \Sigma$ and $\pi \in \Pi'$.% of the same grid size. 
We need to find a sequence of strategies $\pi_\lambda$ for any $\lambda$ such that
$$ \probm^{\sigma, \pi}_{\game_i}[\reach^{\leq T}\goal]
\; \geq \; 
\lim_{\lambda \to \infty} \probm^{\sigma, \pi_\lambda}_{\game_i}[\reach^{\leq T}\goal].$$

For any $\lambda > 0$, we define $\pi_\lambda(\history)$ for $\history = s_0 t_0 \cdots t_{n-1} s_n$ using $\pi$ as follows. Intuitively, if $\pi$ chooses to wait for time $t$ and then makes action $a$, the simulating strategy $\pi_\lambda$ repeatedly waits for random time with exponential distribution until the sum of the random waiting times exceeds $t$ and then makes action $a$; the larger the rate $\lambda$, the more precise is this simulation. Notice that the history of the play with strategy $\pi_\lambda$ contains a lot of waiting steps that are not in the history of the play with strategy $\pi$. Therefore, we need a mapping $\destutter$ that removes these superfluous waiting steps and replaces them with the single waiting $\pi$ would perform.
We define it inductively by $\destutter(s_0) = s_0$ and for $$\history=\history's' t_0 s t_1 s t_2\cdots s t_{n} s''$$ where $s'\neq s\neq s''$ (corresponding to waiting steps of $\pi_\lambda$ where the state is not changed, assuming no Markovian self-loops in $\product$) we set $$\destutter(\history):=\destutter(\history' s' t'_0 s) \, t'_1\,s\cdots s\,t'_k s''$$ where 
\begin{itemize}
 \item $t_0'=t_0$,
 \item $\pi(\destutter(\history' s' t'_0 s)\cdots t'_\ell s)=t'_{\ell+1}-t'_\ell$ for all $0\leq \ell< k-1$,
 \item $t'_k=t'_{k-1}$ if the last transition was immediate, and $t'_k=t_n$ if the last transition was Markovian of $\product_i$
\end{itemize}

Furthermore, let $a'$ be the first action taken by $\pi$ at total time $t'$ for history $\destutter(\history)$ if no Markovian transition occurs (notice that strategy $\pi$ may decide to wait subsequently for several times before it chooses an action; $a'$ is the first action taken by $\pi$ if none of the waiting is interrupted by a Markovian transition). 
We finally set $\pi_\lambda(\history)$ to choose 

\begin{align*}
 \pi_\lambda(\history t s) =  \begin{cases}
                          \exponential(\lambda) & \text{if $t < t'$;} 
% \text{if either $t < t'$ or both $t_{n-1} >0$ and $s_n \neq s_{n-1}$;} 
\\
			  a' & \text{if $t\geq t'$.}
%\text{if $\total{\history} \geq t'$ and either $t_{n-1} = 0$ or $s_n = s_{n-1}$.}
                          \end{cases}
\end{align*}
%where $\exponential(\lambda)$ denotes 
the exponential distribution with rate $\lambda$ in timed states and $\pi(\destutter(\history))$ in immediate states. Notice that the strategy $\pi_\lambda$ is by definition $\lambda$-exponential.

We now define a set of runs $X_\lambda$ in the game with $\pi_\lambda$ where the imprecision in the simulation does not cause any difference with respect to the time bounded reachability.
Let $\delta > 0$ be the clock resolution of $\sigma$. A run in the CE game with strategies $\sigma, \pi_\lambda$ belongs to $X_\lambda$ if for all $k \in \{0,1,\ldots,T/\delta\}$ we have that
\begin{itemize}
 \item no non-self-loop transition occurs at the total time neither in the interval $[k\delta,k\delta + \delta/\sqrt{\lambda}]$ nor in the interval $[(k+1)\delta - \delta/\sqrt{\lambda},(k+1)\delta]$.
 \item the first transition after total time $k\delta$ is a self-loop transition and occurs in the interval $[k\delta,k\delta + \delta/\sqrt{\lambda}]$;
\end{itemize}
The proof of the lemma is concluded by the following claim.
\begin{claim}
For $\lambda \to \infty$ we have 
\begin{align}
\probm^{\sigma, \pi_\lambda}_{\game_i}[X_\lambda]
& \; \to \;
1  \label{eq:nice-runs-to-one} \\                                                                                           
\probm^{\sigma, \pi_\lambda}_{\game_i}[\reach^{\leq T}\goal \mid X_\lambda]
& \; \to \;
\probm^{\sigma, \pi}_{\game_i}[\reach^{\leq T}\goal] \label{eq:nice-runs-are-nice}                                                                                 
\end{align}
\end{claim}
\begin{proof}
As regards (\ref{eq:nice-runs-to-one}), we deal with the conditions on runs in $X_\lambda$ one by one. First, notice that the Lebesgue measure of all the forbidden intervals tend to $0$ as $\lambda$ goes to infinity; hence, the probability of a Markovian transition occurring in any such interval tends to $0$. Second, we can underestimate the probability of $X_\lambda$ by considering only the waiting transitions of $\pi_\lambda$ as self-loops. The probability that the waiting transition occurs in each such interval can be bounded by 
$$
\left(1 - e^{\lambda \cdot \delta/\sqrt{\lambda}}\right)^{T/\delta}
\; = \;
\left(1 - e^{\sqrt{\lambda} \delta}\right)^{T/\delta}
\; \to \;
1
$$
since $T/\delta$ is constant and $e^{\sqrt{\lambda}\delta} \to 0$ as $\lambda \to \infty$.

As regards (\ref{eq:nice-runs-are-nice}), notice that the delay caused by the exponential simulation does not qualitatively change the behaviour. Namely, under the condition of $X_\lambda$,
\begin{itemize}
 \item any transition made by $\pi$ is simulated by $\pi_\lambda$ at most $\delta/\sqrt{\lambda}$ later;
%\todo{nemuze to tady prelezt?}
 \playerconc cannot interfere meanwhile because the states are either timed or immediate, never both;
%if there is an external transition enabled, there cannot be any internal transitions enabled by the Assumption; \todo{rozsirit bez Ass A}
 \item also no Markovian transition occurs meanwhile;
 \item the decision of the players after the delayed transition are the same as in the original play, since the first player plays the same in each whole interval $[k\delta,(k+1\delta)$ and the second player is asked what he would do if the $\lambda$-transition was precisely on time.
\end{itemize}
The change is only quantitative because we limit the Markovian transitions, but this change tends to zero as the probability of the set we condition by goes to one.
\qed
\end{proof}
\qed
\end{proof}

An exponential strategy in $\game$ can be simulated by an IMC environment of $\imc$:
\begin{lemma}\label{lem:exponential-to-imc}
For any scheduler $\sigma$ we have %and the corresponding strategy $\sigma'$
$$
\inf_{\substack{\lambda \in \Rsetp\\ \pi' \in \Pi_\lambda}} \probm^{\sigma, \pi'}_{\game_i}[\reach^{\leq T}\goal]
\quad \geq  \quad
\inf_{\substack{\env \in \ENV'\\\pi \in \scheduler_\sigma(\closedspeci{i})}}
\probm^{\pi}_{\closedspeci{i}}[\reach^{\leq T}\goal]
$$
\end{lemma}
\begin{proof}[Idea]
 Since only one rate is used, we can build a ``universal'' environment (w.r.t. this rate), that can freely select on which actions to synchronize and waiting with exactly this rate.
\end{proof}

\begin{proof}

We fix an arbitrary scheduler $\sigma$ and use the same strategy $\sigma$ as before %:= \sigma \circ p'$ on the left hand side as in Lemma~\ref{lem:cegame-without-self-loops}.
(observe that a scheduler has the same type as a strategy of $\playercon$).
Furthermore, we fix an arbitrary $\lambda \in \Rsetp$ and a $\lambda$-exponential strategy $\pi'$.
We choose $\env$ to be the environment of $ENV'$ depicted below for $\actnotau = \{\mathsf{a}\}$. It is very similar to $\commit$ from Section~\ref{sec:ceg}. The action alphabet of $\env$ is $\actnotau\cup\{\nowtran\}$, the state space is $2^{\actnotau} \cup \{\commitstate,\nowstate\}$and the transitions are for every $A\subseteq\actnotau$
\begin{itemize}
 \item $\commitstate\acttran{\tau}A$,
 \item $A\acttran{a}\commitstate$, for each $a\in A$,
 \item $A\acttran{\lambda}\nowstate$,
 \item $\nowstate\acttran{\tau}\commitstate$ and $\nowstate\acttranlong{\nowtran}\commitstate$.
\end{itemize}
Note that the only rate is $\lambda$. This is in some sense universal environement in $\ENV'$ for $\product$. Its power is only limited by $\lambda$ (for $\lambda\to\infty$ it can simulate any other environment).

\begin{center}
% \begin{tikzpicture}[/tikz/initial text=,font=\footnotesize,x=2cm,outer sep=1]
% \node[state,initial] (0) at (0,0) {$sync$};
% \node[state] (1) at (1,0) {$wait$};
% \node[state] (2) at (2,0) {$\nowstate$};
% %\node[state] (2)  at (2,-1) {$j_B$};
% \path[->] (0) edge[in=75,out=105,loop] node[above] {$\actnotau$} ()
%           (0) edge[] node[above] {$\tau$} (1)
% 	(1) edge[] node[above] {$\lambda$} (2)
% 	(2) edge[bend right=50] node[above] {$\nowtran$} (0)
% 	(2) edge[bend left=50] node[below] {$\tau$} (0)
% ;
% \end{tikzpicture}
\begin{tikzpicture}[/tikz/initial text=,font=\footnotesize,->,sloped,rounded corners,
state/.style={shape=circle,draw,inner sep=0mm,outer sep=0.8mm, minimum size=.7cm,}
]

\node[state,inner sep=4.25] (C) at (-0.3,-0.5) {$com.$};
\node[state,inner sep=4.25] (N) at (3.7,-0.5) {$now?$};
\node[state] (a) at (1.5,0) {$\{\mathsf{a}\}$};
\node[state] (e) at (1.5,-1) {$\emptyset$};

\path ($(C.west)+(-0.3,0)$) 		edge		(C);

\path (a) edge [bend right=30] node[above] {$\mathsf{a}$} (C);
\path (C) edge node[above] {$\tau$} (a);
\path (C) edge node[below] {$\tau$} (e);
\path (a) edge node[above] {$\lambda$} (N);
\path (e) edge node[below] {$\lambda$} (N);

% \draw (N) |- (1.5,.5) node[above,pos=0.3] {$\tau$} -| (C);
% \draw (N) |- (1.5,-1.5) node[right,pos=0.1] {$\nowtran$} -| (C);
\draw (N) |- (1.5,.5)  -| (C);
\draw (N) |- (1.5,-1.5)   -| (C);
\node at (3.9,0.3) {$\tau$};
\node at (4.2,-1.3) {$\nowtran$};
\end{tikzpicture}
\end{center}

We set $\pi$ to be scheduler that chooses the same transitions as the strategy $\pi'$. And when $\pi'$ decides to wait exponentially with $\lambda$ in a timed state, we are necessarily in some $A\subseteq\actnotau$ and thus automatically wait with $\lambda$ exponential waiting. This definition is correct as the paths of $\closedspeci{i}$ directly correspond to histories of $\game_i$.

% Formally, for a path $\mypath = (c_0,e_0) \, t_0 \, (c_1,e_1) \, t_1 \, \cdots t_{n-1} \, (c_n,e_n)$ where each $c_i$ is the state of the IMC component and each $e_i \in \{e_d, e_w\}$ is the state of $E$, we set
% \begin{itemize}
%  \item $\pi(\mypath) = (c_n,e_w)$ if $\pi'(\id(\mypath))$ chooses exponential waiting,
%  \item $\pi(\mypath) = (c_{n+1},e_d)$ if $\pi'(\id(\mypath))$ chooses $c_{n+1} \in \suco(c_n)$
% \end{itemize}
% where $\id: \paths(C(E_\lambda)) \to \histories(\game')$ is the first projection of the path (leaving out the states of the environment).% that alternates the original and duplicated states in any sequence of Markovian transitions of $\imc$.%, i.e. $\id(\mypath) = c_0 \, t_0 \, c_1 \, t_1 \cdots t_{n-1} \, c_n$. 
%By $\id'$ we denote the natural extension of $\id$ to infinite runs. 

Since for any measurable set of runs $X$ in  $\game_i$ we have
$\probm^{\sigma',\pi'}_{\game_i}[X] = \probm^{\sigma,\pi}_{\closedspeci{i}_\lambda}[\id^{-1}(X)]$ we also hav
$$
\probm^{\sigma', \pi'}_{\game'}[\reach^{\leq T}\goal]
\; = \; 
\probm^{\pi}_{\closedspeci{i}_\lambda}[\reach^{\leq T}\goal]
$$
% The key observation is that any path ending with a state of the form $(c,e_w)$ where the scheduler $\pi$ cannot do anything is mapped by $\id$ on a history where the $\lambda$-exponential strategy must wait. Furthermore, in all other situations the decisions of the schedulers and strategies coincide w.r.t. $\id$. Again, it is easy to see that 
\qed
\end{proof}

Finally, the proof of $v_{\game_i}\geq v_{\mathit{product}}(\product_i)$ follows easily from Lemmata~
% \ref{lem:grid-strategies-suffice}, 
\ref{lem:exponential-strategies-suffice}, and \ref{lem:exponential-to-imc} since we have
%\begin{align*}
\[
 \sup_{\sigma \in \Sigma} \inf_{\pi\in\Pi} \probm^{\sigma, \pi}_{\game_i}[\reach^{\leq T}\goal]
\; = \;
 \sup_{\sigma \in \Sigma} \inf_{\substack{\lambda\in\Rsetp\\\pi_\lambda \in \Pi_\lambda}} \probm^{\sigma, \pi_\lambda}_\game[\reach^{\leq T}\goal]
\; \geq \;
 \sup_{\sigma \in \Sigma} \inf_{\substack{\env \in \ENV'\\\pi \in \scheduler_\sigma(\closedspeci{i})}}
\probm^{\pi}_{\closedspeci{i}}[\reach^{\leq T}\goal] 
\]
%\end{align*}
\qed
\newpage

\subsection{Proof of Theorem~\ref{thm:ceg-direct}}

\begin{reftheorem}{thm:ceg-direct}
For every IMC $\imc$, MCA $\spec$, we have $\displaystyle v_{\game}=\lim_{i\to\infty}v_{\game_i}$. 
\end{reftheorem}

\noindent {\bf Proof idea}

%\begin{proof}[Proof Idea]
The states of the hyper-Erlang phase-type give approximate information about the remaining time in the current state of the specification. Moreover, this time is known shortly after entering the phase-type: in $\game_i$ after taking the first (fast) transition, in $\game$ after an arbitrarily short time chosen by $\playerenv$. Furthermore, the greater the hyper-Erlang, the more precise time estimation we have. In the limit, we thus know (from after the first transition till the sink) what the sampled remaining time exactly is. We can thus provide simulations back and forth.
\qed
%\end{proof}

\bigskip

\noindent {\bf Proof}

%\begin{proof}

``$\leq$'':

We need to simulate $\pi_i$ of $\game_i$ for a given fixed $i$. Here it is sufficient to:
\begin{itemize}
 \item upon entering a specifcation state wait with time distributed according to the first rate $2^{2^i}$ of the respective hyper-Erlang, and decide which branch $j$ we take in the simulated $\game_i$ (see below) when we get the sampled time $t$ in $\game$;
 \item simulate $k$th Markovian transition of the hyper-Erlang branch $j$: here we simply randomly choose time for this Markovian transition and check whether it happens before or after the proposed waiting time and perform the earlier (and possibly finish the waitingh later). The time for the Markovian transition is chosen according to the hyper-Erlang rate $\sqrt i$ under the condition that we are in the current branch $j$ at the $k$th node and we should get to sink in time $t-t'$ where $t'$ is the time spent in the current state $q$ of $\spec$.
 \item All other choices are the same as $\pi_i$ does in the respective (straightforwardly defined) history of $\game_i$.
\end{itemize}
It remains to show how to pick which branch to simulate, i.e. choose $j$. Firstly, there is a distribution on which length to choose under the condition that we should reach sink in precisely $t$, denote its cdf by $Branch$. We consider an arbitrary fixed mapping $Indep:[0,1]\to[0,1]$ where the argument is independent of the result. Denoting $F_{\sqrt i}$ the cdf of exponential distribution with rate $\sqrt i$, we $Branch^{-1}(Indep(F_{\sqrt i}(t)))$. This way, we use the random waiting (that can be seen in the history) as a random generator for the choice of the branch and thus we keep this choice implicitely in the history of the game.

\bigskip

``$\geq$'':

We simulate the behaviour of $\pi$ in $\game$ who has precise information about the time progress of the specification by a $\pi_i$ in $\game_i$ who only knows his position in the respective hyper-Erlang, so that for $i\to\infty$ the error approaches $0$. The main idea is that when we reach in the $\spec_i$ component a state of the form $(j,1)$ we guess how long we have before $\spec$ changes its state and then behave according to what $\pi$ would do with this time.

When in $(q,j,1)$ the cdf to reach sink is say $E_j$ (with the mean $j/\sqrt{i}$). Further, let $F_{2^{2^i}}$ be the cdf of Exp($2^{2^i}$). For time $t$ (which it took to take the $2^{2^i}$ transition) and $(j,1)$ (the reached target of this transition), we define $$Time_i(j,t):=E_j^{-1} (Indep(F_{2^{2^i}}(t)))$$
i.e.~we use the random quantile of the transition duration to get the random quantile for the time left in the current specification location. 
%The dependence of these two values causes no problems as the actual valu of $t$ will be irrelevant for the behaviour of the whole system as disucssed below

We now define a mapping $GetTimes:\histories(\game_i)\to\histories(\game)$. For a history
$$
(c_0,(q_{prev},x,y),e_0)t_1 (c_1,(q,1,0), e_1)t_2(c_2,(q,j,1),e_2)\cdots t_{n}(c_{n},(q,j,z),e_{n})t (c,(q_{next},1,0),e)  \history
$$
with $q_{prev}\neq q\neq q_{next}$ we have $c_2=c_1$, denote $t'=Time_i(j,t_2-t_1)$, and define the value of $GetTimes$ as follows:
\begin{itemize}
 \item remove all transitions corresponding to the moves of the specification while its state is still $q$,
 \item replace every $(q,j,k)$ (as well as $(q,1,0)$) by $(q,t')$,
 \item decrease all times from $t_2$ onwards (now also in $\history$!) by $t_2-t_1$,
 \item if the sink is reached (say at $t_k$) then replace $t_k$ by $t'$, 
%and decrease all times from $t_{k+1}$ onwards (also in $\history$!) by $t_k-t'$,
 \item we process $(c_{n},(q,j,z),e_{n})t (c,(q_{next},1,0),e) t \history$ the same way. If the end of $\history$ ends in the middle of a hyper-Erlang branch, the fourth point does not apply.
\end{itemize}
This way, we pretend the transition from the initial state of the hyper-Erlang took no time and we guessed the correct time to the sink. Observe that for $i\to\infty$, both errors approach zero. 

Let now $\sigma$ be any scheduler (thus a strategy in both  $\game_i$ and $\game$) and $\pi$ a strategy of $\playerenv$ in $\game$. We now define $\pi_i$. For a history $\history$ ending at time $t$, we have a history $\bar\history:=GetTimes(\history)$ ending at time $\bar t$. In immediate states, we set $\pi_i(\history):=\pi(\bar\history)$. In timed states, (assuming the $\pi$ is deterministic, see Claim in the proof of the previous theorem) we set
\begin{itemize}
 \item $\pi_i(\history):=\pi(\bar\history)$ if $\bar t < t$,
 \item $\pi_i(\history):=\pi(\bar\history)-(\bar t-t)$ if $\bar t > t$ and the result is positive,
 \item $\pi_i(\history):=\pi(\bar\history)/2^{i+k}$ otherwise, where $k$ is the length of the current history. (Intuitively, when a hyper-Erlang branch finishes later than it should have according to the guess, we slow down our waiting so that $\pi$ catches up.)
\end{itemize}

We now define a sequence of sets $(X_\ell)_{\ell\in\Nset}$ such that for every $\ell$
$$\lim_{i\to\infty}\probm^{\sigma, \pi_i}_{\game_i}\big[X_\ell\big]=1 $$ and
$$\lim_{\ell\to\infty}\lim_{i\to\infty}\probm^{\sigma, \pi_i}_{\game_i}\big[\reach^{\leq T}\goal\mid X_\ell\big]\leq 
\lim_{\ell\to\infty} \probm^{\sigma, \pi}_{\game}\big[\reach^{\leq T}\goal\big]$$

Recall the clock resolution $\delta$. The set $X_\ell$ is defined as the set of runs where:
\begin{itemize}
 \item no Markovian transitions occur at times in $[k\delta-\delta/\ell,k\delta+\delta/\ell]$ for any $k\in\{1,\ldots,T/\delta\}$
 \item the sum of durations of all transitions from the initial states of hyper-Erlangs before time $T$ does not exceed $\delta/\ell^2$, and 
 \item for each pair of $t',t_k$ from above it holds $|t'-t_k|<\delta/\ell$.
% \item the difference between the time of the last transition before $T$ in $\game_i$ and the corresponding transition in the simulated $\game$ with times of $GetTime$ is less than $\delta/\ell$
\end{itemize}

The first equation clearly holds by the weak law of large numbers and the fact that hyper-Erlangs approximate any continuous distributions.

The second equation then follows because:
\begin{itemize}
 \item At all moments the last time of $\history$ is in the same $\delta$-slot $[k\delta,(k+1\delta))$ as the last time of $GetTimes(\history)$ since we always keep these two aligned, except when $\history$ is ahead by $x$ and $\pi$ chooses to wait for less than $x$. But then we slow down our progress (see the third line of the definition of $\pi_i$). Further, under these conditions in total we wait for less than $\sum_{k=1}^\infty{\delta/2^{i+k}}=\delta/2^i$, which is smaller than $\delta/\ell$ for sufficiently large $i$. Moreover, for sufficiently large $i$, it is  even smaller than $\delta/\ell-\delta/\ell^2$. Hence by waiting for $\pi$ to catch up, we cannot be pushed out ofthe same slot as $GetTime(\history)$ is in, not even because of the inital transitions in the hyper-Erlangs.
 \item Therefore, $\sigma$ in $\game_i$ plays as in the simulated $\game$ since it makes the same decision throughout each whole $\delta$-slot and for each history $\history$ ending in time $t$, $GetTime(\history)$ ends in time $t'$ which is in the same slot.
 \item Thus for each $\ell$, we finish in the same slot as $\pi$, hence at the same time and state.
\end{itemize}
The second equation then concludes also the proof of this direction of the theorem.
\qed
%\end{proof}

%%% Local Variables:
%%% mode: latex
%%% TeX-master: "main"
%%% End:

%% file: app-disc2.tex
\newpage
\allowdisplaybreaks

\section{Definition of the CE game $\game$}\label{app:ce-precise}

The CE game $\game$ is defined on the game arena $G$ obtained from the game arena $G_1$. First, we need to alter $G_1$ a bit.
To simplify the argumentation, we assume that in $\game_1$, each immediate state has an internal transition that $\playercon$ can choose. If there is none, we add an internal transition to any goal state that each strategy $\sigma$ has to choose with probability $1$. This does not change the value $v_\game$ as the strategy $\pi$ can always reject such choice. The internal transitions that $\playercon$ can choose will be denoted by $v \tauc v'$, internal transitions that $\playerenv$ can choose will be denoted by $v \taue v'$. For each pair of states $v \probtranlong{1/2} v'$ where the Markovian transition corresponds to the flow of time in the specification component with constraint $\bowtie d$, we remove this Markovian transition, and write $v \probtran{d} v'$ instead. For each state $v$ for which there is no $v'$ and $d$ such that $v \probtran{d} v'$, we write $v \probtran{d} v$ for some distribution $d$ from the specification.

% First, let $F = \{v \in G_1 \mid \exists v',d: v \probtran{d} v' \}$ denote the states with flow of time in the specification component. Then 
The game arena is $G = (G_1\times (\Rsetp)^2)$. The first real number in a state is the time to wait in the current state of the specification as sketched in the main body of the text. The second real number is artificial, included for later proofs. We set $\histories(\game) = (G \times \Rsetpo)^* \times G$, similarly to the definition of $\game_i$. We define a sigma-field $\pathsfield$ over $\histories(\game)$ to be the naturally induced product sigma-field $\pathsfield$ where for each discrete component we use the sigma-field induced by the discrete topology and for each real component we use the Borel sigma-field.

A strategy of player $\playercon$ is a measurable
% \footnote{The strategy $\sigma$ is measurable with respect to $\pathsfield$ and $2^G$, the strategy $\pi$ is measurable w.r.t. $\pathsfield$ and $2^G \cup \mathbb{B}$ where $\mathbb{B}$ is the Borel sigma-field on $\Rsetp$.}
function $\sigma: \paths(\game) \to \dist(G)$ and a strategy of player $\playerenv$ is a measurable function $\pi: \paths(\game) \to \dist(\{\checkmark\} \cup G) \cup \dist(\Rsetp)$, where $\mathbb{B}$ denotes the Borel sigma-field over $\Rsetp$. For any history $\history$, we require that $\sigma(\history)$ and $\pi(\history)$ support only finitely many states -- those that can be reached by internal transitions where the real component remains intact.

For a given pair of strategies $\sigma,\pi$ we define the semantics of the CE game as a discrete-time Markov chain over the measurable space $(\paths(\game),\pathsfield)$. The transition kernel $\kernel$ of this chain, where $P(\history,A)$ denotes the probability to move in one step from the history $\history$ to any history in the set $A$, is defined as follows. Let us fix a history $\history = (v_0,r_0,u_0) \, t_1 \, (v_1,r_1,u_1) \, t_2 \, (v_2,r_2,u_2) \, \cdots \, t_{n} \, (v_n,r_n,u_n)$ where each $v_i \in G_1$ and $r_i,u_i \in \Rsetpo$. Further, we fix a measurable set $A$ of histories.

\begin{itemize}
 \item If $v_n$ is an immediate state, let $d = \sigma(\history)$ and $e = \pi(\history)$. We have
\begin{align*}
\kernel(\history,A) &= 
\sum_{v\in \supp(e)} e(v) \cdot [\history \, t_n \, (v,r_n,u_n) \in A]  \\
& \quad + 
e(\checkmark) \cdot \sum_{v_i \tauc v_{n+1}} d(v_{n+1}) \cdot [\history \, t_n \, (v_{n+1},r_n,u_n) \in A]
\intertext{where $[condition]$ is the indicator function of the condition $condition$. Observe that the first line above corresponds to player $\playerenv$ rejecting the choice of $\playercon$ and choosing his own state $v$, the second line corresponds to accepting the choice.
% If $v = \checkmark$ (player $\playerenv$ accepts the choice of $\playercon$), we have
}
% \kernel(\history,A) &= 
% \sum_{v_i \tauc v_{n+1}} d(v_{n+1}) \cdot [\history \, t_n \, (v_{n+1},r_n) \in A]
\intertext{
\item If $v_n$ is a first timed state visited after a new specification state is entered, the waiting time for the specification is generated as follows.}
\kernel(\history,A) &= \int_0^\infty \int_0^\infty ud(r,u) \cdot [\history \, (t_n) \, (v_n,r,u) \in A] \de{r} \de{u} \\
\intertext{
where $ud$ is the density of the uniform distribution over the area $\{(r,x) \mid r \in \Rsetpo, 0 < x < f(r) \}$ below the curve of $f$ where $f$ is the density of the distribution $d$.
\item For other timed states $v_n$, let $F = \pi(\history)$. We set}
\kernel(\history,A) &= \int_{t_e \in \Rsetpo} F(\de{t_e}) \cdot \left(E(t_e)\cdot[t_e < r_n] + S(r_n) \cdot [r_n \leq t_e] + \sum_{v_i \probtran{\lambda} v} M_{v}(\min\{t_e,r_n\})\right)
\intertext{where the terms $E(t_e)$, $S(r_n)$, and $M_v(t)$ describe the impact of the $\committran$ transition at time $t_e$, the flow transition in the specification at time $r_n$, and the Markovian transition to $v$ up to time $t$, respectively.}
E(t) &= e^{-\mu t} \cdot [\history \, (t_n+t) \, (v',r_n-t,u_n) \in A] \\
S(t) &= e^{-\mu t} \cdot [\history \, (t_n+t) \, (v'',0,u_n) \in A] \\
M_v(t) &= \frac{\lambda}{\mu} \cdot \int_0^t \mu \cdot e^{-\mu x} \cdot [\history \, (t_n+x) \, (v,r_n-x,u_n) \in A] \de{x}
\intertext{where $v_n \probtran{\lambda} v$, $\mu = \sum_{v_n \probtran{\lambda} v_{n+1}} \lambda$, $v_n \acttranlong{\committran} v'$, and $v_n \probtran{d} v''$.}
\end{align*}
\end{itemize}

\newpage

\section{Proof of Theorem~\ref{thm:disc}}\label{app:disc}

\begin{reftheorem}{thm:disc}
\theoremerrorbound
\end{reftheorem}
\begin{proof}
The proof is performed in several steps:
\begin{enumerate}
 \item An approximate game $\igame$ is defined where at most one Markovian transition occurs in each interval $[\ell\step,(\ell+1)\step)$ for $\ell\in\Nseto$. Furthermore new waiting times for timed transitions in the specification are randomly generated only at times $\ell\step$ for $\ell\in\Nseto$. This game approximates the game $\game$ by the bounds above. No other approximation error is involved in the further steps.
 \item A discrete step game $\dgame$ is defined, which is very similar to $\igame$, where every $\kappa$ time units an artificial self-loop is introduced and the set of actions of player $\playerenv$ is slightly extended; it is shown to have the same value as $\igame$.
 \item Thanks to the extended set of actions, a class of \emph{grid} strategies, which have finite representation, are shown to suffice in $\dgame$.
 \item Thanks to the artificial self-loops in $\dgame$ and the grid strategies, a discrete stochastic game played on a tree $\gamed$ is obtained directly from $\dgame$. These games have equal value.
 \item The discrete game $\gamed$ is shown to be solved in time polynomial in its (exponential) size.
\end{enumerate}

 Formally, these steps are proved in Lemmata~\ref{lem:game-igame}, \ref{lem:igame-dgame}, \ref{lem:dgame-gamed}, and \ref{lem:gamed-solution}. \QED
\end{proof}

\subsection{The approximate game $\igame$}

For a fixed $\step > 0$, we define game $\igame$ over the same state space of histories of $\game$ with the same set of strategies. The transition kernel $\kernel'$ of $\igame$ agrees with $\kernel$ of $\game$ on immediate states, as regards times states there are a few differences. Grid of intervals of length $\step$ plays a crucial rule in the semantics.
\begin{itemize}
 \item In the grid slot where the specification state changes, no Markovian transition occurs. Furthermore, when the specification changes state at time $l\step + x$, a new random number is generated as before, but the remaining time till the end of the current interval $\kappa - x$ is added to the newly generated number. Notice that this simulates the situation where the new random number is actually generated at the end of the interval.
 \item Within one interval of the grid, at most one Markovian transition occurs.
\end{itemize}

Formally, let us again fix a history $\history = (v_0,r_0,u_0) \, t_1 \, \cdots \, t_{n} \, (v_n,r_n,u_n)$ where each $v_i \in G_1$, $r_i,u_i \in \Rsetpo$, and $v_n$ is a timed state. We also fix a measurable set $A$ of histories. Further, let $a$ be the minimal number such that $t_n + a = \ell \cdot \step$ for some $\ell \in \Nset$, i.e. the remaining time till a grid line.  Let $b$ be the maximal number such that $b < r_n$ $b - a = \ell\cdot \step$ for some $\ell \in \Nset$, i.e. the remaining time till the last grid line before the specification changes its state. Finally, we set $c = a$ if there was a Markovian transition in the last $\step -a$ time; and we set $c = 0$, otherwise. For the first timed state visited after a new specification state is entered, we have
\begin{align*}
\kernel'(\history,A) &= \int_0^\infty \int_0^\infty ud(r',u') \cdot [\history \, (t_n) \, (v_n,(r'+b+\step-r_n),u') \in A] \de{r'} \de{u'} \\
\intertext{whereas for other timed states, we have}
 \kernel'(\history,A) &= \int_{t_e \in \Rsetpo} F(\de{t_e}) \cdot \left(E(t_e)\cdot[t_e < r_n] + S(r_n,b) \cdot [r_n \leq t_e] + \sum_{v_i \probtran{\lambda} v} M_{v}(\min\{t_e,b\})\right) \\
E(t) &= e^{-\mu t} \cdot [\history \, (t_n+t) \, (v',r_n-t,u_n) \in A] \\
S(t,b) &= e^{-\mu b} \cdot [\history \, (t_n+t) \, (v'',0,u_n)) \in A] \\
M_v(t) &= \frac{\lambda}{\mu} \cdot \int_c^t \mu \cdot e^{-\mu x} \cdot [\history \, (t_n+x) \, (v,r_n-x,u_n) \in A] \de{x}
\end{align*}

\begin{lemma}\label{lem:game-igame} Denoting by $v_\igame$ the value of the game $\igame$, we have
 $$v_\igame - 10 \step (\maxconst T)^2 \ln\frac{1}{\step} \;\; \leq \;\; v_\game \;\; \leq \;\; v_\igame + 10 \step (\maxconst T)^2 \ln\frac{1}{\step}$$
\end{lemma}
and a strategy $\sigma$ guaranteeing reachability probability $v$ in $\igame$, guarantees in $\game$ reachability probability in the interval $[v-10 \step (\maxconst T)^2 \ln\frac{1}{\step},v+10 \step (\maxconst T)^2 \ln\frac{1}{\step}]$.
\begin{proof}

For a fixed continuous density function $f$, we first define the set of \emph{simulable} points in $\game$ and $\igame$ denoted $R_f,R'_f \subseteq \{(r,x) \mid r \in \Rsetpo, 0 < x < f(r)\}$, respectively. A point $(r,x)$ is \emph{simulable in $\game$} if $r \geq \step$ and $x \leq f(r')$ for any $r' \in [r-\step,r]$.
A point $(r',x')$ is \emph{simulable in $\igame$} if $x' \leq f(r)$ for $r \in [r',r'+\step]$.

We define a set of runs $X$ and $Y$ such that whenever up to time $T$ a timed transition in the specification is taken, the newly randomly generated pair $(r,x)$ is simulable in $\game$ and simulable in $\igame$, respectively (by generating in $\igame$ we mean the number that is randomly picked, not the shifted number that is actually stored in the state space).
Next, we define a set of runs $Z$ such that at most one Markovian transition occurs in each interval $[\ell\step,(\ell+1)\step]$ for $0 \leq \ell < T/\step$.

Now we show that for any strategy $\sigma$ and $\pi$ it holds
 \begin{align}
 \prob_\igame^{\sigma,\pi}[\reach^{\leq T} G \cap X] 
 &\leq \prob_\game^{\sigma,\pi}[\reach^{\leq T} G \mid Z] \label{eq:cond-prob1} \\
 \prob_\igame^{\sigma,\pi}[\reach^{\leq T} G] 
 &\geq \prob_\game^{\sigma,\pi}[\reach^{\leq T} G \cap Y \mid Z] \label{eq:cond-prob2}
 \end{align}
Conditioning by $Z$ only equalizes the behavior of the Markovian transitions as in $\igame$ it is set by definition. The inequalities are obtained by the following idea:
Whenever in $\igame$ a number $z$ ($=b+\step-r_n$) is added to the randomly generated waiting time $r$, the behavior is the same is when in $\game$ the number $r+z$ is randomly generated. For each run in $X$ holds that the same run is also a run in $\igame$ such that these (identical) runs are equivalent w.r.t. the time-bounded reachability. Furthermore, this mapping preserves measure due to the simulability of the choices in the specification. To each simulable choice $(r,x)$ in $\game$, and any shift $z \in [-\step,0]$, there is enough marginal density to generate $(r-z,\cdot)$ in $\igame$. The same arguments hold vice versa for $\igame$.

 From (\ref{eq:cond-prob1}), we get
\begin{align}
\prob_\igame^{\sigma,\pi}[\reach^{\leq T} G \cap X] 
 &\leq \prob_\game^{\sigma,\pi}[\reach^{\leq T} G \cap Z] / \prob_\game^{\sigma,\pi}[Z] \notag \\ 
 \prob_\igame^{\sigma,\pi}[\reach^{\leq T} G] - \prob_\igame^{\sigma,\pi}[\reach^{\leq T} G \cap \lnot X] 
 &\leq \prob_\game^{\sigma,\pi}[\reach^{\leq T} G \cap Z] / \prob_\game^{\sigma,\pi}[Z] \notag \\
\prob_\igame^{\sigma,\pi}[\reach^{\leq T} G] 
 &\leq \prob_\game^{\sigma,\pi}[\reach^{\leq T} G] / \prob_\game^{\sigma,\pi}[Z] +\prob_\igame^{\sigma,\pi}[\lnot X] \notag \\
\prob_\igame^{\sigma,\pi}[\reach^{\leq T} G] 
 &\leq \prob_\game^{\sigma,\pi}[\reach^{\leq T} G] + 2(1- \prob_\game^{\sigma,\pi}[Z]) + (1- \prob_\igame^{\sigma,\pi}[X]) \label{eq:error-left}
\intertext{where the last manipulation holds for $\prob_\game^{\sigma,\pi}[Z] > 1/2$. Similarly from (\ref{eq:cond-prob2}), we get}
\prob_\game^{\sigma,\pi}[\reach^{\leq T} G \cap Y \cap Z] / \prob_\game^{\sigma,\pi}[Z] 
 &\leq \prob_\igame^{\sigma,\pi}[\reach^{\leq T} G]  \notag \\
\prob_\game^{\sigma,\pi}[\reach^{\leq T} G] 
 &\leq \prob_\igame^{\sigma,\pi}[\reach^{\leq T} G] \cdot \prob_\game^{\sigma,\pi}[Z] + 
 \prob_\game^{\sigma,\pi}[\reach^{\leq T} G \cap (\lnot Y \cup \lnot Z)] \notag \\
\prob_\game^{\sigma,\pi}[\reach^{\leq T} G] 
 &\leq \prob_\igame^{\sigma,\pi}[\reach^{\leq T} G] + \prob_\game^{\sigma,\pi}[\lnot Y \cup \lnot Z] \notag \\
\prob_\game^{\sigma,\pi}[\reach^{\leq T} G] 
 &\leq \prob_\igame^{\sigma,\pi}[\reach^{\leq T} G] + 
(1-\prob_\game^{\sigma,\pi}[Y]) + (1-\prob_\game^{\sigma,\pi}[Z]) \label{eq:error-right}
 \end{align}

Finally, we need to bound $(1-\prob_\game^{\sigma,\pi}[X])$, $(1-\prob_\game^{\sigma,\pi}[Y])$, and $(1-\prob_\game^{\sigma,\pi}[Z])$. As regards $\prob_\game^{\sigma,\pi}[Z]$, notice that there are $T/\step$ intervals of length $\kappa$. Due to the memoryless property of the exponential distribution, we can bound the probability by summing $T/\step$ times the probability $p$ that in one interval there are two or more Markovian transitions.
\begin{align}
 (1-\prob_\game^{\sigma,\pi}[Z]) \notag
&\leq \frac{T}{\step} \cdot p
\intertext{As the worst case we assume rate $\maxconst$ which bounds the maximal rate of $\imc$. The probability $p$ can be bounded by}
&\leq \frac{T}{\step} \cdot \frac{(\maxconst\step)^2}{2} 
\leq \frac{1}{2}\kappa\maxconst^2 T \label{eq:error-markov}
\end{align}
which follows from the properties of the Poisson distribution with parameter $\maxconst\step$ using the very same arguments as in~\cite[Lemma 6.2]{Neuhauser:PhD}. As regards $\prob_\game^{\sigma,\pi}[X]$ and $\prob_\game^{\sigma,\pi}[Y]$ we first bound how many times a timed transition in the specification can be taken. We assume the fastest possible transitions, i.e. uniformly distributed on $[0,1/\maxconst]$ (recall that $\maxconst$ bounds the maximal density of transitions in $\spec$). Let us express the probability $q$ that no more than $K = 2\maxconst T (1+\ln\frac{1}{\step})$ transitions occur within time $T$, i.e. that the sum of the random times of the first $K$ transitions exceed $T$. As the expected value of this sum is $T (1+\ln\frac{1}{\step})$, we can use the Hoeffding's inequality to bound the probability that the sum is not lower than its expected value by more than $T \ln\frac{1}{\step})$
$$
q 
\; \leq \;
\exp \left(-\frac{2 \left(T \ln\frac{1}{\step}\right)^2}{2\maxconst T (1+\ln\frac{1}{\step}) \cdot (1/\maxconst)^2}\right)
\; = \;
\exp\left(-\frac{T\maxconst \ln\frac{1}{\step}}{2}\right) 
\; = \; 
\kappa^{T\maxconst/2}
\; \leq \; 
\kappa
$$
due to the assumption that $b > 2/T$. Notice that the first manipulation holds for $\step \leq 1/3$ which we can easily assume. Further, observe that the probabilities $r,r'$ that in one transition a point is sampled that is not simulable in $\game$ and $\igame$, respectively, is bounded by $r \leq \step \cdot \maxconst + T \step \maxconst \leq 2 T \step \maxconst$ and $r' \leq T \step \maxconst$ as $\maxconst$ bounds the maximal density as well as the maximal derivation of the densities in $\spec$ (we further assume that $T \geq 1$. Hence,
\begin{align}
(1-\prob_\game^{\sigma,\pi}[X]) \notag
& \leq 
\kappa + K (2 T \step \maxconst) 
\; = \;
\kappa + \left(2\maxconst T \left(1+\ln\frac{1}{\step}\right)\right) (2 T \step \maxconst) \\
& \leq \;
\kappa + 8 \step (\maxconst T)^2 \ln\frac{1}{\step} 
\; \leq \; 
9 \step (\maxconst T)^2 \ln\frac{1}{\step}, \label{eq:error-spec}
\end{align}
and the same bound holds for $(1-\prob_\game^{\sigma,\pi}[Y])$, as well. All in all, from (\ref{eq:error-left}), (\ref{eq:error-right}), (\ref{eq:error-markov}), and (\ref{eq:error-spec}) we obtain the lemma.
\QED
\end{proof}

\subsection{The discrete-step game $\dgame$}

The goal is to obtain a game very close to the discretized game $\gamed$. Conceptually, $\dgame$ differs only a little from the game $\igame$. There are two differences:
\begin{itemize}
 \item every $\step$ time, there is a self-loop transition which materializes the grid introduced in $\igame$, we call these self-loops \emph{artificial ticks};
 \item in a timed state at time $\ell\step + x$ for $\ell\in \Nseto$ and $x\in[0,\step)$, the player $\playerenv$ has two additional actions: $0$, and $\almost{(\step-x)}$ which means playing almost $\step-x$, i.e. almost the time that remains until the grid line. The reason for these two actions is that optimizing the behavior in $\igame$ may force the player to take an as small number as possible, or a number as close to the grid line from left as possible. The set of strategies in $\dgame$, denoted by $\dPi$, thus simplifies the notion of optimality to be transfered to the discrete game $\gamed$.
\end{itemize}

\noindent
Histories in $\dgame$ are $\histories(\dgame) = (G \times (\Rsetpo \cup \{\almost{(\ell\step)} \mid \ell\in\Nseto \}))^* \times G$, where $\almost{x}$ denotes almost time $x$. Algebraically, $\almost{x} = x$, the only difference is that player $\playercon$ takes at time $\almost{(\ell\step)}$ decision as in the interval $[(\ell-1)\step,\ell\step)$. To this end, $\lfloor \history \rfloor$ has all $\almost{(\ell\step)}$ replaced by $(\ell-1)\step$.

Let us fix a history $\history = (v_0,r_0,u_0) \, t_1 \, \cdots \, t_{n} \, (v_n,r_n,u_n)$ where each $v_i \in G_1$, $r_i \in \Rsetpo \cup \{\almost{(\ell\step)} \mid \ell\in\Nseto \}$, and $u_i \in \Rsetpo$. We fix a measurable set $A$ of histories. Further, let $a = 0$ if $t_n = \, \almost{(\ell\step)}$ for some $\ell \in \Nset$. Otherwise, let $a$ be the minimal number such that $t_n + a = \ell \step$ for some $\ell \in \Nset$, i.e. the remaining time till a grid line and $b$  be as in the definition of $\igame$.

\begin{itemize}
 \item If $v_n$ is an immediate state, let $d = \sigma(\lfloor \history \rfloor)$ and $e = \pi(\history)$. We have
\begin{align*}
\kernel''(\history,A) &= 
\sum_{v\in \supp(e)} e(v) \cdot [\history \, t_n \, (v,r_n,u_n) \in A]  \\
& \quad + 
e(\checkmark) \cdot \sum_{v_i \tauc v_{n+1}} d(v_{n+1}) \cdot [\history \, t_n \, (v_{n+1},r_n,u_n) \in A]
% \kernel(\history,A) &= 
% \sum_{v_i \tauc v_{n+1}} d(v_{n+1}) \cdot [\history \, t_n \, (v_{n+1},r_n) \in A]
\intertext{
\item If $v_n$ is a timed state, let $F = \pi(\history)$. We distinguish four situations. 
For the first timed state visited after a new specification state is entered, we have}
\kernel''(\history,A) &= \int_0^\infty \int_0^\infty ud(r',u') \cdot [\history \, (t_n) \, (v_n,(r'+b+\step-r_n),u') \in A] \de{r'} \de{u'} \\
\intertext{where $ud$ is the density of the uniform distribution over the area $\{(r,x) \mid 0 < x < f(r) \}$ below the curve of $f$ where $f$ is the density of the distribution $d$. If a Markovian transition occured in $\history$ since the last artificial tick, we have}
\kernel''(\history,A) &= \int_{t_e \in \mathbb{R}'} F(\de{t_e}) \cdot \left(E'(t_e)\cdot[t_e < a \,\lor\, t_e = \, \almost{a}] + T'(a)\cdot [t_e \geq a] \right),
\intertext{i.e. no Markovian transition can occur; similarly if $r_n < a$, we have}
\kernel''(\history,A) &= \int_{t_e \in \mathbb{R}'} F(\de{t_e}) \cdot \left(E'(t_e)\cdot[t_e < r_n] + S'(r_n)\cdot [t_e \geq r_n] \right), \\
\intertext{i.e. no Markovian transition and no artificial tick can occur as first; otherwise}
\kernel''(\history,A) &= \int_{t_e \in \mathbb{R}'} F(\de{t_e}) \cdot \Bigg(E(t_e)\cdot[t_e < a \,\lor\, t_e = \, \almost{a}] + T(a) \cdot [t_e \geq a] \\
& \qquad\qquad\qquad\qquad\qquad + \sum_{v_i \probtran{\lambda} v} M_{v}(\min\{t_e,a\})\Bigg)
\intertext{where $\mathbb{R}' = \Rsetpo \cup \{\almost{a}\}$, the term $T(a)$ describes the impact of the the artificial tick, the terms $E', S'$, and $T'$ describe the situation where there is no Markovian transition to compete with}
T'(a) &= [\history \, (t_n+a) \, (v_n,r_n-a,u_n) \in A], \\
E'(t) &= [\history \, (t_n+t) \, (v',r_n-t,u_n) \in A], \\
S'(t) &= [\history \, (t_n+t) \, (v'',0,u_n) \in A], \\
\intertext{and $T(a) = e^{-\mu a} \cdot T'(a)$ where $\mu = \sum_{v_n \probtran{\lambda} v_{n+1}} \lambda$, $v_n \acttran{\committran} v'$, and $v_n \probtran{d} v''$.}
\end{align*}
\end{itemize}

\begin{lemma}\label{lem:igame-dgame}
 The game $\igame$ has the same value as the game $\dgame$, i.e.
 $$
 \sup_{\sigma \in \Sigma} \inf_{\pi\in\Pi} \probm^{\sigma, \pi}_{\igame}\big[\reach^{\leq T}\goal\big] 
 = 
 \sup_{\sigma \in \Sigma} \inf_{\pi\in\dPi} \probm^{\sigma, \pi}_{\dgame}\big[\reach^{\leq T}\goal\big] 
 $$
and a strategy $\sigma$ guarantees both in $\igame$ and $\dgame$ the same value.
\end{lemma}
\begin{proof}
We fix a strategy $\sigma \in \Sigma$. As regards $\geq$, the only change are the artificial ticks. The strategy $\pi$ of $\playerenv$ simulates in history $\history$ in $\dgame$ what the strategy $\pi'$ of $\playerenv$ does in history $\history'$ in $\igame$, where $\history'$ is obtained from $\history$ by removing the artificial ticks. In immediate states, simply $\pi(\history) = \pi(\history')$. In timed states the distribution on time $\pi(\history)$ is obtained from $\pi(\history')$ by conditioning by the amount of time $(t-t')$ that has been already spent waiting where $t$ and $t'$ is the total time of $\history$ and $\history'$, respectively. This way, we get an obvious correspondence of runs that preserves measure, i.e. we obtain completely the same probability to reach the target in $\dgame$ as in $\igame$.

Ar regards $\leq$, for a strategy $\pi$ in $\dgame$ a sequence of strategies $(\pi_i)_{i\in\Nset}$ is defined that wait $\delta/(i\cdot 2^j)$ instead of $\almostzero$ and $x-\delta/(i\cdot 2^j)$ instead of $\almost{x}$ in the $j$-th step. After such an imprecise waiting, the strategy further simulates what $\pi$ would do if it waits precisely $\almostzero$ or $\almost{x}$. I.e. for each history $\history$ in $\igame$, the simulating strategy uses the decisions of $\pi(\history')$ where $\history'$ is obtained from $\history$ by inserting the artificial ticks and replacing the imprecise waiting by the precise waiting as chosen by $\pi$. Furthermore, observe that the waiting of $\pi$ gets interrupted by the artificial ticks, i.e. its plans beyond the closes tick are irrelevant. Each $\pi_i$ has to plan the waiting in advance, i.e. connects the waiting distributions of $\pi$ in the current moment, after one artificial tick, after two artificial ticks, etc., as follows. Let $(v,r)$ be the last state of $\history'
$ and let $t$ and $t'$ denote the total time of $\history$ and $\history'$ and $x,x' \in\Rsetp$ be the minimal numbers such that $t+x = \delta k$ and $t'+x' = \delta k'$  for some $k.k'\in\Nset$. The distribution $\pi_i(\history)$ is defined 
\begin{itemize}
 \item on $[(t'-t),(t'-t)+x']$ using $\pi(\history')$ on $[0,x']$ (if $(t'-t)$ is negative, the distribution set to negative numbers is concentrated on $\delta/(i\cdot 2^j)$ instead),
 \item on $[(t'-t)+x',(t'-t)+x'+\delta]$ using $\pi(\history'\,(t+x') \, (v,r-t-x'))$ on $[0,\delta]$ conditioned by the waiting step not being taken in the interval $[0,x']$,
 \item on $[(t'-t)+x'+\delta,(t'-t)+x'+2\delta]$ using $\pi(\history'\,(t+x') \, (v,r-t+x')\,(t+x'+\delta) \, (v,r-t-x'-\delta))$ on $[0,\delta]$ conditioned by the waiting step not being taken in the interval $[0,x'+\delta]$, etc.
\end{itemize}

The behavior of $\igame$ and $\dgame$ differs only if a Markovian transition occurs within the imprecision in waiting. Since the total sum of the imprecision in waiting on \emph{any} run is at most $\delta/i$, the measure of runs that differs tends to $0$ as $i \to \infty$.
\end{proof}

\subsection{Grid strategies in $\dgame$}

We say that two histories $\history$ and $\history'$ \emph{follow the same pattern}, denoted $\history \pattern \history'$ if $\history = (v_0,r_0,u_0) t_1 \ldots t_{n} (v_n,r_n,u_n)$, $\history' = (v'_0,r'_0,u'_0) t'_1 \ldots t'_{n} (v'_n,r'_n,u'_n)$ with $v_i = v'_i$ and $t_i$ and $t'_i$ as well as $r_i$ and $r'_i$ equal when rounded down to a multiple of $\step$ for all $1 \leq i \leq n$. 
% Note that for every $n \in \Nset$, the relation $\pattern$ induces a finite partition of histories of length $n$.

We say that a strategy $\pi$ of $\env$ is a \emph{grid} strategy, denoted $\pi \in \bar{\Pi}_\delta$, if
\begin{enumerate}
 \item the strategy is deterministic;
 \item for any history $\history$ of time $t$ ending in a timed state, $t + \pi(\history) = k \cdot \delta$ for $k \in \Nset$;
 \item for any histories $\history$ of time $t$ and $\history'$ of time $t'$ that follow the same pattern, we have $\pi(\history) = \pi(\history')$ if $\history$ ends in an immediate state, and
$t + \pi(\history) = t' + \pi(\history')$ if $\history$ ends in a timed state ($\pi$ plans the $\committran$ at the same absolute time).
\end{enumerate}

\begin{lemma}\label{lem:grid-strategies}
 In the game $\dgame$, grid strategies suffice for player $\playerenv$, i.e. 
 $$
 \sup_{\sigma \in \Sigma} \inf_{\pi\in\bar{\Pi}} \probm^{\sigma, \pi}_{\dgame}\big[\reach^{\leq T}\goal\big] 
 = 
 \sup_{\sigma \in \Sigma} \inf_{\pi\in\bar{\Pi}_\delta} \probm^{\sigma, \pi}_{\dgame}\big[\reach^{\leq T}\goal\big] 
 $$
\end{lemma}
\begin{proof}
First, observe that for each $\eps > 0$ there is $n \in \Nset$ such that 
 $$
 \sup_{\sigma \in \Sigma} \inf_{\pi\in\bar{\Pi}} \probm^{\sigma, \pi}_{\dgame}\big[\reach^{\leq T}_{\leq n}\goal\big] 
 \leq 
 \sup_{\sigma \in \Sigma} \inf_{\pi\in\bar{\Pi}} \probm^{\sigma, \pi}_{\dgame}\big[\reach^{\leq T}\goal\big] + \eps,
$$
where $\reach^{\leq T}_{\leq n}\goal$ is the set of runs that reach the target within $n$ discrete steps within time $T$. Indeed, the probability of $n$ Markovian transitions to occur within time $T$ tends to $0$ as $n \to \infty$, the number of internal transitions in $\imc$ that can be performed between two Markovian transitions is limited by the acyclicity assumption, and  the set of runs where $\pi$ makes infinitely many transitions within time $T$ is $0$.

Thanks to this fact, it is enough to show that for any $n$ the grid strategies suffice for the $n$-step time bounded reachability. Let us fix $\sigma$ and $n \in \Nset$. We show that there is a grid strategy $\pi^\ast$ which is optimal w.r.t. $n$-steps, i.e.
$$
\probm^{\sigma, \pi^\ast}_{\dgame}\big[\reach^{\leq T}_{\leq n}\goal\big] 
 =
 \inf_{\pi\in\bar{\Pi}} \probm^{\sigma, \pi}_{\dgame}\big[\reach^{\leq T}_{\leq n}\goal\big].
$$

We construct $\pi^\ast$ by induction on the number of steps already performed in a history $\history$, denoted $|\history|$. 
Along the construction, we also prove the following claim
\begin{claim} Let $0 \leq i \leq n$ and $\history$ be an $i$-step history.
\begin{enumerate}
 \item If $\history$ has total time $k\delta$ for some $k \in \Nset$,
$$\probm^{\sigma, \pi^\ast}_{\dgame,\history}\big[\reach^{\leq T}_{\leq n}\goal\big] 
 =
 \inf_{\pi\in\bar{\Pi}} \probm^{\sigma, \pi}_{\dgame,\history}\big[\reach^{\leq T}_{\leq n}\goal\big].$$
\item If a Markovian transition already occurred in $\history$ since the last artificial tick, $$\probm^{\sigma, \pi^\ast}_{\dgame,\history}\big[\reach^{\leq T}_{\leq n}\goal\big] = \inf_{\pi\in\bar{\Pi}} \inf_{\history' \pattern \history} \probm^{\sigma, \pi}_{\dgame,\history'}\big[\reach^{\leq T}_{\leq n}\goal\big].$$
\end{enumerate}
\end{claim}

\noindent
Here, $\probm^{\sigma, \pi}_{\dgame,\history}$ is the probability measure of the chain where $\history$ is the initial state. Note that the second point of the claim applies to two situations: in the current ``grid slot'', either a Markovian transition already occurred, or the specification state is about to change.

As the induction base, we take arbitrary grid strategy $\pi^\ast$. For $i = n$, the claim obviously holds. For the induction step, we assume that $\pi^\ast$ satisfies the claim for all $i > j$. We alter $\pi^\ast$ to be satisfy the claim for $i = j$ as well. First, let $\history$ be a $i$-step history of total time $k\delta$ for some $k \in \Nset$. Let $\history$ end in a state $(v,r)$.
\begin{itemize}
 \item If $v$ is immediate, we set $\pi(\history)$ to the action that minimizes
$$\min\{ 
    \sum_{v \tauc v'} \sigma(\history)(v') \cdot \probm^{\sigma, \pi^\ast}_{\dgame,\history\,t\,(v',r)}\big[\reach^{\leq T}_{\leq n}\goal\big],
    \min_{v \taue v'} \probm^{\sigma, \pi^\ast}_{\dgame,\history\,t\,(v',r)}\big[\reach^{\leq T}_{\leq n}\goal\big]
\},$$
where the first choice corresponds to the action \checkmark.
 \item If $v$ is timed and a Markovian transition can occur before the next artificial tick, we set $\pi(\history)$ to $\almostzero$, $\almost{\delta}$, or $\delta$ depending on which minimizes
$$
\min \{ 
 \probm^{\sigma, \pi^\ast}_{\dgame,\history\,t\,(v'',r)}\big[\reach^{\leq T}_{\leq n}\goal\big],
 \mathcal{A} + \probm^{\sigma, \pi^\ast}_{\dgame,\history\,\overline{t}\,(v'',r-\delta)}\big[\reach^{\leq T}_{\leq n}\goal\big],
 \mathcal{A} + \probm^{\sigma, \pi^\ast}_{\dgame,\history\,t\,(v,r-\delta)}\big[\reach^{\leq T}_{\leq n}\goal\big]
\}
$$
where $v \acttranlong{\committran}v''$ and $\mathcal{A} = \int_0^\delta \mu \cdot e^{-\mu x} \sum_{v \probtran{\lambda} v'} \frac{\lambda}{\mu} \cdot \probm^{\sigma, \pi^\ast}_{\dgame,\history\,t+x\,(v',r-x)}\big[\reach^{\leq T}_{\leq n}\goal\big] \de{x}$.
\item If $v$ is timed and the specification state is about to change before the next artificial tick, i.e. at time $b < \delta$, we set $\pi(\history)$ to $\almostzero$ or to $\delta$, depending on which minimizes
$$
\min\{ \;
  \probm^{\sigma, \pi^\ast}_{\dgame,\history\,t\,(v'',r)}\big[\reach^{\leq T}_{\leq n}\goal\big], \;
  \int_0^\infty f(x) \cdot \probm^{\sigma, \pi^\ast}_{\dgame,\history\,t+b\,(v''',x)}\big[\reach^{\leq T}_{\leq n}\goal\big] \de{x},
\;\}
$$
where $v'''$ is reached from $v$ when the specification changes state and $f$ is the density of waiting in the new specification state.
\end{itemize}
Second, let $\history$ be a $i$-step history of total time $t$ where a Markovian transition already occurred since the last artificial tick. Let $\history$ end in a state $(v,r)$ and let $a$ be minimal such that $t+a = k\delta$ for some $k\in\Nset$.
\begin{itemize}
 \item If $v$ is immediate, we again set $\pi(\history)$ to the action that minimizes
$$\min\{ 
    \sum_{v \tauc v'} \sigma(\history)(v') \cdot \probm^{\sigma, \pi^\ast}_{\dgame,\history\,t\,(v',r)}\big[\reach^{\leq T}_{\leq n}\goal\big],
    \min_{v \taue v'} \probm^{\sigma, \pi^\ast}_{\dgame,\history\,t\,(v',r)}\big[\reach^{\leq T}_{\leq n}\goal\big]
\}.$$
\item If $v$ is timed, we set $\pi(\history)$ to $\almost{a}$ or to $a$ depending on which minimizes
$$
\min \{ \;
 \probm^{\sigma, \pi^\ast}_{\dgame,\history\,t+\overline{a}\,(v'',r-a)}\big[\reach^{\leq T}_{\leq n}\goal\big], \;
 \probm^{\sigma, \pi^\ast}_{\dgame,\history\,t+a\,(v,r-a)}\big[\reach^{\leq T}_{\leq n}\goal\big]
\; \}
$$
\end{itemize}
For all remaining $i$-step histories we set the strategy so that it is a grid strategy.

As regards the second point of the claim, observe that all the reachability probabilities are from longer histories, i.e. they do not depend on exact timing of $\history$. Hence, also the choice and the reachability probability in the $i$-th step does not depend on the exact timing, it cannot be lower for any $\history' \pattern \history$. Observe that for an immediate state, no other choice (possibly mixing among the pure choices) can yield a lower reachability probability. For a timed state, the chain must move either to history $\history\,t+a\,(v,r-a)$ (no action taken until the artificial tick) or to a history of the form $\history\,t+\overline{x}\,(v'',r-x)$ for some $x \leq a$ ($\committran$ taken at time $x$). From the induction hypothesis, the reachability probability does not depend on $x$, i.e. restricting to the pure choice $\almost{a}$ does not hamper optimality and results in a grid strategy.

As regards the first point of the claim, we again fix an $i$-step history $\history$ of total time $k\delta$ for some $k \in \Nset$. Let $\history$ end in a state $(v,r)$.
\begin{itemize}
 \item If $v$ is immediate, again, no other choice (possibly mixing among the pure choices) can yield a lower reachability probability.
 \item If $v$ is timed and a Markovian transition can occur before the next artificial tick, the situation is much more complicated. First observe that when a Markovian transition occurs in state $v'$ at time $x$, the (optimal) probability to reach the target does not depend on $x$ due to the second point of the claim. Thus, we can denote it $p_{v'}$. Waiting in the interval $[a,b]$ in state $v'$ for a Markovian transition to occur contributes to the reachability probability with $\int_{a}^{b} \mu \cdot e^{-\mu x} \cdot p_{v'} \de{x} = p_{v'} (e^{-\mu a}- e^{-\mu b})$. Hence optimizing the decisions in $v$ boils down to spending the $\delta$ time in such a state $v^\ast$ where the contribution is maximal, i.e. where $p_{v^\ast}$ is maximal; then before or after the next artificial tick moving to the state where the contribution is maximal for the following interval of size $\delta$ (the strategy $\sigma$ may take different actions before and after the artificial tick, hence we need to consider both options). 
Indeed, there is no reason to hesitate with moving to such $v^\ast$ as to probability to move to such a state also does not depend on time $x$ when the move is taken. Precisely, taking $\committran$ in state $v'$ at time $x$ results in traversing a finite sequence of states in $0$ time, ending in some timed state where a Markovian transition is again awaited. Importantly, the (possibly random) decisions of $\playercon$ in this sequence do not depend on $x$. All in all, action $\almostzero$ is taken to change the current state, and actions $\almost{\delta}$ or $\delta$ are taken if the contribution of the current state is optimal and the state is to be changed before or after the next artificial tick.
\end{itemize}

This concludes the proof of the claim as well as the proof of the lemma.
\QED
\end{proof}

\subsection{Discrete game $\gamed$}

We will define the discrete game $\gamed$ as an extensive-form game~\cite{DBLP:conf/stoc/KollerMS94}.
As the game $\dgame$ with a grid strategy has an almost discrete structure not much work is left to define the discrete game. Observe that for a grid strategy $\pi$ the behaviour of the Markov chain of $\dgame$ in a history $\history$ does not depend on exact timing of $\history$, i.e. for all $\history' \pattern \history$ it holds $\probm^{\sigma,\pi}_{\dgame,\history} \left[ \reach^{\leq T} G \right] = \probm^{\sigma,\pi}_{\dgame,\history'} \left[ \reach^{\leq T} G \right]$. Hence, we can define the same game on the partition $V = \{[\history]_\pattern \mid \history \in X\}$ of the set of histories $X \subset \histories(\dgame)$ where the total time is $ \leq T$ and that, intuitively speaking, can possibly be ever played. Formally $X$ are the histories where
\begin{itemize}                                                                                                                                                                                                                                                                                                                                                                                                                                                                                                                                                                                                                               \item the total time is $\leq T$,
\item there is at most one Markovian transition in each grid interval,
\item the $\committran$ transitions are taken only at times $\ell\step$ or $\almost{(\ell\step)}$. \end{itemize}
Due to these restrictions, only a limited number of steps can be played up to time $T$, hence, $V$ is finite. For a vertex $v\in V$, we denote by $\last(v)$ the last state of all the histories in the class $v$.
To comply with the definition of extensive-form game, we divide the vertices where  $\playercon$ takes decisions from the vertices where $\playerenv$ takes decisions and from the stochastic vertices. Hence, we set 
$$
V' = 
    V \;\cup\; 
    V_i \times \{\checkmark\} \;\;\cup\;\;
    (V_\leftend \times \{0,\almost{\step},\step\} \;\cup\; 
    V_\inside \times \{\almost{\step},\step\} \;\cup\; 
    V_\rightend),
$$
where $V_i$ are the vertices corresponding to immediate states, $V_\leftend$, $V_\inside$, and $V_\rightend$ are the vertices corresponding to timed states with total time $\ell\step$, $\ell\step+x$ and $\almost{\ell\step}$, respectively, for some $\ell \in \Nseto$ and $x \in (0,\step)$. The vertices are divided among the players as follows.
\begin{itemize}
 \item $V_i \cup V_\leftend \cup V_\inside$ are the vertices of the first player (player $\playerenv$) with actions 
$\{v' \in V \mid \current(v) \taue \current(v')\} \cup \{\checkmark\}$ if $v\in V_i$,
$\{0,\step,\almost{\step}\}$ if $v\in V_\leftend$, and $\{\step,\almost{\step}\}$ if $v\in V_\inside$;
 \item $V_i \times \{\checkmark\}$ are the vertices of the second player (player $\playercon$) with actions $\{v' \in V \mid \last(v) \tauc \last(v')\}$ for any vertex $(v,\checkmark)$ of $\playercon$;
 \item $V_\leftend \times \{0,\almost{\step},\step\} \;\cup\; V_\inside \times \{\almost{\step},\step\} \;\cup\; V_\rightend$ are the stochastic vertices.
\end{itemize}

Timed vertices with total time $T$ are the terminal vertices, denoted $Z$. For all other remaining vertices, the tree-like transition structure is defined as follows.
In a vertex $v \in V$, any action $a \in \{\checkmark, 0, \almost{\step},\step\}$ leads to the vertex $(v,a)$. In a vertex $v \in V'$, action $v' \in V$ leads to the vertex $v'$. The probability matrix $\kernel^\gamed$ for the stochastic vertices is defined as follows. 
\begin{itemize}
 \item For $v \in V_\rightend$ with $\history\in v$ and $v' \in V'$, we set $\kernel^\delta(v,v') = \kernel''(\history,v')$ where $\kernel''$ is induced by arbitrary strategies $\sigma$ and $\pi$.
 \item For $(v,a) \in (V_\leftend \times \{0,\almost{\step},\step\}) \cup (V_\inside \times \{\almost{\step},\step\})$ with $\history\in v$ and $v' \in V'$, we set $\kernel^\delta((v,a),v') = \kernel''(\history,v')$ where $\kernel''$ is induced by any strategy $\sigma$ and a grid strategy that chooses in $\history$ action $a$.
\end{itemize}

Each terminal vertex $v\in Z$ has payoff $u(v)$ associated: vertex of histories that visit $G$ in the first component have payoff $1$, other vertices have payoff $0$. Recall that $\step = n\cdot \delta$. The observation sets for the first player are 
$$H_1 = \{ \{v \in V' \mid \last(v) = (c,s,e), \; \text{$t$ is the total time of $v$}, \lfloor t/\delta \rfloor = i\} \mid i \in \Nseto, c \in C \},$$
i.e. the vertices in one observation set agree on the first component of the last state and on the total time up to the precision of $\delta$. Observation sets for the second player are singletons, i.e.
$$H_2 = \{ \{v\} \mid v \in V'\},$$

A \emph{behavioural strategy}  of the first and the second player is a function that assigns to each observation of the player a probability distribution over actions available in this vertex. We denote these sets by $\Sigma_\gamed$ and $\Pi_\gamed$. Observe that the strategies of player $\playercon$ and grid strategies of player $\playerenv$ in the game $\dgame$ coincide with the strategies of the first and the second player in $\gamed$.
The value of the game is defined as the expected payoff: $$\sup_{\sigma\in\Sigma_\gamed}\inf_{\pi\in\Pi_\gamed} E^{\sigma,\pi}[u].$$ Notice that this definition is equivalent to the definition via reachability in the main body. From these observations and from the construction we immediately get that the values equal:
\begin{lemma}\label{lem:dgame-gamed}
 We have that $v_\dgame = v_\gamed$ and any strategy of player $1$ in $v_\gamed$ corresponds to a strategy of player $\playercon$ in $\dgame$ guaranteeing the same value.
\end{lemma}

\subsection{Solution of the discrete game $\gamed$}

\begin{lemma}\label{lem:gamed-solution}
 The value $v_\gamed$ and an optimal strategy can be computed in time polynomial in $|\gamed|$.
\end{lemma}

\begin{proof}
Because the observation sets of player $1$ form a tree such that each action of player $1$ results in a move in this tree, it is easy to see that $\gamed$ satisfies the condition of perfect recall (see~\cite{DBLP:conf/stoc/KollerMS94}). The value and the optimal strategies can be then computed using a linear program~\cite{DBLP:conf/stoc/KollerMS94} of size linear in the state space. 
\end{proof}